\documentclass[a4paper,10pt,english,onecolumn]{IEEEtran} 

\usepackage[english]{babel}
\usepackage[latin1]{inputenc}
\usepackage[T1]{fontenc}
\usepackage{pifont}		
\usepackage[active]{srcltx}
\usepackage{eurosym}
\usepackage{siunitx}
\DeclareSIUnit{\var}{var}
\DeclareSIUnit{\rad}{rad}
\DeclareSIUnit{\umdrehung}{U}
\usepackage{lscape}		
\usepackage{framed}		
\usepackage{longtable}
\usepackage{booktabs}	
\usepackage{fancyhdr}
\usepackage{setspace}
\usepackage{enumerate} 
\usepackage{multienum}
\usepackage[dvipsnames]{xcolor}
\usepackage{comment}
\usepackage{subdepth} 
\usepackage{trfsigns}	
\usepackage{mathrsfs} 
\usepackage{mathdots}	
\usepackage{cancel}
\usepackage{amsmath}
\usepackage{amssymb}
\usepackage{amsfonts}
\usepackage{stmaryrd}	
\usepackage{amsthm}
\usepackage{thmtools}
\usepackage[font=footnotesize,format=hang]{caption}
\usepackage{pdfpages}
\usepackage{float}
\usepackage{subfig} 
\usepackage{graphicx}
\usepackage{psfrag}
\usepackage[]{silence} %
\usepackage[process=auto,crop=pdfcrop,cleanup={.tex,.dvi,.ps,.pdf,.snm,.out,.bbl,.nav,.auxlock,.toc}]{pstool}
\usepackage{makeidx}
\usepackage{lastpage}
\usepackage{algorithmicx}
\usepackage{listings}
\definecolor{brown}{RGB}{153, 51,0} 
\definecolor{orange}{RGB}{255, 127.5,0}
\definecolor{darkgray}{RGB}{127.5, 127.5,127.5}
\definecolor{cyan}{RGB}{0,255,255} 
\definecolor{white}{RGB}{255,255,255} 
\definecolor{green}{RGB}{0,255,0} 
\definecolor{darkgreen}{RGB}{0,100,0} 
\definecolor{darkred}{RGB}{100,0,0} 

\definecolor{mycyan}{RGB}{0,255,255} 
\definecolor{mygray}{RGB}{204,204,204} 
\definecolor{TUMblue1}{RGB}{0,82,147}%
\definecolor{black}{RGB}{0,0,0}%
\definecolor{TUMcyan}{RGB}{100,160,200}%
\definecolor{TUMorange}{RGB}{227,114,34}%
\definecolor{TUMgrey}{RGB}{156,157,159}%
\definecolor{TUMgreen}{RGB}{162,173,0}%
\definecolor{matlabgreen}{rgb}{0, 0.5, 0}
\definecolor{matlabmagenta}{rgb}{0.75, 0, 0.75}

\definecolor{tumblue}{RGB}{0,101,189} 		
\definecolor{tumblack}{rgb}{0,0,0} 		
\definecolor{tumwhite}{rgb}{1,1,1} 		

\definecolor{tumdarkgray}{rgb}{0.2,0.2,0.2}	
\definecolor{tumgray}{rgb}{0.5,0.5,0.5}
\definecolor{tumlightgray}{rgb}{0.8,0.8,0.8}
\definecolor{tumdarkblue}{RGB}{0,51,89} 	
\definecolor{tumblue2}{RGB}{0,82,147}		
\definecolor{tumblue3}{RGB}{0,115,207}		
\definecolor{tumblue4}{RGB}{100,160,200}	
\definecolor{tumlightblue}{RGB}{152,198,234}  	

\definecolor{tumyellow}{RGB}{255,180,0}
\definecolor{tumorange}{RGB}{227,128,0}
\definecolor{tumred}{RGB}{229,52,24}
\definecolor{tumgreen}{RGB}{145,172,107}
\definecolor{tumlightgreen}{RGB}{181,202,130}

\definecolor{plot_color_my_slow}{rgb}{0.15,0.09,0.93}
\definecolor{plot_color_my_fast}{rgb}{0.28,0.54,0.29}
\definecolor{plot_color_sogi}{rgb}{0.43,0.93,0.97}
\definecolor{plot_color_anf}{rgb}{0,0,0}
\definecolor{plot_color_ref}{rgb}{0.94,0.17,0.08}

\newcommand{\grayline}{\protect\tikz{\protect\draw[very thick,gray] (0,-0.5ex)(0,0)--(3.3ex,0);}}

\newcommand{\reddashedline}{\tikzset{external/export next=false}\tikz{\protect\draw[very thick,dashed,red] (0,-0.5ex)(0,0)--(3.3ex,0);}}

\newcommand{\blueline}{\protect\tikz{\protect\draw[very thick,blue] (0,-0.5ex)(0,0)--(3.3ex,0);}}

\newcommand*{\bluedashdottedline}{\tikzset{external/export next=false}\tikz{\protect\draw[very thick,blue,dashdotted] (0,-0.5ex)(0,0)--(4ex,0);}}
\newcommand*{\bluedottedline}{\tikzset{external/export next=false}\tikz{\protect\draw[very thick,blue,dotted] (0,-0.5ex)(0,0)--(4ex,0);}}
\newcommand*{\bluedashedline}{\tikzset{external/export next=false}\tikz{\protect\draw[very thick,dashed,blue] (0,-0.5ex)(0,0)--(4ex,0);}}
\newcommand*{\cyanline}{\tikzset{external/export next=false}\tikz{\protect\draw[very thick,cyan] (0,-0.5ex)(0,0)--(4ex,0);}}
\newcommand*{\cyandashdottedline}{\tikzset{external/export next=false}\tikz{\protect\draw[very thick,dashdotted,cyan] (0,-0.5ex)(0,0)--(4ex,0);}}
\newcommand*{\cyandashedline}{\tikzset{external/export next=false}\tikz{\protect\draw[very thick,dashed,cyan] (0,-0.5ex)(0,0)--(4ex,0);}}
\newcommand*{\cyandottedline}{\tikzset{external/export next=false}\tikz{\protect\draw[very thick,dotted,cyan] (0,-0.5ex)(0,0)--(4ex,0);}}
						\newcommand{\mv}[1]{\boldsymbol{#1}}
						\newcommand{\mm}[1]{\boldsymbol{#1}}

						\newcommand{\ddtsmall}{\tfrac{\textrm{d}}{\textrm{d}t}}

						\newcommand{\N}{\mathbb{N}}

						\newcommand{\Q}{\mathbb{Q}}
						\newcommand{\R}{\mathbb{R}}
						\newcommand{\C}{\mathbb{C}}
						\newcommand{\Rzp}{\mathbb{R}_{\geq 0}}
						\newcommand{\Rzpos}{\mathbb{R}_{\geq 0}}
						
						\newcommand{\Rpos}{\mathbb{R}_{> 0}}

						\newcommand{\Cneg}{\mathbb{C}_{< 0}}

						\newcommand{\fset}[1]{\mathcal{#1}}			

						\newcommand{\norm}[1]{\left\lVert #1 \right\rVert}

\DeclareMathOperator{\diag}{diag}
\DeclareMathOperator{\blockdiag}{blockdiag}

\DeclareMathOperator{\rank}{rank}
\DeclareMathOperator{\sat}{sat}

\newcommand{\Linf}{\fset{L}^{\infty}}

\newcommand{\Czero}{\fset{C}}

\newcommand{\musteq}{\stackrel{!}{=}}

\newcommand*{\Iss}{\mathbb{I}_{\textrm{ss}}} 

\newcommand*{\esnorm}[1]{\|#1\|_{\infty}}

\newcommand{\dx}[1]{\textrm{d}#1}

\newcommand{\ve}[1]{\boldsymbol{#1}}

\newcommand{\br}[1]{\left(#1\right)}
\newcommand{\bs}[1]{\left[#1\right]}
\newcommand{\bc}[1]{\left\{#1\right\}}

\newcommand{\sine}[1]{\sin\!\br{#1}}
\newcommand{\cosine}[1]{\cos\!\br{#1}}

\newcommand{\artan}[1]{\arctan\!2\!\br{#1}}

\newcommand{\qqquad}{\qquad\quad}

\newcommand{\yhout}{\widehat{y}}
\newcommand{\yhp}[1]{\widehat{y}_{#1}}

\newcommand{\qhp}[1]{\widehat{q}_{#1}}

\newcommand{\JMAT}{\ve{J}}
\newcommand{\J}{\ve{J}}
\newcommand{\Jsub}[1]{\ve{J}_{#1}}
\newcommand{\Jbar}{\overline{\ve{J}}}

\newcommand{\AMAT}{\ve{A}}
\newcommand{\Asubt}[1]{\ve{A}_{#1}}

\newcommand{\Rsub}[1]{\ve{R}_{#1}}

\newcommand{\lvec}{\ve{l}}

\newcommand{\cyvec}{\ve{c}}

\newcommand{\xsogi}{\widehat{\ve{x}}}
\newcommand{\xsogip}[1]{\widehat{\ve{x}}_{#1}}

\newcommand{\omegah}{\widehat{\omega}}
\newcommand{\omegahp}[1]{\widehat{\omega}_{#1}}

\newcommand{\ey}{e_y}

\newcommand{\xsogissp}[1]{\xsogi_{#1,\infty}}

\newcommand{\eyssp}[1]{e_{#1,\infty}}

\newcommand{\kgainp}[1]{k_{#1}}
\newcommand{\ggainp}[1]{g_{#1}}

\newcommand{\Ytra}[1]{\mathcal{\widehat{Y}}_{#1}\!}
\newcommand{\Qtra}[1]{\mathcal{\widehat{Q}}_{#1}\!}
\newcommand{\Eytra}{\mathcal{E}_{\mathrm{y}}\!\!\:}

\newcommand{\AYtra}[1]{A_{\Ytra{#1}}}
\newcommand{\AQtra}[1]{A_{\Qtra{#1}}}
\newcommand{\AEytra}{A_{\Eytra}}

\newcommand{\PYtra}[1]{\Phi_{\Ytra{#1}}\!}
\newcommand{\PQtra}[1]{\Phi_{\Qtra{#1}}\!}
\newcommand{\PEytra}{\Phi_{\Eytra}\!}

\newcommand{\lamvec}{\ve{\lambda}}
\newcommand{\lamvecp}[1]{\ve{\lambda}_{#1}}
\newcommand{\lamyp}[1]{\lambda_{\mathrm{y},#1}}
\newcommand{\lamqp}[1]{\lambda_{\mathrm{q},#1}}

\newcommand{\SMAT}{\ve{S}}
\newcommand{\Ssub}[1]{\ve{S}_{#1}}

\newcommand{\cpol}{\ve{p}_{\AMAT}}

\newtheorem{theorem}{Theorem}[section]
\newtheorem{proposition}[theorem]{Proposition}
\newtheorem{corollary}[theorem]{Corollary}
\newtheorem{lemma}[theorem]{Lemma}
\newtheorem{remark}[theorem]{Remark}
\newtheorem*{remark*}{Remark}

\newtheorem*{fact}{Fact}

\newtheorem{example}[theorem]{Example}
\newtheorem*{example*}{Example}
\newtheorem{examples}[theorem]{Examples}
\newtheorem*{examples*}{Examples}

\newtheorem*{assumption*}{Assumption}
\newtheorem{definition}[theorem]{Definition}
\newtheorem*{definition*}{Definition}
\newtheorem*{definitions*}{Definitions}
\newtheorem{case}{Case}
\usepackage{tikz}
	\usetikzlibrary{fit}
	\usetikzlibrary{calc}
 	\pgfdeclarelayer{background}
 	\pgfsetlayers{background,main}

\usepackage{pgfplots}
\pgfplotsset{compat=newest}
\pgfplotsset{plot coordinates/math parser=false}
\newlength\figureheight
\newlength\figurewidth
\usepgfplotslibrary{fillbetween}
\usepackage{pgfplotstable}
\usepackage{pgfgantt}

\usetikzlibrary{arrows}
\usetikzlibrary{shapes}
\usetikzlibrary{decorations.pathmorphing}
\usetikzlibrary{decorations.pathreplacing}
\usetikzlibrary{decorations.shapes}
\usetikzlibrary{decorations.text}

\usepgfplotslibrary{external} 
\tikzsetexternalprefix{./arxiv/} 

\usepackage[european]{circuitikz} %
\usetikzlibrary{circuits.ee.IEC}

\usetikzlibrary{arrows}
\usetikzlibrary{calc}
\usetikzlibrary{intersections}
\usetikzlibrary{backgrounds}
\usetikzlibrary{shapes}
\usetikzlibrary{positioning}
\usetikzlibrary{trees}
\usetikzlibrary{datavisualization}
\usetikzlibrary{datavisualization.formats.functions} 

\tikzset{set voltage source graphic = voltage source IEC graphic}
\tikzset{voltage source IEC graphic/.style={circuit symbol lines, circuit symbol size = width 3 height 3, shape = generic circle IEC, fill = white,
    /pgf/generic circle IEC/before background={
				\draw[-] (0, -1pt) -- (0, 1pt);
    }
}}

\tikzstyle{sum} = [draw, circle, inner sep = 0pt, minimum width=0.2cm, fill = white]

\tikzexternaldisable

\allowdisplaybreaks

\begin{document}

\title{Modified second-order generalized integrators with modified frequency locked loop for fast harmonics estimation of distorted single-phase signals \newline (LONG VERSION)}

\author{C.M.~Hackl$^{\ddagger,\star}$ and M.~Landerer$^{\dagger}$}

\renewcommand{\thefootnote}{\fnsymbol{footnote}}

\footnotetext[3]{C.M.~Hackl is with the Munich University of Applied Sciences (MUAS) and head of the ``Control of Renewable Energy Systems (CRES)'' research group at Technical University of Munich (TUM), Germany (e-mail: christoph.hackl@hm.edu).}
\footnotetext[2]{M. Landerer is with the research group ``Control of renewable energy systems'' (CRES) at the Munich School of Engineering (MSE), Technical University of Munich (TUM), Germany (e-mail: markus.landerer@tum.de).}
\footnotetext[1]{Authors are in alphabetical order and contributed equally to the paper. Corresponding author is C.M. Hackl (christoph.hackl@hm.edu).}
\renewcommand{\thefootnote}{\arabic{footnote}}

\maketitle
\thispagestyle{empty}

\begin{abstract}
This paper proposes \emph{modified Second-Order Generalized Integrators} (mSOGIs) for a fast estimation of all harmonic components of arbitrarily distorted single-phase signals such as voltages or currents in power systems. The estimation is based on the internal model principle leading to an overall observer system consisting of parallelized mSOGIs. The observer is tuned by pole placement. For a constant fundamental frequency, the observer is capable of estimating all harmonic components with prescribed settling time by choosing the observer poles appropriately. For time-varying fundamental frequencies, the harmonic estimation is combined with a \emph{modified Frequency Locked Loop} (mFLL) with gain normalization, sign-correct anti-windup and rate limitation. The estimation performances of the proposed parallelized mSOGIs with and without mFLL are illustrated and validated by measurement results. The results are compared to standard approaches such as parallelized standard SOGIs (sSOGIs) and adaptive notch filters (ANFs).
\end{abstract}


\textbf{***A shorter version of this paper has been submitted to and accepted for publication in IEEE Transactions on Power Electronics (DOI:~10.1109/TPEL.2019.2932790; for more details see \cite{2019_Hackl_mSOGIswithmFLLforfastharmonicsestimationofdistortedsingle-phasesignals}) *** \\ This long version includes (i) more detailed explanations, (ii) more simulation and measurement results and (iii) a thorough theoretical analysis in its Appendix.}

\tableofcontents

\subsection*{Notation}
$\N, \R, \C, \Q$:~natural, real, complex and rational numbers. 
For the following, let $n,m \in\N$. 
$\mv{x} := (x_1, \dots, x_n)^\top\in \R^{n}$:~column vector (where $:=$ means ``is defined as'' and $^\top$ means ``transposed'').
$\mv{0}_n := (0,\,0,\,,\dots, 0)^\top \in \R^n$:~zero vector.
$\mv{1}_n := (1,\,1,\,,\dots, 1)^\top \in \R^n$:~vector of ones.
$\norm{\mv{x}} := {\small \sqrt{\mv{x}^{\top}\mv{x}}}$: Euclidean norm  of $\mv{x}$.
$\mm{A} \in \R^{n \times m}$:~real (non-square) matrix. 
$\diag(\mv{a}) \in \R^{n \times n}$:~diagonal matrix with diagonal entries taken from vector $\mv{a} = (a_1, \dots, a_n)^\top \in \R^n$.
$\mm{O}_{n \times m} \in \R^{n \times m}$:~zero (non-square) matrix. 
$\mm{I}_n := \diag(\mv{1}_n) \in \R^{n \times n}$:~identity matrix.
$\blockdiag(\mm{A}_1,\dots,\mm{A}_n) \in \R^{nm \times nm}$:~block diagonal matrix with matrix entries $\mm{A}_i \in \R^{m \times m}$, $i \in \{1,\dots,n\}$.
$\chi_{\mm{A}}(s):=\det[s\, \mm{I}_n - \mm{A}]$, characteristic polynomial of $\mm{A} \in \R^{n \times n}$. $\arctan\!2\big(x,y\big)$ 	 2-argument arctangent, for $x,y\in \R$, defined as\\
\begin{eqnarray} 
		\arctan\!2 \colon   \R^2 \setminus \{(0,0)\} & \to & (-\pi,\pi],   \quad
		(x,y)  \mapsto   \arctan\!2\big(x,y\big) := 
		\begin{cases}
		\arctan\!\big(\tfrac{y}{x}\big) &, \; x >0 \wedge y \in \R \\\\
		\arctan\!\big(\tfrac{y}{x}\big)+\pi &, \; x <0 \wedge y >0 \\
		\pm\pi & , \; x <0 \wedge y =0\\
		\arctan\!\big(\tfrac{y}{x}\big)-\pi &, \; x <0 \wedge y <0 \\
		+\tfrac{\pi}{2} & ,\; x = 0 \wedge y > 0\\ 
		-\tfrac{\pi}{2} & ,\; x = 0 \wedge y < 0.\\
		\end{cases}
		\label{eq:[N]definition of atan2}
\end{eqnarray}
%

\section{Introduction}

\subsection{Motivation and literature review}
In view of the increasing number of decentralized generation units with power electronics based grid connection and the decreasing number of large-scale generators, the overall inertia in the grid is diminishing. This results in a faster transient response and higher harmonic distortion of physical quantities (such as currents or voltages) of the power system~\cite{2018_Milano_FoundationsandChallengesofLow-InertiaSystemsInvitedPaper}. Fast frequency fluctuations endanger stability of the power grid. Significant harmonic distortions of voltages and currents can deteriorate power quality and lead to damage or even destruction of grid components. To be capable of taking appropriate countermeasures such as (i) improving stability and quality and (ii) compensating for such deteriorated operation conditions, it is crucial to detect and estimate fundamental and higher harmonic components of the considered quantities in real time as fast and accurate as possible. Hence, grid state estimation became of particular interest to the research community in the last years and has been studied extensively (see e.g.~\cite{2006_Rodriguez_AdvancedGridSynchronizationSystemforPowerConvertersunderUnbalancedandDistortedOperatingConditions,
	2007_Mojiri_Time-DomainSignalAnalysisUsingAdaptiveNotchFilter,
	2017_Chilipi_AdaptiveNotchFilter-BasedMultipurposeControlSchemeforGrid-InterfacedThree-PhaseFour-WireDGInverter,
	2010_Fedele_StructuralpropertiesoftheSOGIsystemforparametersestimationofabiasedsinusoid, 
	2011_Rodriguez_MultiresonantFLLsforGridSynchronizationofPowerConvertersUnderDistortedGridConditions,
	2011_Muzi_Areal-timeharmonicmonitoringaimedatimprovingsmartgridpowerquality,
	2011_Luo_FrequencymeasurementusingaFFL,
	2013_Park_Advancedsingle-phaseSOGI-FLLusingself-tuninggainbasedonfuzzylogic,
	2013_Kulkarni_AnoveldesignmethodforSOGI-PLLforminimumsettlingtimeandlowunitvectordistortion,
	2014_Panda_AnimprovedmethodoffrequencydetectionforgridsynchronizationofDGsystemsduringgridabnormalities,
	2015_Cossutta_HighspeedsinglephaseSOGI-PLLwithhighresolutionimplementationonanFPGA,
	2016_Xin_Re-InvestigationofGeneralizedIntegratorBasedFiltersFromaFirst-Order-SystemPerspective,
	2016_Patil_ModifieddualSOGIFLLforsynchronizationofadistributedgeneratortoaweakgrid,
	2016_Golestan_ARobustandFastSynchronizationTechniqueforAdverseGridConditions,
	2017_Matas_AFamilyofGradientDescentGridFrequencyEstimatorsfortheSOGIFilter,  
	2017_Ralev_AdoptingaSOGIfilterforflux-linkagebasedrotorpositionsensingofSRM,
	2017_Xiao_AFrequency-FixedSOGI-BasedPLLforSingle-PhaseGrid-ConnectedConverters, 2017_Yi_ImpedanceAnalysisofSOGI-FLL-BasedGridSynchronization,
	2017_Golestan_ACriticalExaminationofFrequency-FixedSOGI-BasedPLLs,
	2017_Golestan_Three-PhasePLLs:AReviewofRecentAdvances,
	2018_Dai_AFixed-LengthTransferDelay-basedAFLLforSingle-PhaseSystems,
	2018_Karkevandi_FrequencyestimationwithantiwinduptoimproveSOGIfiltertransientresponsetovoltagesags,
	2019_He_ReinvestigationofSingle-PhaseFLLs} to name a few).  \\

It is well known that a signal with significant harmonic distortion can be decomposed and analyzed by the Fast Fourier Transformation (FFT). However, this method requires a rather long computational time and a large amount of data to be processed~\cite[p.~320]{2000_Rade_SpringersMathematischeFormeln}. Usually, several multiples of the fundamental period ($\geq \SI{200}{\milli\second}$) are needed to estimate the harmonics with acceptable accuracy~\cite{2011_Muzi_Areal-timeharmonicmonitoringaimedatimprovingsmartgridpowerquality}; when the frequency is estimated online as well, the estimation time is even longer.\\

The majority of the publications deals only with the estimation of fundamental signal parameters (such as amplitude and phase) and fundamental frequency. For signals with negligible harmonic distortion, several well known and rather fast methods are available~\cite[Chapter~4]{2011_Teodorescu_GridConvertersforPhotovoltaicandWindPowerSystems} such as \emph{Second-Order Generalized Integrator (SOGI)} or \emph{Adaptive Notch Filters (ANF)} with and without \emph{Phase-Locked Loop (PLL)}~\cite{2017_Golestan_Three-PhasePLLs:AReviewofRecentAdvances} or \emph{Frequency Locked-Loop (FLL)}~\cite{2006_Rodriguez_AdvancedGridSynchronizationSystemforPowerConvertersunderUnbalancedandDistortedOperatingConditions,2011_Rodriguez_MultiresonantFLLsforGridSynchronizationofPowerConvertersUnderDistortedGridConditions}. However, if the signals to be estimated have significant harmonic distortion, these approaches fail and have to be extended by the parallelization of several SOGIs (see e.g.~\cite{2011_Rodriguez_MultiresonantFLLsforGridSynchronizationofPowerConvertersUnderDistortedGridConditions,
	2011_Muzi_Areal-timeharmonicmonitoringaimedatimprovingsmartgridpowerquality,
	2011_Luo_FrequencymeasurementusingaFFL,
	2016_Xin_Re-InvestigationofGeneralizedIntegratorBasedFiltersFromaFirst-Order-SystemPerspective,
	2019_He_ReinvestigationofSingle-PhaseFLLs}) or several ANFs  (see e.g.~\cite{2007_Mojiri_Time-DomainSignalAnalysisUsingAdaptiveNotchFilter,
2017_Chilipi_AdaptiveNotchFilter-BasedMultipurposeControlSchemeforGrid-InterfacedThree-PhaseFour-WireDGInverter}); each of those being capable of estimating the individual harmonics separately.  However, the resulting estimation system is highly nonlinear (in particular in combination with FLL or PLL) and difficult to tune. The estimation speed is usually faster than those of FFT approaches but still rather slow. Other estimation approaches use Adaptive Linear Kalman Filters~\cite{2016_Reza_AccurateEstimationofSingle-PhaseGridVoltageFundamentalAmplitudeandFrequencybyUsingaFrequencyALKF}, SOGIs in combination with discrete Fourier transforms~\cite{2014_Reza_AccurateEstimationofSingle-PhaseGridVoltageParametersUnderDistortedConditions} or circular limit cycle oscillators~\cite{2018_Ahmed_FFLBasedEstimationofSingle-PhaseGridVoltageParameters}. A comparison of estimation speed and estimation accuracy mainly focuses on fundamental signal and frequency estimation. A comparison of all the results presented in the contributions above yields that the estimation speeds vary between $40-\SI{1200}{\milli\second}$. The estimation speed depends on the tuning of the estimation algorithms and the operation conditions (such as changing harmonics with varying amplitudes, phases and frequencies) during the estimation process. In particular, when the frequency is changing abruptly, the overall estimation process is drastically decelerated. The FLL can be considered as the bottleneck of grid state estimation. Moreover, the performance of the parallelized estimation of the individual harmonics is mostly not discussed and evaluated. \\

Exceptions are the contributions~\cite{2011_Rodriguez_MultiresonantFLLsforGridSynchronizationofPowerConvertersUnderDistortedGridConditions,2011_Muzi_Areal-timeharmonicmonitoringaimedatimprovingsmartgridpowerquality} and \cite{2007_Mojiri_Time-DomainSignalAnalysisUsingAdaptiveNotchFilter}; which explicitly discuss and show the estimation performance of the parallelized SOGIs and ANFs, respectively, for \emph{each} considered harmonic component. For example, in~\cite{2011_Rodriguez_MultiresonantFLLsforGridSynchronizationofPowerConvertersUnderDistortedGridConditions}, the proposed parallelized SOGIs with FLL (called MSOGI-FLL) are capable of extracting fundamental frequency and amplitudes and phases of a pre-specified number $n$ of harmonics $\nu \in \{\nu_1,\dots, \nu_n\}$. Local stability analysis and tuning of the parallelized SOGIs and FLL were thoroughly discussed. As outcomes of the tuning rules, the gain $k_\nu$ of the $\nu$-th SOGI should be chosen to be $k_\nu = \tfrac{1}{\nu} \sqrt{2}< \tfrac{1}{\nu} 2$ which represents a ``tradeoff between settle time, overshooting and harmonic rejection''. Simulation and measurement results were presented for three-phase signals. Six harmonics (including fundamental positive sequence) and the fundamental frequency were correctly estimated. The estimation speeds for the individual harmonics vary between $40-\SI{140}{\milli\second}$. Frequency estimation takes about $\SI{300}{\milli\second}$ to return to a constant value. In~\cite{2011_Muzi_Areal-timeharmonicmonitoringaimedatimprovingsmartgridpowerquality}, a similar idea is proposed. The proposed algorithm is also based on parallelized SOGIs but a FLL has not been implemented. If the frequency is known, the method is capable of estimating the harmonics in approximately $40-\SI{60}{\milli\second}$\footnote{Note that, the authors state that the estimation takes less than $\SI{20}{\milli\second}$, which seems not correct as can be observed in Fig.~6 and Fig.~8 in~\cite{2011_Muzi_Areal-timeharmonicmonitoringaimedatimprovingsmartgridpowerquality}.}. Only simulation results were presented for seven harmonics. No results were presented when the frequency is unknown and varying. Implementation and tuning of the parallelized SOGIs are hardly discussed. In~\cite{2007_Mojiri_Time-DomainSignalAnalysisUsingAdaptiveNotchFilter}, parallelized ANFs with FLL are implemented. For a constant (estimated) frequency, a complete stability proof for the parallelized structure  is presented showing that stability is preserved if all gains are chosen positive. The parallelized ANFs with FLL are implemented in Matlab/Simulink to estimate a signal with six harmonic components (including fundamental). The fundamental frequency of the considered signal undergoes step-like changes of $+\SI{4}{\hertz}$ and $-\SI{2}{\hertz}$. Frequency and harmonics estimation errors tend to zero; but the estimation speed is rather slow and varies between $1-\SI {1.5}{\second}$.\\

As already noted, the (parallelized) SOGIs and ANFs rely on a precise estimate of the fundamental (angular) frequency for proper functionality. If the frequency is known a priori, it can be fed directly to the parallelized systems. Otherwise, the observers must be combined with a~FLL (or PLL), which allows to additionally estimate the fundamental angular frequency online. Since the FLL estimation depends on the harmonic amplitudes of the input signal, \cite{2011_Rodriguez_MultiresonantFLLsforGridSynchronizationofPowerConvertersUnderDistortedGridConditions, 2017_Matas_AFamilyofGradientDescentGridFrequencyEstimatorsfortheSOGIFilter, 2016_Patil_ModifieddualSOGIFLLforsynchronizationofadistributedgeneratortoaweakgrid} describe a \emph{Gain Normalization} (GN) which robustifies the frequency estimation. Nevertheless, due to its nonlinear and time-varying dynamics, the tuning of the overall estimator consisting of parallelized SOGIs or ANFs and FLL is a non-trivial task. As a rule of thumb (coming from the steady-state derivation of the FLL adaption law), the tuning of the FLL should be slow compared to the dynamics of the parallelized SOGIs or ANFs and, therefore, significantly degrades the settling time of the overall estimation system~\cite{2013_Park_Advancedsingle-phaseSOGI-FLLusingself-tuninggainbasedonfuzzylogic}. Apart from that, negative estimates of the angular frequency lead to instability. In this context, \cite{2014_Panda_AnimprovedmethodoffrequencydetectionforgridsynchronizationofDGsystemsduringgridabnormalities} describes a saturation of the estimated angular frequency to avoid a sign change. However, this saturation does not necessarily (i) ensure convergence of the estimation error or (ii) accelerate the transient response of the FLL. In~\cite{2018_Karkevandi_FrequencyestimationwithantiwinduptoimproveSOGIfiltertransientresponsetovoltagesags}, the FLL is extended by output saturation and anti-windup to avoid too large estimation values. But the proposed anti-windup strategy comes with additional feedback gain (tuning parameter), which, if not properly chosen, might lead to instability. Other approaches for frequency detection are based on \emph{Phase Locked Loops (PLLs)}~\cite{2017_Xiao_AFrequency-FixedSOGI-BasedPLLforSingle-PhaseGrid-ConnectedConverters}, \cite{2013_Kulkarni_AnoveldesignmethodforSOGI-PLLforminimumsettlingtimeandlowunitvectordistortion, 2015_Cossutta_HighspeedsinglephaseSOGI-PLLwithhighresolutionimplementationonanFPGA,2017_Golestan_ACriticalExaminationofFrequency-FixedSOGI-BasedPLLs,2017_Golestan_Three-PhasePLLs:AReviewofRecentAdvances} which can be combined with SOGIs as well. PLL approaches are not considered in this paper.\\

The remainder of this paper focuses on modifications of the parallelized ``standard SOGIs'' and the ``standard FLL'' as introduced in~\cite{2006_Rodriguez_AdvancedGridSynchronizationSystemforPowerConvertersunderUnbalancedandDistortedOperatingConditions} and~\cite{2011_Rodriguez_MultiresonantFLLsforGridSynchronizationofPowerConvertersUnderDistortedGridConditions} which will allow to improve estimation speed and estimation accuracy significantly. Key observation which motivates the modifications is that almost all papers above, except~\cite{2016_Xin_Re-InvestigationofGeneralizedIntegratorBasedFiltersFromaFirst-Order-SystemPerspective}, do only consider one single tuning factor (gain) for individual SOGI design. This single tuning factor limits the possible estimation performance. In~\cite{2016_Xin_Re-InvestigationofGeneralizedIntegratorBasedFiltersFromaFirst-Order-SystemPerspective}, two gains are considered but their influence on the speed of harmonics estimation is not exploited and investigated. Therefore, this work proposes a \emph{modified} (generalized) algorithm which achieves a prescribed settling time of the estimation process. It is capable of estimating amplitudes, angles and angular frequencies of all harmonic components of interest in real time. The proposed algorithm consists of parallelized \emph{modified} SOGIs tuned by pole placement. The modified SOGIs come with \emph{additional} feedback gains (additional tuning parameters) which provide the required degrees of freedom to ensure a desired (prescribed) settling time. Since the standard FLL was derived and is working for the standard SOGI only (as shown in~\cite{2011_Rodriguez_MultiresonantFLLsforGridSynchronizationofPowerConvertersUnderDistortedGridConditions} or~\cite{2011_Luo_FrequencymeasurementusingaFFL}), also a \emph{modified} FLL is proposed to guarantee functionality in combination with the parallelized modified SOGIs. The novelty of this paper is characterized by the following five main contributions:
\begin{enumerate}[(i)]
 \item \emph{Modification (generalization)} of standard SOGIs (sSOGIs) to \emph{modified SOGIs (mSOGIs)} with \emph{prescribed settling time} (see Sections~\ref{sec:modified-sogi-msogi-for-the-nu-th-harmonic-component} and~\ref{sec:Estimation performance of sSOGI and mSOGI}); 
 \item \emph{Parallelization} of the mSOGIs and their \emph{tuning by pole placement} (see Section~\ref{sec:Parallelization of SOGIs});
 \item \emph{Modification (generalization)} of the standard FLL to the \emph{modified FLL (mFLL)} with phase-correct adaption law, sign-correct anti-windup strategy and rate limitation for enhanced functionality in combination with the proposed mSOGIs (see Section~\ref{sec:modified-fll-mfll});
 \item \emph{Theoretical derivation} of the pole placement algorithm and the generalized adaption law for the mFLL (see Appendix~\ref{sec_app_pole_placement} and~\ref{sec:generalization-of-the-adaption-law-of-the-mfll-for-the-parallelized-msogis},respectively);
 \item \emph{Implementation and validation} of the proposed estimation algorithm by simulation and measurement results and \emph{Comparison} of the estimation performances of parallelized mSOGIs, sSOGIs and (ANFs) \emph{with} and \emph{without} FLL (see Section~\ref{sec:implementation-and-measurement-results}).
\end{enumerate}
%

\subsection{Problem statement}
\label{sec:Problem statement}
Single-phase grid signals (e.g.~voltages or currents) with significant and arbitrary harmonic distortion are considered. The considered signals are assumed to have the following form
\begin{equation}
\forall \, t \geq 0 \colon \quad  y(t) := \sum\limits_{\nu \in \mathbb{H}_n} \underbrace{a_{\nu}(t)\cos\big(\phi_{\nu}(t)\big)}_{=: y_\nu(t)}  \quad \text{ where } \quad \mathbb{H}_n:= \{1,\nu_2, \dots,\nu_n\} \subset \Q_{> 0},
\label{eq:input signal y}
\end{equation}
with fundamental amplitude $a_1$, harmonic amplitudes $a_{\nu_2}, \dots, a_{\nu_n} \geq 0$ and angles $\phi_{\nu}$ (in \si{\radian}), respectively; where $\nu \in \mathbb{H}_n$ indicates the $\nu$-th harmonic component (per definition $\nu_1:=1$). Observe that $\nu$ does not necessarily need to be a natural number or larger than one; any rational number is admissible as well (e.g.~$\nu_2 = 5/3$ or $\nu_3 = 1/5$).  Moreover, to consider the most general case, the phase angles
$$
\forall \nu \in \mathbb{H}_n\; \forall t \geq 0\colon \qquad \phi_{\nu}(t) = \int_0^t \nu\, \omega\br{\tau} \dx{\tau} + \phi_{\nu,0},
$$
of the $\nu$-th harmonic component depend on the time-varying angular \emph{fundamental} frequency $\omega(\cdot):=\omega_1(\cdot) > \SI{0}{\radian\per\second}$ and the initial harmonic angle $\phi_{\nu,0} \in \R$. The main goal of this paper is threefold:
\begin{enumerate}[(i)]
 \item to propose a \emph{modified Second-Order Generalized Integrator (mSOGI)} with \emph{prescribed settling time} for a \emph{fast online estimation} of amplitudes $\widehat{a}_{\nu}$ and angles $\widehat{\phi}_{\nu}$, such that, after a \emph{specified} transient phase, estimated signal $\widehat{y}$ (indicated by ``~$\widehat{~}$~'') and original signal $y$ do not differ more than a given threshold $\varepsilon_y > 0$. More precisely, the following should hold
 \begin{equation}
 \forall t \geq t_{\mathrm{set}} \colon \quad | y(t) - \yhout(t) | \leq \varepsilon_y \quad \text{ where } \quad \widehat{y}(t):= \sum\limits_{\nu \in \mathbb{H}_n} \underbrace{\widehat{a}_{\nu}(t)\cos\big(\widehat{\phi}_{\nu}(t)\big)}_{=: \yhp{\nu}(t)} 
 \label{eq:estimated input signal yhat}
 \end{equation}
after a \emph{prescribed (specified) settling time} $t_{\mathrm{set}} > \SI{0}{\second}$; 
\item to propose a \emph{modified Frequency Locked Loop (mFLL)} ensuring stable operation and fast estimation of the angular frequency in combination with the proposed parallelized mSOGIs; and
\item to show the overall estimation performance and compare it to other standard approaches such as parallelized sSOGIs and ANFs with and without FLL.
\end{enumerate}

\begin{remark}
 Note that in~\eqref{eq:input signal y}, time-varying amplitudes (of each harmonic component) \emph{and} time-varying angles are considered. The typical assumption~(see, e.g.~\cite[Appendix~A]{2011_Teodorescu_GridConvertersforPhotovoltaicandWindPowerSystems}) of a \emph{constant} fundamental angular frequency $\omega > 0$ such that $\phi_{\nu}(t) = \nu\omega t$ is \emph{not} imposed, since it is not true in general.
\end{remark}

\subsection{Principle idea of proposed solution}
\label{sec:Principle idea of proposed overall solution}
The principle idea of the proposed solution is illustrated in Fig.~\ref{fig:whole_model}. The depicted block diagram is fed by the input signal~$y$ to be estimated. All subsystems of the overall nonlinear observer are shown. The outputs of the block diagram are the respective estimated components of the input signal (see Sect.~\ref{sec:Problem statement}). 
\begin{figure}[ht!] 
	\centering
	\includegraphics{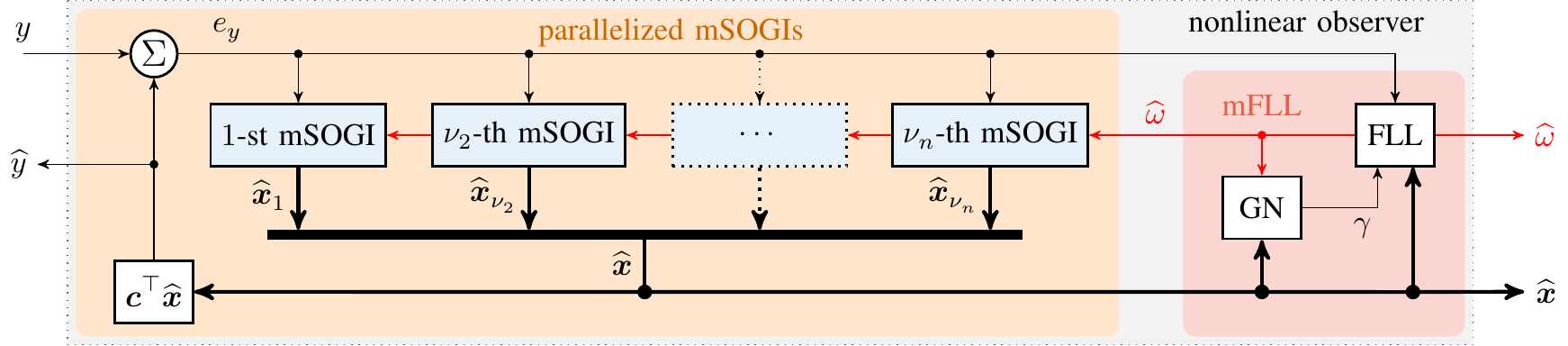}
	\caption{Block diagram of the nonlinear observer (consisting of parallelized mSOGIs and modified FLL (mFLL) with gain normalization (GN)).}
	\label{fig:whole_model}
\end{figure}
In Fig.~\ref{fig:whole_model}, all components (subsystems) of the nonlinear observer are explicitly shown. One can summarize:~For $\nu \in \mathbb{H}_n$, the overall nonlinear observer consists of the following subsystems:
\begin{itemize}
 \item a \emph{parallelization} of \emph{modified Second-Order Generalized Integrators (mSOGIs)} to estimate amplitude and phase of each of the harmonic components of the input signal $y$. The $\nu$-th mSOGI will output the estimated state vector $$\widehat{\mv{x}}_{\nu} := \big( \widehat{x}_{\nu}^{\alpha}, \, \widehat{x}_{\nu}^{\beta} \big)^{\top} = \big( \widehat{y}_{\nu}, \, \widehat{q}_{\nu} \big)^{\top}$$ compromising estimates of in-phase and quadrature signals of the $\nu$-th harmonic component, i.e.~$\widehat{y}_{\nu}=\widehat{x}_{\nu}^{\alpha}$ and $\widehat{q}_{\nu}=\widehat{x}_{\nu}^{\beta}$, respectively. All $n$ estimated signal vectors $\widehat{\mv{x}}_{\nu}$ are merged into the overall estimation vector 
 \begin{equation}
  \xsogi :=  \big( \underbrace{\br{\yhp{1},\qhp{1}}}_{=: \xsogip{1}^{\top}}, \underbrace{\br{\yhp{\nu_2},\qhp{\nu_2}}}_{=: \xsogip{\nu_2}^\top}, \dots, \underbrace{\br{\yhp{\nu_n}, \qhp{\nu_n}}}_{=: \xsogip{\nu_n}^\top}\big)^\top \in \R^{2n}.
 \label{eq:estimated output vector after parallelized SOGIs}
 \end{equation} 
 The output signal $\yhout = \sum_{\nu \in \mathbb{H}_n} \yhp{\nu} = \cyvec^\top\xsogi$ represents the estimate of the input signal $y$ and is established by the sum of all estimates of the in-phase signals $\widehat{y}_\nu=\widehat{x}_\nu$ of the mSOGIs;
 \item a \emph{modified Frequency Locked Loop (mFLL)} with gain normalization, generalized frequency adaption law, sign-correct anti-windup strategy and rate limitation to obtain the estimate $\omegah$ of the fundamental angular frequency $\omega$. The mFLL is tuned by an adaptive gain $\gamma$ which depends on estimation input error $e_y := y - \widehat{y}$, estimation vector $\widehat{\mv{x}}$ and estimated angular frequency $\widehat{\omega}$;  
\end{itemize}
Section~\ref{sec:SOGIs} and Section~\ref{sec:Frequency-locked loop} introduce the different subsystems (i.e.~mSOGIs and mFLL) illustrated in Fig.~\ref{fig:whole_model} and explain in more detail their contribution to the proposed solution for real-time estimation of amplitudes and phases of all $n$ harmonics of the input signal $y$ as in~\eqref{eq:input signal y} as well as the fundamental frequency $\omega$.


\section{Second-Order Generalized Integrators (SOGIs):~In-phase and quadrature signal estimation}
\label{sec:SOGIs}

The key tool to estimate in-phase and quadrature signals of a measured sinusoidal signal is a \emph{Second-Order Generalized Integrator (SOGI)}~\cite[App.~A]{2011_Teodorescu_GridConvertersforPhotovoltaicandWindPowerSystems}. Their parallelization in combination with a FLL (see Sect.~\ref{sec:Frequency-locked loop}) allows to detect all harmonic components and the fundamental frequency. First, a standard SOGI (as e.g.~discussed in~\cite{2011_Rodriguez_MultiresonantFLLsforGridSynchronizationofPowerConvertersUnderDistortedGridConditions}) for the $\nu$-th harmonic is revisited. After that, the proposed modification to it is introduced to obtain the modified SOGI with prescribed settling time. It is shown that the modified SOGI is actually a generalization of the standard SOGI. Next, their estimation performances are compared. Finally, to be capable of estimating all $n$ harmonics, the proposed modified (or standard) SOGIs are parallelized to obtain the overall observer system. Throughout this paper, the more powerful state space representation will be used, since the considered parallelized SOGIs with FLL represent a nonlinear system and transfer functions are not applicable.

\subsection[Standard SOGI (sSOGI) for the $\nu$-th harmonic component]{Standard SOGI (sSOGI) for the $\nu$-th harmonic component~\cite{2011_Rodriguez_MultiresonantFLLsforGridSynchronizationofPowerConvertersUnderDistortedGridConditions}}
\label{sssec_st_sogi}
For now, let $\nu \in \mathbb{H}_n$ and consider only the $\nu$-th harmonic component $y_{\nu}(t) := a_{\nu}(t) \cosine{\phi_{\nu}(t)}$. If the estimate $\omegahp{\nu} := \nu\omegah$ of the $\nu$-th harmonic frequency is known, the implementation of a sSOGI for the signal $y_{\nu}$ allows to obtain \emph{online} estimates $\widehat{y}_{\nu} = \widehat{x}^{\alpha}_\nu$ and $\widehat{q}_{\nu}=\widehat{x}^{\beta}_\nu$ of in-phase and quadrature signal, respectively. A sSOGI for the $\nu$-th harmonic component is depicted in Fig.~\ref{fig:sogi_and_poles}\,(a). Its dynamics are given by the following time-varying differential equation
\begin{equation}
\left.
\begin{array}{rcl}
	\ddtsmall \overbrace{\begin{pmatrix} \widehat{x}_\nu^{\alpha}(t) \\ \widehat{x}_\nu^{\beta}(t) \end{pmatrix}}^{=:\xsogip{\nu}(t)\in \R^2} & = &  \widehat{\omega}(t) \overbrace{\begin{bmatrix} -  \nu \kgainp{\nu} & -\nu \\ \nu & 0 \end{bmatrix}}^{=: \Asubt{\nu}(k_\nu) \in \R^{2 \times 2}}  \xsogip{\nu}(t) + \widehat{\omega}(t)   \overbrace{\begin{pmatrix} \nu \kgainp{\nu} \\ 0 \end{pmatrix}}^{=: \mv{l}_{\nu}(k_\nu) \in \R^2} \hspace*{-2ex} y_{\nu}(t), \qquad \xsogip{\nu}\br{0} = \xsogip{\nu,0} \in \R^2 \vspace*{2ex}\\
	\widehat{y}_\nu(t) & = & \underbrace{\begin{pmatrix} 1,  & 0  \end{pmatrix}}_{=:\mv{c}_\nu^\top \in \R^{1\times 2}}\xsogip{\nu}(t)
\end{array}
\right\}
\label{eq:sogi_state_space}
\end{equation}
with arbitrary initial value $\xsogip{\nu,0} \in \R^2$ (mostly set to zero), gain $k_{\nu} > 0$ (\emph{single} tuning factor) and estimate $\widehat{\omega}$ (possibly time-varying) of the fundamental angular frequency $\omega$. The gain $\kgainp{\nu}$ only allows for a \emph{limited} tuning of the dynamic response of the sSOGI. For constant $\widehat{\omega} > 0$ only, characteristic equation and poles of the $\nu$-th sSOGI are given as follows\footnote{Note that for time-varying or nonlinear systems, the analysis of poles is \emph{not} sufficient to check stability~\cite[Example 3.3.7]{2005_Hinrichsen_MathematicalSystemsTheoryI---ModellingStateSpaceAnalysisStabilityandRobustness}.}
\begin{equation}
\chi_{\nu}\br{s} := \det\bs{s \ve{I}_2 - \widehat{\omega}\Asubt{\nu}} = s^2 + \nu \widehat{\omega}  \kgainp{\nu} \, s + (\nu\widehat{\omega})^2 \musteq 0 \quad \Longrightarrow \quad p_{\nu,1/2} = - \frac{\nu\widehat{\omega}\kgainp{\nu}}{2} \br{1 \pm \sqrt{1 - \tfrac{4}{\kgainp{\nu}^2}}}.
\label{eq:characteristic polynomial and roots of nuth sSOGI}
\end{equation}
The respective root locus is shown in Fig.~\ref{fig:sogi_and_poles}\,(b). Hence, stability is guaranteed for all $\kgainp{\nu} > 0$. However, since the pole closest to the imaginary axis determines the settling time of the system, the smallest settling time is obtained for $\kgainp{\nu} = 2$ which clearly limits the tuning of the transient performance of the sSOGI. Moreover, this choice leads to two real poles at $- \frac{\nu\widehat{\omega}\kgainp{\nu}}{2}$ and, hence, the sSOGI is \emph{not} capable of oscillating by itself. Therefore, common tunings are $\kgainp{\nu} = \sqrt{2}/\nu$~\cite{2011_Rodriguez_MultiresonantFLLsforGridSynchronizationofPowerConvertersUnderDistortedGridConditions} or $\kgainp{\nu} = 1$~\cite{2007_Mojiri_Time-DomainSignalAnalysisUsingAdaptiveNotchFilter}. 

\begin{figure}[!tb]
	\centering
	\begin{tabular}{c}
		\subfloat[Block diagram of $\nu$-th sSOGI with \emph{one} single tuning parameter $k_\nu$.]{\includegraphics{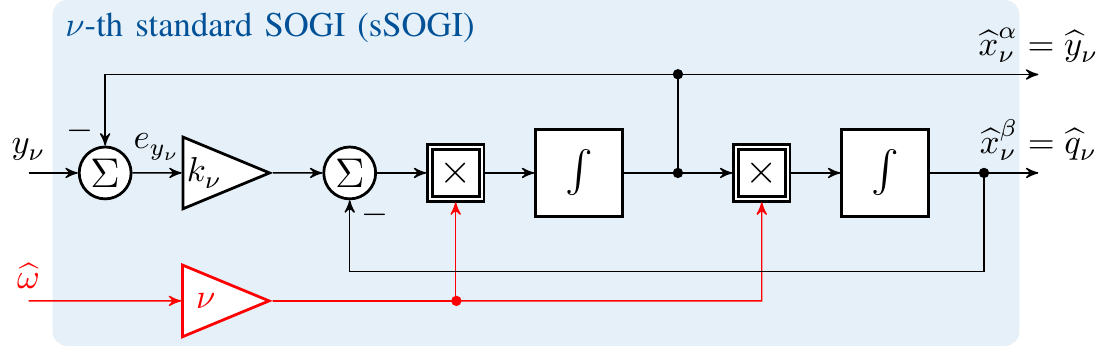}} \\[3ex] 
		\subfloat[Root locus of $\nu$-th sSOGI for $k_\nu > 0$ (unstable for $k_\nu<0$).]{\includegraphics{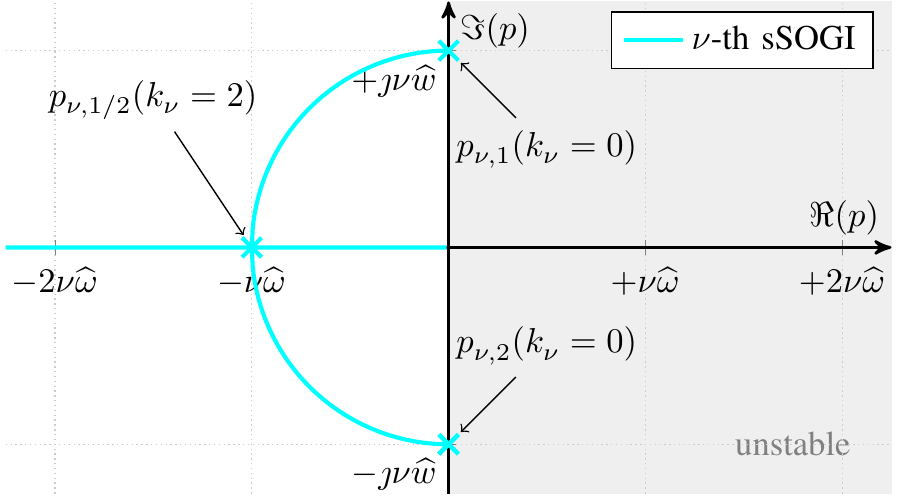}}
	\end{tabular}
	\caption{Standard Second-Order Generalized Integrator (\textbf{sSOGI})~\cite{2011_Rodriguez_MultiresonantFLLsforGridSynchronizationofPowerConvertersUnderDistortedGridConditions}: (a) Block diagram and (b) root locus of $\nu$-th sSOGI.}
	\label{fig:sogi_and_poles}
\end{figure}

\subsection{Modified SOGI (mSOGI) for the $\nu$-th harmonic component}\label{sec:modified-sogi-msogi-for-the-nu-th-harmonic-component}

To overcome the problem of the limited tuning without the possibility to prescribe the settling time, the modified SOGI (mSOGI) with additional gain $\ggainp{\nu}$ is introduced. The resulting block diagram of the $\nu$-th mSOGI is illustrated in Fig.~\ref{fig:esogi_and_poles}\,(a). Note that the additional gain does \emph{not} impair functionality but gives the necessary degree of freedom to enhance the transient performance as will be shown in the next subsection. 
\begin{figure}[!ht]
\centering
\begin{tabular}{c}
	  \subfloat[Block diagram of $\nu$-th mSOGI with \emph{two} tuning parameters $k_\nu$ and $g_\nu$.]{\includegraphics{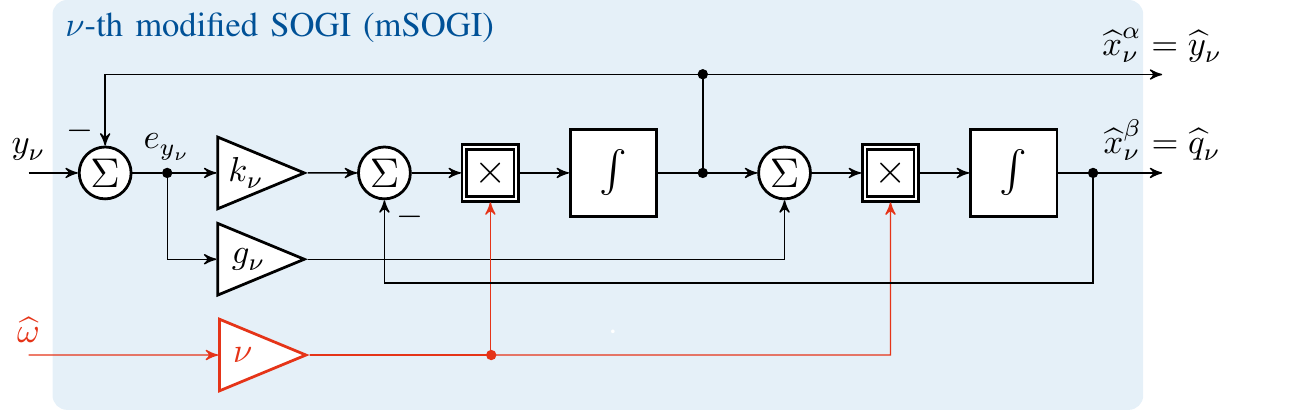}}\vspace{-0.3cm} \\
	  \subfloat[Root locus of $\nu$-th mSOGI for $k_\nu>0$ and $g_\nu =  - \tfrac{k_\nu^2}{4}$ (in general unstable for $k_\nu < 0$ or $\ggainp{\nu} > 1$).]{\includegraphics{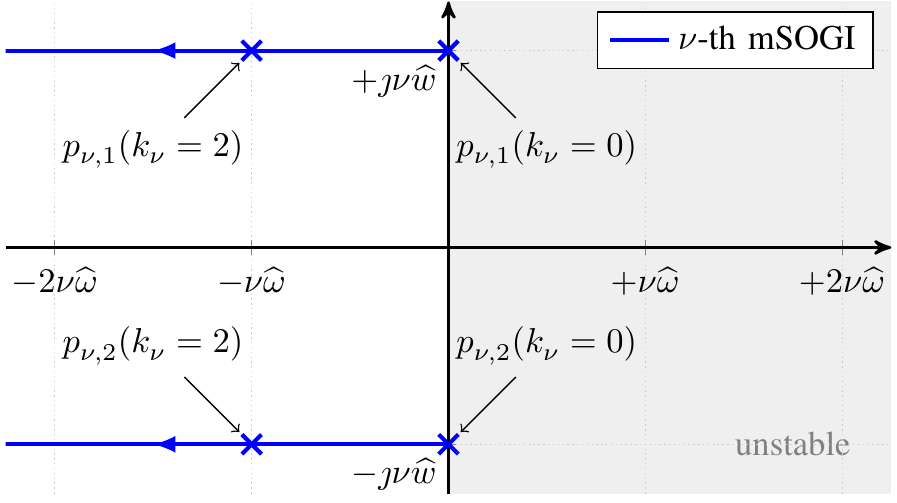}}
\end{tabular}
\caption{Modified Second-Order Generalized Integrator (\textbf{mSOGI}): (a) Block diagram and (b) root locus of $\nu$-th mSOGI.}
\label{fig:esogi_and_poles}
\end{figure}
The state space representation of the $\nu$-th mSOGI is given by the following time-varying differential equation:
\begin{equation}
\boxed{\left.
\begin{array}{rcl}
\ddtsmall \xsogip{\nu}(t) &  = & \widehat{\omega}(t)\overbrace{\begin{bmatrix} - \nu\kgainp{\nu} & -\nu \\ \nu(1 - \ggainp{\nu}) & 0 \end{bmatrix}}^{=: \Asubt{\nu}(k_\nu,g_\nu)} \xsogip{\nu}(t) + \widehat{\omega}(t) \hspace*{-1ex}\overbrace{\begin{pmatrix} \nu\kgainp{\nu} \\ \nu\ggainp{\nu} \end{pmatrix}}^{=: \mv{l}_{\nu}(k_\nu,g_\nu)} \hspace*{-1ex} y_{\nu}(t), \qquad \xsogip{\nu}\br{0} = \xsogip{\nu,0} \in \R^2 \\[2ex]
\widehat{y}_\nu(t) & = & \mv{c}_\nu^\top \xsogip{\nu}(t)
\end{array}
\right\}}
\label{eq:esogi_state_space}
\end{equation}
with arbitrary initial value $\xsogip{\nu,0} \in \R^2$ and estimate $\widehat{\omega}$ of $\omega$. The gains $\kgainp{\nu}$ and $\ggainp{\nu}$ now allow (theoretically\footnote{Of course, noise will limit the feasible tuning.}) for a \emph{limitless} tuning of the dynamic response of the mSOGI. The tiny but crucial difference between the mSOGI in~\eqref{eq:esogi_state_space} and the sSOGI in~\eqref{eq:sogi_state_space} is the additional gain $g_\nu$ in the system matrix $\mm{A}_{\nu}(k_\nu,g_\nu)$ and the vector $\mv{l}_{\nu}(k_\nu,g_\nu)$. For a constant frequency $\widehat{\omega}$, the characteristic equation and the poles of the $\nu$-th mSOGI are given by
\begin{eqnarray}
\chi_{\nu}\br{s} := \det\bs{s \ve{I}_2 - \widehat{\omega}\Asubt{\nu}} = s^2 + \nu\widehat{\omega} \kgainp{\nu}s + \br{1 - \ggainp{\nu}} (\nu\widehat{\omega})^2 \musteq 0 
& \Longrightarrow &  p_{\nu,1/2} = -\tfrac{ \nu \widehat{\omega} \kgainp{\nu}}{2}\br{1 \pm \sqrt{1 - 4\tfrac{(1- \ggainp{\nu})}{\kgainp{\nu}^2}}} \notag \\
& \stackrel{g_\nu =  - \tfrac{k_\nu^2}{4}}{\Longrightarrow} & p_{\nu,1/2} = -\tfrac{\nu \widehat{\omega} \kgainp{\nu}}{2} \pm \jmath \nu \widehat{\omega}.
\label{eq:characteristic polynomial and roots of nuth mSOGI}
\end{eqnarray}
The special choice of the additional gain $g_\nu = - \tfrac{k_\nu^2}{4}$ in~\eqref{eq:characteristic polynomial and roots of nuth mSOGI} gives the key feature of the mSOGI:~For any $k_\nu>0$, the real parts of the poles $p_{\nu,1/2}$ in~\eqref{eq:characteristic polynomial and roots of nuth mSOGI} can be chosen arbitrarily; whereas the capability of the mSOGI to oscillate with angular frequency $\nu\widehat{\omega}$ is preserved (see imaginary parts of $p_{\nu,1/2}$). The root locus of the $\nu$-th mSOGI is depicted in Fig.~\ref{fig:esogi_and_poles}\,(b). The mSOGI is stable for any $\kgainp{\nu} > 0$; and, the larger $k_\nu$ is chosen, the faster is its transient response.

\begin{remark}[Generaliziation of the sSOGI]
The introduction of the additional gain $g_\nu$ for the mSOGI in~\eqref{eq:esogi_state_space} represents actually a generalization of the sSOGI in~\eqref{eq:sogi_state_space}. Clearly, for $g_\nu = 0$,  the mSOGI simplifies to the sSOGI. In other words, only now, the term ``second-order \emph{generalized} integrator'' is really appropriate. 
\end{remark}

\subsection{Comparison of the estimation performances of sSOGI and mSOGI}
\label{sec:Estimation performance of sSOGI and mSOGI}
If $\widehat{\omega} = \omega$, both SOGIs are capable of estimating in-phase signal $\widehat{y}_{\nu} = \widehat{x}^{\alpha}_\nu$ and quadrature signal $\widehat{q}_{\nu}=\widehat{x}^{\beta}_\nu$ of the $\nu$-th harmonic signal $y_{\nu}(t) := a_{\nu}(t) \cosine{\phi_{\nu}(t)}$. The estimated amplitude 
\begin{equation} 
\widehat{a}_{\nu}(t) := \norm{\xsogip{\nu}(t)} = \sqrt{\yhp{\nu}(t)^2 + \qhp{\nu}(t)^2}
\label{eq:estimated amplitude of nu-th harmonic}
\end{equation}
is given by the norm of the estimated signal and its quadrature signal. The estimated phase angle is given by
\begin{equation} 
\widehat{\phi}_{\nu}(t) = \arctan\!2\big(\yhp{\nu}(t),\,\qhp{\nu}(t)\big) \quad \text{ with } \quad \arctan\!2(\cdot,\cdot) \text{ as in } \eqref{eq:[N]definition of atan2}.
\label{eq:estimated phase angle of nu-th harmonic}
\end{equation}
Hence, the parameters $\widehat{a}_{\nu}$ and $\widehat{\phi}_{\nu}$ of the $\nu$-th harmonic can be detected online. \\

In Fig.~\ref{fig:Comparison of transient response of nuth sSOGI and MSOGI}, the transient responses of sSOGI and mSOGI are shown in cyan and blue, respectively, for the first harmonic (i.e.~$\nu=1$, see Fig.~\ref{fig:Comparison of transient response of nuth sSOGI and MSOGI}(a)) and for second harmonic (i.e.~$\nu=2$, see Fig.~\ref{fig:Comparison of transient response of nuth sSOGI and MSOGI}(b)). Four tunings of the gain $k_\nu$ are implemented and illustrated by different line types:~$k_\nu=0.5$ (dotted), $k_\nu=1$ (dashed), $k_\nu=2$ (dash-dotted) and $k_\nu=10$ (solid). The larger $k_\nu$ is chosen, the faster is the transient response of the mSOGI. Moreover, for $k_\nu = 2$ (dash-dotted) or $k_\nu = 10$ (solid), settling times of e.g.~$t_{\mathrm{set}} = \SI{0.01}{\second}$ and $t_{\mathrm{set}} = \SI{0.005}{\second}$ can be guaranteed for the fundamental signal, respectively. For the second harmonic, the transient response is twice as fast as for the fundamental signal. For the sSOGI, a prescribed settling time \emph{cannot} be ensured, since one pole approaches the imaginary axis for large choices of $k_\nu$ (see also Fig.~\ref{fig:sogi_and_poles}). In particular, the estimation of the quadrature component is slow (see $e_q$ in Fig.~\ref{fig:Comparison of transient response of nuth sSOGI and MSOGI}) which degrades the estimation speed of positive, negative and zero sequences in three-phase systems (not considered in this paper).
\begin{figure}[!tb]
	\begin{tabular}{cc}
		\subfloat[Estimation of fundamental component of first harmonic ($\nu = 1$).]{%
			\includegraphics[clip,width=8.75cm]{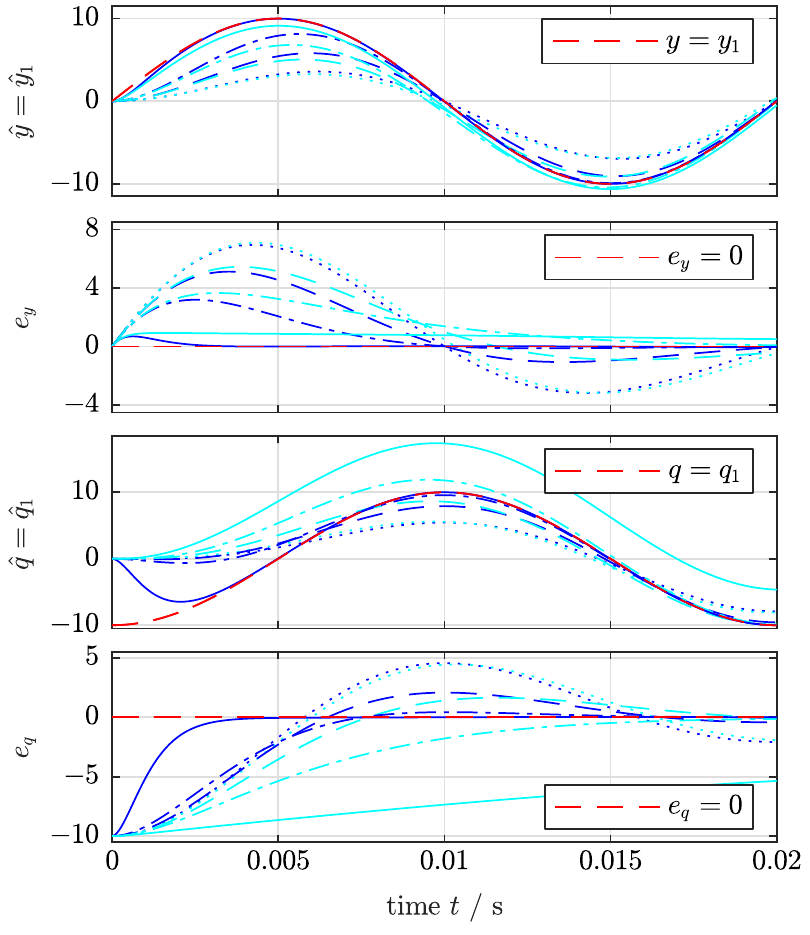}}
		&
		\subfloat[Estimation of second harmonic component ($\nu = 2$).]{
			\includegraphics[clip,width=8.75cm]{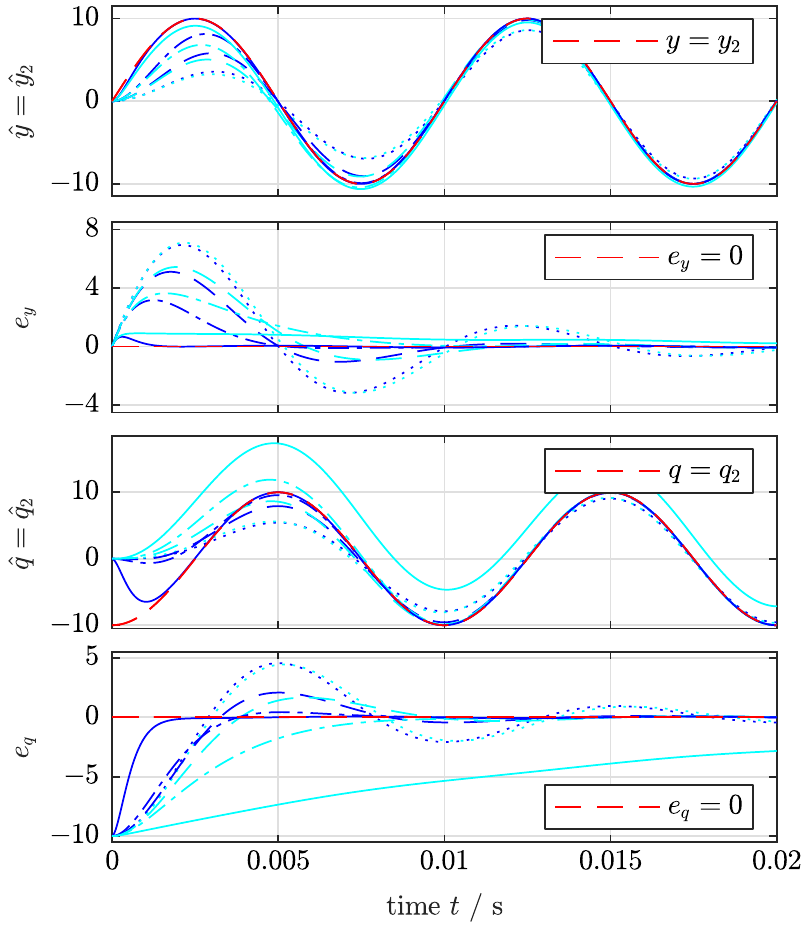}}
	\end{tabular}
	\caption{Comparison of estimation performances of $\nu$-th sSOGI (\protect\cyandottedline, \protect\cyandashedline,  \protect\cyandashdottedline, \protect\cyanline) and  $\nu$-th mSOGI (\protect\bluedottedline, \protect\bluedashedline,  \protect\bluedashdottedline, \protect\blueline) for four different tunings of gain $k_\nu \in \{0.5, 1, 2, 10\}$, respectively. Signals shown in (a) for $\nu = 1$ and in (b) for $\nu=2$ are from top to bottom: input signal $y_\nu$ and its estimate $\widehat{y}_\nu$, estimation in-phase error  $e_{y} = y_\nu - \widehat{y}_\nu$, quadrature signal $q_\nu$ and its estimate $\widehat{q}_\nu$ and estimation quadrature error  $e_{q} = q_\nu - \widehat{q}_\nu$}
	\label{fig:Comparison of transient response of nuth sSOGI and MSOGI}
\end{figure}

\subsection{Parallelization of the mSOGIs}
\label{sec:Parallelization of SOGIs}
This far, the presented SOGIs (sSOGIs and mSOGIs) can only estimate in-phase signal $\widehat{y}_{\nu} = \widehat{x}^{\alpha}_\nu$ and quadrature signal $\widehat{q}_{\nu}=\widehat{x}^{\beta}_\nu$ of the $\nu$-th harmonic signal $y_{\nu}(t) := a_{\nu}(t) \cosine{\phi_{\nu}(t)}$. By parallelizing $n$ of the mSOGIs or sSOGIs (see Fig.~\ref{fig:whole_model}), it is possible to extract in-phase and quadrature signal of each harmonic component $y_{\nu}$ for all $\nu \in \mathbb{H}_n$. For the parallelized sSOGIs, stability is preserved for a positive choice of all gains, i.e.~$\kgainp{\nu} > 0$ for all $\nu \in \mathbb{H}_n$~\cite{2007_Mojiri_Time-DomainSignalAnalysisUsingAdaptiveNotchFilter}. Stability for the parallelized mSOGIs will be guaranteed by pole placement. Moreover, the settling time can only be pre-specified by the parallelized mSOGIs.\\

The idea of the parallelization can be motivated by recalling the \emph{internal model principle} which states that "[e]very good regulator [or observer] must incorporate a model of the outside world being capable to reduplicate the dynamic structure of the exogenous signals which the regulator [or observer] is required to process."~\cite{1985_Wonham_LinearMultivariableControl:AGeometricApproach}. In the considered case, the exogenous signal $y$ as in~\eqref{eq:input signal y} can be reduplicated by the parallelization of $n$ sinusoidal internal models~\cite[Chapter~20]{2017_Hackl_Non-identifierbasedadaptivecontrolinmechatronics:TheoryandApplication}, which have the overall dynamics 
\begin{equation}
\left.
\begin{array}{rcl}
\ddtsmall \mv{x}(t)  &=&  \omega(t)\JMAT\ve{x}(t), \qquad \qquad \qquad \qquad  \ve{x}\br{0} = \ve{x}_{0} \neq \mv{0}_{2n} \in \R^{2n} \vspace*{1ex}\\
y(t)  & = & \underbrace{(1,\, 0,\, 1,\, 0,\, \cdots,\, 1,\, 0)}_{=: \,\cyvec^{\top} \in \R^{2n}}\ve{x}(t)
\end{array} \qquad 
\right\}
\label{eq:state space dynamics of overall IM}
\end{equation}
where 
\begin{equation}
\mv{x}:=\big((\underbrace{x_1^{\alpha},\, x_1^{\beta})}_{=: \ve{x}_1^{\top}},\, \ldots,\, \ve{x}_n^{\top}\big)^{\top}, \; \JMAT := \blockdiag\br{\Jbar, \nu_2\Jbar, \cdots, \nu_n\Jbar} \in \R^{2n \times 2n} \; \text{ and } \; \Jbar = \begin{bmatrix} 0 & -1 \\ 1 & 0 \end{bmatrix} =  -\Jbar^\top = - \Jbar^{-1}.
\label{eq:definitions of x, J and Jbar}
\end{equation}
The initial values of the internal model in~\eqref{eq:state space dynamics of overall IM} allow to determine amplitude $a_\nu$ and angle $\phi_{\nu}$ of the $\nu$-th harmonic. For constant $\omega>0$ and differing harmonics $\nu_i \neq \nu_j$ for all $i\neq j \in \{1,\dots n\}$, the overall internal model~\eqref{eq:state space dynamics of overall IM} is completely state observable (see Proposition~\ref{prop:observability of parallelized IMs} in the appendix). \\

Now, by substituting estimate $\widehat{\omega}$ for $\omega$,  the observer is obtained and consists of the parallelized mSOGIs (as introduced in~\eqref{eq:esogi_state_space} for the $\nu$-th harmonic). The observer dynamics are nonlinear and given by
\begin{equation}
\boxed{%
\left.\begin{array}{rcl}
\ddtsmall\xsogi(t) &=& \widehat{\omega}(t)\JMAT\xsogi(t) + \widehat{\omega}(t) \,\mv{l}\, \big(y(t) - \overbrace{\mv{c}^\top \widehat{\mv{x}}(t)}^{=\yhout(t)}\big)  \vspace*{1ex}\\
 &=& \widehat{\omega}(t)\big[\underbrace{\JMAT - \mv{l}\mv{c}^\top}_{=:\mm{A}}\big]\xsogi(t) + \widehat{\omega}(t) \,\mv{l}\, y(t) , \qquad \qquad \widehat{\mv{x}}(t) = \widehat{\mv{x}}_0 \in \R^{2n} \vspace*{2ex}\\
\yhout(t) &=& \underbrace{\big(\mv{c}_1^\top, \mv{c}_{\nu_2}^\top, \cdots,  \mv{c}_{\nu_n}^{\top} \big)}_{\stackrel{\eqref{eq:state space dynamics of overall IM},\eqref{eq:esogi_state_space}}{=}\cyvec^{\top}}\xsogi(t),
\end{array}\right\}}
\label{eq:observer_sogi}
\end{equation}
where observer state vector $\widehat{\mv{x}} \stackrel{\eqref{eq:esogi_state_space}}{=} \big(\widehat{\mv{x}}_1^\top, \widehat{\mv{x}}_{\nu_2}^\top, \cdots,  \widehat{\mv{x}}_{\nu_n}^{\top} \big)^\top \in \R^{2n}$ and observer gain vector 
\begin{equation}
\mv{l} := \big(\mv{l}_1^\top, \mv{l}_{\nu_2}^\top, \cdots,  \mv{l}_{\nu_n}^{\top} \big)^\top \stackrel{\eqref{eq:esogi_state_space}}{=} \br{\kgainp{1},\, \ggainp{1},\, \ldots,\, \nu_n\kgainp{n},\, \nu_n\ggainp{n}}^{\top} \in \R^{2n}
\label{eq:general feedback gain vector l}
\end{equation}
merge the individual sub-state estimation vectors $\widehat{\mv{x}}_\nu$ and gain vectors $\mv{l}_\nu$ of the $\nu$ mSOGIs as in~\eqref{eq:esogi_state_space}. The observer will be tuned by pole placement and, hence, the gains in $\mv{l}$ can be determined by comparing the coefficients of the characteristic polynomial
\begin{equation}
\chi_{\AMAT}\br{s} = \prod\limits_{i = 1}^{n}\br{s^2 + \nu_i^2} - \sum\limits_{i = 1}^{n}\ggainp{i}\nu_i^2\prod\limits_{\substack{k = 1 \\ k \neq i}}^{n}\br{s^2 + \nu_{k}^2} + s\sum\limits_{i = 1}^{n}\kgainp{i}\nu_i\prod\limits_{\substack{k = 1 \\ k \neq i}}^{n}\br{s^2 + \nu_k^2}
\label{eq:char_poly}
\end{equation}
of the closed-loop system matrix $\AMAT := \JMAT - \mv{l}\cyvec^{\top}$ in \eqref{eq:observer_sogi}
and the coefficients of a desired polynomial
\begin{equation}
	\chi_{\mm{A}}^*\br{s} := \prod\limits_{i = 1}^{2n}\br{s - p_i^*}
\label{eq:des_char_poly}
\end{equation}
with $2n$ prescribed stable roots (poles) $p_i^* \in \Cneg$, $i \in \{1,\dots,2n\}$, in the negative complex half-plane. The detailed derivation of the analytical solution of the pole placement algorithm is presented in Appendix~\ref{sec_app_pole_placement}. The resulting feedback gain vector $\mv{l}$ is obtain as follows
\begin{equation}
\boxed{	\mv{l} = \SMAT \,\widetilde{\mv{p}}_{\mm{A}}^*, }
	\label{eq:feedback gain vector l of parallelized mSOGIs}
\end{equation}
where
\begin{equation}
\SMAT := \begin{bmatrix} \Ssub{1,1} & \cdots & \Ssub{n,1} \\ \vdots & \ddots & \vdots \\ \Ssub{1,n} & \cdots & \Ssub{n,n} \end{bmatrix}, \quad \Ssub{c,r} := \br{- 1}^{c + 1}\nu_r^{2\br{n - c}}\Rsub{r}\prod\limits_{\substack{i = 1 \\ i \neq r}}^{n}\br{\nu_r^2 - \nu_i^2}^{- 1} \quad \text{and} \quad \Rsub{i} := \begin{bmatrix} 1 & 0 \\ 0 & - \tfrac{1}{\nu_i} \end{bmatrix}
 \label{eq:definition of S, S_cr and R_i}
\end{equation}
and
\begin{equation}
 \widetilde{\mv{p}}_{\mm{A}}^* := \br{- \sum\limits_{i = 1}^{2n}p_i^*,\;\; \sum\limits_{i = 1}^{2n}p_i^*\sum\limits_{j = i + 1}^{2n}s_j - \sum\limits_{i = 1}^{n}\nu_i^2,\;\; - \sum\limits_{i = 1}^{2n}p_i^*\sum\limits_{j = i + 1}^{2n}p_j^* \sum\limits_{k = j + 1}^{2n}p_k^*,\;\; \ldots,\;\; \prod\limits_{i = 1}^{2n}p_i^* - \prod\limits_{i = 1}^{n}\nu_i^2}^{\top}.
 \label{eq:coefficient vector p_prime_des of desired polynomial}
\end{equation}
It can be shown that, for any positive (but possibly time-varying) angular frequency estimate $\widehat{\omega}(t) \geq \varepsilon_{\omega} > 0$ for all $t\geq 0$,  the closed-loop observer system~\eqref{eq:observer_sogi} is bounded-input bounded-output (BIBO) stable \emph{and} input-to-state stable (ISS). Moreover, if $\widehat{\omega} \to \omega$, then the estimation state error $\mv{e}_x := \mv{x} - \widehat{\mv{x}} \to \mv{0}_{2n}$  decays \emph{exponentially} to zero (see Theorem~\ref{thm:BIBO/S stability}, Theorem~\ref{thm:asymptotic tracking} and Remark~\ref{rem:Exponential stability and input-to-state stability} in the appendix).

\begin{remark}[\texttt{place} command in Matlab versus analytical expression in~\eqref{eq:feedback gain vector l of parallelized mSOGIs}]
	For small $n$ (e.g.~$n \leq 10$), the  Matlab command \texttt{place} can be used to compute $\mv{l} = \texttt{place}(\JMAT^\prime, \mv{c}, \dots)$. 
	For large $n$, \texttt{place} might not provide a proper result. Moreover, \texttt{place} cannot place poles with multiplicity greater than $\rank\br{\cyvec} = 1$. That is why, the analytical expression in~\eqref{eq:feedback gain vector l of parallelized mSOGIs} has been derived. It can be used to achieve pole placement for arbitrarily large $n$.
\end{remark}

\section{Frequency-Locked Loop (FLL):~Frequency estimation}
\label{sec:Frequency-locked loop}

As mentioned above, a correct estimate of the fundamental angular frequency is essential for a proper functionality of the parallelized SOGIs and the harmonics detection. The following subsections motivate and discuss the necessary modifications of the FLL to ensure its functionality also with the parallelized mSOGIs.

\subsection[Standard FLL (sFLL)]{Standard FLL (sFLL)~\cite{2011_Rodriguez_MultiresonantFLLsforGridSynchronizationofPowerConvertersUnderDistortedGridConditions}}

In this subsection, the standard FLL is re-visited. Its block diagram is shown in Fig.~\ref{fig:block diagram of sFLL}. Adaption law and gain normalization are briefly explained. 
\begin{figure}[ht!]
	\centering
	\includegraphics{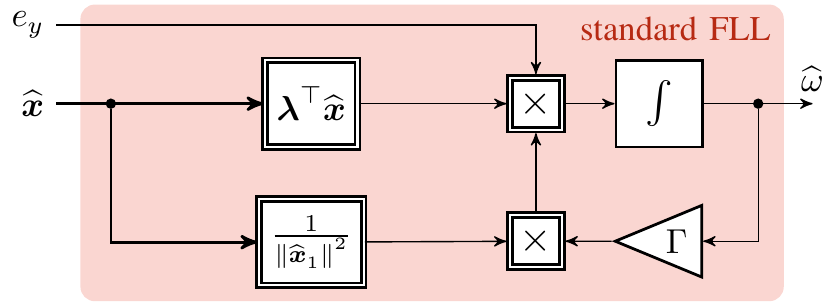}
	\caption{Block diagram of standard Frequency Locked Loop (sFLL) with gain normalization~\cite{2007_Mojiri_Time-DomainSignalAnalysisUsingAdaptiveNotchFilter,2011_Rodriguez_MultiresonantFLLsforGridSynchronizationofPowerConvertersUnderDistortedGridConditions, 2017_Matas_AFamilyofGradientDescentGridFrequencyEstimatorsfortheSOGIFilter}.}
	\label{fig:block diagram of sFLL}
\end{figure}

\subsubsection{Adaption law}
\label{sec:FLL adaption law}
As shown in Fig.~\ref{fig:whole_model}, any of the $n$ parallelized mSOGIs requires an estimate $\widehat{\omega}$ of the fundamental angular frequency $\omega$. The estimate $\widehat{\omega}$ is the output of the FLL. The nonlinear adaption law of the sFLL is given by~\cite{2011_Rodriguez_MultiresonantFLLsforGridSynchronizationofPowerConvertersUnderDistortedGridConditions} 
\begin{equation}
\ddtsmall\omegah(t) = \gamma(t)\lamvec^{\top}\xsogi(t)\ey(t) \stackrel{\text{\cite{2007_Mojiri_Time-DomainSignalAnalysisUsingAdaptiveNotchFilter}}}{=} \gamma(t)\, k_1 \,\widehat{x}_1^\beta(t)\,\ey(t), \qquad \widehat{\omega}(0) =  \widehat{\omega}_0 \in \R, 
\label{eq:fll_adaption_rule}
\end{equation}
where $\gamma(\cdot) > 0$ is a positive but non-constant adaptive gain, $\mv{\lambda} = (0,-k_1,\mv{0}_{2n-2}^\top)^\top$~\cite{2007_Mojiri_Time-DomainSignalAnalysisUsingAdaptiveNotchFilter,2011_Rodriguez_MultiresonantFLLsforGridSynchronizationofPowerConvertersUnderDistortedGridConditions} is a constant "selection" vector (to extract only the fundamental estimate $\widehat{x}_1^\beta$ from $\widehat{\mv{x}}$), $\widehat{\mv{x}}$ is the estimation vector of the parallelized sSOGIs and $e_y := y - \widehat{y}$ is the estimation error (difference between input $y$ and estimated input $\widehat{y}$). A proper choice of the initial value, e.g.~$\widehat{\omega}_0 \in \{2\pi\,50, 2\pi\,60 \}$, of the sFLL adaption law is beneficial for functionality and adaption speed.

\begin{remark}[Impact of negative estimates of the angular frequency]
\label{rem:Negative angular frequency estimates}
Note that, in view of the adaption law in~\eqref{eq:fll_adaption_rule}, the estimated angular frequency might also become negative, i.e.~$\omegah(\tau)<0$ for some time instant $\tau \geq 0$. However, a negative $\omegah<0$ will result in \emph{instability} of the parallelized sSOGIs and all estimated states will \emph{diverge}. 
\end{remark}

\subsubsection{Gain Normalization (GN)}

The FLL should be robustified to work for signals with arbitrary fundamental amplitudes and angular frequencies (see \cite{2011_Rodriguez_MultiresonantFLLsforGridSynchronizationofPowerConvertersUnderDistortedGridConditions, 2017_Matas_AFamilyofGradientDescentGridFrequencyEstimatorsfortheSOGIFilter}). This can be achieved by introducing the following adaptive sFLL gain 
\begin{equation}
\gamma(t) := \Gamma\tfrac{\omegah(t)}{\norm{\xsogip{1}(t)}^2} \qquad \Longrightarrow \qquad  \ddtsmall\omegah = \Gamma \tfrac{\omegah(t)}{\norm{\xsogip{1}(t)}^2} \, k_1 \,\widehat{x}_1^\beta(t)\,\ey(t),
\label{eq:sFLL}
\end{equation}
which depends on gain $\Gamma > 0$, frequency estimate $\omegah$ and norm of the \emph{fundamental} estimation vector $\widehat{\mv{x}}_1= (\widehat{x}_1^\alpha, \, \widehat{x}_1^\beta)^\top$ leading to a "normalized" FLL adaption law. The gain $\Gamma > 0$ is a \emph{constant} tuning factor of the FLL. 
\begin{remark}[Avoiding division by zero]
\label{rem:Avoiding a division by zero}
Depending on the initial values $\xsogi\br{0}$ and the time evolution of estimation process, the denominator $\norm{\xsogip{1}(t)}^2$ in~\eqref{eq:sFLL} might become zero for certain time instants $t\geq 0$. This must and can easily be avoided by introducing a minimal positive value for the denominator by substituting $\max\br{\norm{\xsogip{1}(t)}^2, \varepsilon}$ for $\norm{\widehat{\ve{x}}_1(t)}$ in~\eqref{eq:sFLL}  where $\varepsilon > 0$ is a small positive constant.
\end{remark}

\subsection{Modified FLL (mFLL)}\label{sec:modified-fll-mfll}
The FLL is the weakest subsystem (bottleneck) of the overall grid estimation system; in particular, its tuning endangers system stability, estimation accuracy and estimation speed. Only if the frequency is detected correctly, the mSOGIs or sSOGIs work properly. Therefore, to improve stability and performance of the estimation process a modified FLL is proposed.
The block diagram of the proposed mFLL is depicted in Fig.~\ref{fig:block diagram of mFLL}. Remarks~\ref{rem:Negative angular frequency estimates} and~\ref{rem:Avoiding a division by zero} have already been considered in the block diagram. In addition, the mFLL is equipped with a generalized adaption law, a sign-correct anti-windup strategy and a rate limitation. All three modifications enhance performance and stability of the mFLL. The generalized adaption law increases adaption speed. The anti-windup strategy guarantees that the estimated angular frequency $\omegah$ remains bounded and positive for all time and the rate limitation prevents too fast adaption speeds which might endanger stability. Details will be explained in the next subsections.
\begin{figure}[ht!]
	\centering
	\includegraphics{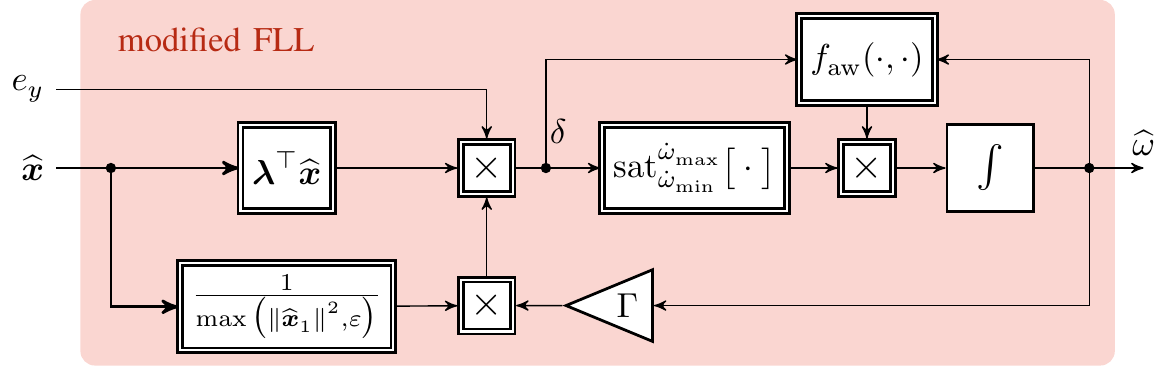}
	\caption{Modified Frequency Locked Loop (mFLL) with gain normalization, sign-correct anti-windup and rate limitation.}
	\label{fig:block diagram of mFLL}
\end{figure}

\subsubsection{Generalized adaption law}
The presented adaption law~\eqref{eq:fll_adaption_rule} of the sFLL does not work properly for the mSOGIs. It does \emph{not} guarantee a a sign-correct adaption for all time. Therefore, the adaption law must be generalized to fit to the parallelized mSOGIs. It is clear that for a sign-correct adaption of the estimated angular frequency $\omegah$, the generalized adaption law must ensure that the following conditions hold
\begin{equation}
\left.\begin{array}{lcl}
\forall\, \omegah < \omega \wedge e_y \neq 0 \quad & \Longrightarrow & \quad  \ddtsmall\omegah \propto \lamvec^{\top}\xsogi \, \ey > 0 ,  \\[0.5ex]
\forall\,  \omegah = \omega  \wedge e_y \neq 0  \quad & \Longrightarrow & \quad  \ddtsmall\omegah \propto  \lamvec^{\top}\xsogi \, \ey  = 0 \quad \text{ and }  \\[0.5ex]
\forall\,  \omegah > \omega  \wedge e_y \neq 0  \quad & \Longrightarrow & \quad  \ddtsmall\omegah \propto  \lamvec^{\top}\xsogi \, \ey < 0. 
\end{array}\right\}
\label{eq:sign correct adaption law of FLL}
\end{equation}
To illustrate the intuition behind these conditions, assume that the input signal has a constant fundamental angular frequency $\omega > 0$ and that the parallelized mSOGIs are fed by an arbitrary positive but constant estimate $0<\omegah \neq \omega$. Then, in \emph{steady state}, the system states $\xsogi(t)$ with their characteristic amplitude and phase responses can be used to analyze whether $\ey$ and $\lamvec^{\top}\xsogi$ are \emph{in}-phase or \emph{counter}-phase. In Appendix~\ref{sec_fll_principle}, it is shown that this sign-correct adaption is guaranteed when the selection vector $\lamvec$ is chosen as follows 
\begin{equation}
\boxed{
\lamvec := \blockdiag\br{\Jbar^{-1},\, \mm{O}_{2\times2},\, \ldots,\, \mm{O}_{2\times2}}\mv{l} = (g_1,-k_1,\mv{0}_{2n-2}^\top)^\top  \in \mathbb{R}^{2n}.}
\label{eq:choice of lambda for sign-correct adaption}
\end{equation}
This choice of $\lamvec$ can be used for sSOGI \emph{and} mSOGI as well. It is actually a generalization of the standard choice $\lamvec = (0,-k_1,\mv{0}_{2n-2}^\top)^\top = \blockdiag \big(\Jbar^{-1},\, \ve{0}_{2\times2},\, \ldots,\, \ve{0}_{2\times2} \big)\mv{l}$ with $g_\nu = 0$ for all $\nu \in \mathbb{H}_n$ (recall~\eqref{eq:fll_adaption_rule}). Finally, note that the sign-correct adaption was derived based on a steady state analysis (see Appendix \ref{sec_fll_principle}). This implies that the mFLL (and sFLL) dynamics should be slow compared to the dynamics of the parallelized mSOGIs (which can be achieved by an adequate choice of $\Gamma$). 

\subsubsection{Sign-correct anti-windup strategy}
Usually, the grid frequency should not exceed a certain interval. This physically motivated limitation can be exploited for the frequency estimation. The principle idea of the proposed sign-correct anti-windup strategy is illustrated in Fig.~\ref{fig:Illustration of principle idea of FLL with anti-windup}. More precisely, the adaption of the estimated angular frequency shall be stopped (i.e.~$\ddtsmall \omegah=0$), when 
\begin{itemize}
	\item the estimated angular frequency $\omegah$ leaves the admissible interval, i.e. $\widehat{\omega} \not\in (\omega_{\min}, \, \omega_{\max})$ with lower and upper limit $0<\omega_{\min}<\omega_{\max}$, respectively (see Fig.~\ref{fig:Illustration of principle idea of FLL with anti-windup}); and
	\item  the right hand side of the adaption law~\eqref{eq:sign correct adaption law of FLL} has wrong sign (otherwise the estimation gets stuck at one of the limits). 
\end{itemize}
\begin{figure}[H]
	\centering
	\includegraphics{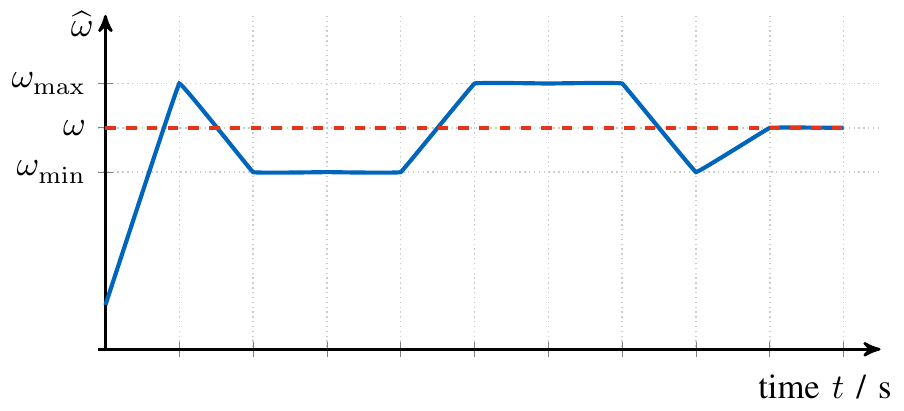}
	\caption{Illustration of the principle idea of the proposed sign-correct anti-windup strategy.}
	\label{fig:Illustration of principle idea of FLL with anti-windup}
\end{figure}
This yields to the following sign-correct anti-windup decision function
\begin{equation}
f_{\mathrm{aw}}(\widehat{\omega},\delta) := \begin{cases}
0, \quad \text{ for }\; \big(\widehat{\omega} \geq \omega_{\max} \wedge \delta \propto \ddtsmall \widehat{\omega} \geq 0 \big) \vee \big(\widehat{\omega} \leq \omega_{\min} \wedge \delta \propto \ddtsmall \widehat{\omega} \leq 0 \big) \\
1, \quad \text{ else}
\end{cases}
\label{eq:sign-correct AW decision function}
\end{equation}
where $\delta \propto \ddtsmall \omegah$ is proportional to the time derivative of the estimated angular frequency as can be seen when decision function and frequency adaption law are combined as follows
\begin{equation}
\ddtsmall \widehat{\omega} = f_{\mathrm{aw}}(\widehat{\omega},\delta)  \underbrace{
	\tfrac{\Gamma \,\widehat{\omega} \, e_y \,\mv{\lambda}^\top \widehat{\mv{x}}}{\max\big(\norm{\widehat{\mv{x}}_1}^2, \varepsilon\big)}}_{=:\delta}.
\label{eq:FLL with AWU}
\end{equation}
The consequences of this adaption law with sign-correct anti-windup are that the estimated angular frequency is \emph{positive} and remains \emph{bounded} for all time, i.e.~$\omegah(t) \in \big[\min(\omegah_0, \omega_{\min}), \, \max(\omegah_0,\omega_{\max})]$  for all $t \geq0$.  Moreover, once within the admissible interval $[\omega_{\min}, \omega_{\max}]$, the estimated angular frequency will remain inside this interval. Clearly, if the initial value $\omegah_0$ of the frequency estimate starts outside of $[\omega_{\min}, \omega_{\max}]$, it will approach the interval due to the sign-correct frequency adaption (as illustrated in Fig.~\ref{fig:Illustration of principle idea of FLL with anti-windup} for $\omegah_0 < \omega_{\min}$). Note that the proposed anti-windup strategy does not require tuning of an additional feedback gain as in~\cite{2018_Karkevandi_FrequencyestimationwithantiwinduptoimproveSOGIfiltertransientresponsetovoltagesags}. Instability can \emph{not} occur, since the proposed sign-correct anti-windup strategy is based on the simple idea of \emph{conditional  integration}~\cite[Section~10.4.1]{2017_Hackl_Non-identifierbasedadaptivecontrolinmechatronics:TheoryandApplication}.

\subsubsection{Rate limitation}
Recall that the overall observer~\eqref{eq:esogi_state_space} is nonlinear. Considering the estimated angular frequency as time-varying parameter, the observer becomes a time-varying linear system.  If the time derivative $\ddtsmall \omegah$ is limited (\emph{rate limitation}), the observer can be considered as \emph{slowly} time-varying system~\cite{1996_Rugh_LinearSystemTheory} (which simplifies stability analysis). For this \emph{rate limitation} of the adaption law, the admissible rate of the estimated angular frequency must be bounded, i.e.~
$\ddtsmall \widehat{\omega} \in \big[ \dot{\omega}_{\min}, \dot{\omega}_{\max} \big]$ where $\dot{\omega}_{\min}<0$ and $\dot{\omega}_{\max}>0$ are desired lower and upper thresholds, respectively. The idea of the rate limitation is illustrated in the graph shown in Fig.~\ref{fig:Illustration of principle idea of FLL with rate limitation}.
\begin{figure}[H]
	\centering
	\includegraphics{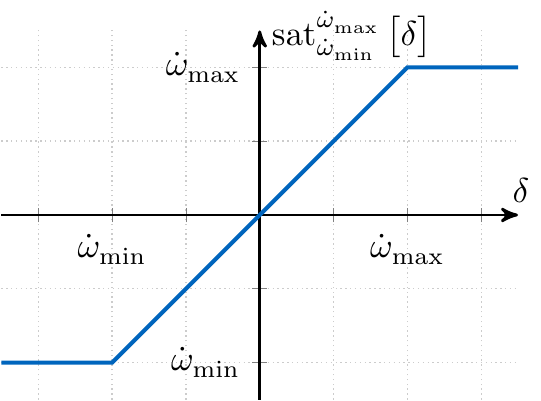}
	\caption{Illustration of the rate limitation of $\delta \propto \ddtsmall \omegah$.}
	\label{fig:Illustration of principle idea of FLL with rate limitation}
\end{figure}
Usually, the rate limitation leads to a smoother adaption and $\dot{\omega}_{\min} = -\dot{\omega}_{\max}$ is a meaningful choice. Reasonable rate thresholds were found out to be $10-\SI{100}{\hertz}$ per $\SI{1}{\milli\second}$. 
The rate limitation can be ensured by introducing an additional saturation function to the adaption law~\eqref{eq:FLL with AWU} leading to the generalized adaption law for the mFLL as shown next. 

\subsection{Generalized adaption law}
Finally, combining rate limitation,  sign-correct anti-windup strategy and sFLL with gain normalization, the generalized adaption law of the mFLL with $\mv{\lambda}$ as in~\eqref{eq:choice of lambda for sign-correct adaption} can be introduced. It is given by 
\begin{equation}
\boxed{%
\ddtsmall \widehat{\omega} = f_{\mathrm{aw}}(\widehat{\omega},\delta) \cdot   \sat_{\dot{\omega}_{\min}}^{\dot{\omega}_{\max}} \Big[ \underbrace{
\tfrac{\Gamma \,\widehat{\omega} \, e_y \,\mv{\lambda}^\top \widehat{\mv{x}}}{\max\big(\norm{\widehat{\mv{x}}_1}^2, \varepsilon\big)}}_{=:\delta}\Big] \; \text{where} \; f_{\mathrm{aw}}(\cdot,\cdot)\; \text{as in~\eqref{eq:sign-correct AW decision function} and} \; \sat_{\dot{\omega}_{\min}}^{\dot{\omega}_{\max}}\!\big[ \delta \big] := 
\begin{cases}
\dot{\omega}_{\max} &, \delta > \dot{\omega}_{\max}  \\ 	
\delta &, \dot{\omega}_{\min} \leq \delta \leq \dot{\omega}_{\max}  \\
\dot{\omega}_{\min} &, \delta < \dot{\omega}_{\min}  \\ 	
\end{cases}
}
\label{eq:mFLL with AWU and rate limiter}
\end{equation}
which guarantees that (i) the derivative of the estimated angular frequency is bounded, i.e.~$\ddtsmall \widehat{\omega}(t) \in \big[ \dot{\omega}_{\min}, \dot{\omega}_{\max} \big]$,  and (ii) the estimated angular frequency is positive and bounded from below and above, i.e.~$\widehat{\omega}(t) \in \big[\min(\omegah_0, \omega_{\min}), \, \max(\omegah_0,\omega_{\max})]$ for all $t \geq 0$.

\section{Implementation and measurement results}\label{sec:implementation-and-measurement-results}

To validate the proposed algorithms, measurements at a laboratory setup are carried out. The laboratory setup is shown in Fig.~\ref{fig:laboratory setup}.
\begin{figure}[!tb]
	\centering
	\includegraphics[clip,width=0.4\textwidth]{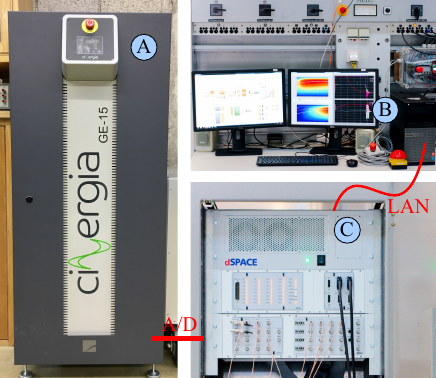}
	\caption{Laboratory setup: (A) Cinergia grid emulator, (B) Host-PC and (C) dSPACE real-time system .}
	\label{fig:laboratory setup}
\end{figure}
For measurements, the voltage is produced by the grid emulator. These voltages are measured by a LEM DVL 500 voltage sensor, analogue-to-digital converted by the dSPACE A/D card DS2004 and internally filtered by a low pass filter with cut-off frequency $\omega_{\text{lpf}} = 5000\si{\radian\per\second}$ to suppress high frequency noise. The implementation is done via Matlab/Simulink R2017a on the Host-PC. The executable observers are downloaded via LAN to the dSPACE Processor Board DS1007 and run in real time. The measurement data is captured and analyzed on the Host-PC after the experiment.  The implementation data of the conducted measurements is listed in Tab.~\ref{tab:implementation data}. \\

\begin{table}[tb!]
	\centering
	\begin{tabular}{ll}
		\toprule
		\textbf{Implementation} & \\ 
		sampling time & $h = \SI{0.1}{\milli\second}$ \\
		low-pass filter & $\omega_{\text{lpf}} = \SI{5e3}{\radian\per\second}$ (for measurements) \\
		\midrule
		\textbf{parallelized mSOGIs}~& \\
		observer gains & $\lvec$ as in~\eqref{eq:feedback gain vector l of parallelized mSOGIs} [$\Longrightarrow \;\forall \nu \in \mathbb{H}_{\nu} \colon \; p_{1,2,\nu}^*=-\tfrac{3}{2} \pm \jmath\nu$] \\
		mFLL & $\Gamma = 60$, $\varepsilon = 0.1$ \\
		Anti-windup & $\omega_{\min} = \SI{39}{\radian\per\second}$, $\omega_{\max}=\SI{61}{\radian\per\second}$ \\
		Rate limitation & $\dot{\omega}_{\max} = 2\pi\times\SI{10e3}{\radian\per\second}$, $\dot{\omega}_{\min} = -\dot{\omega}_{\max}$ \\
		\midrule
		\textbf{parallelized sSOGIs}~\cite{2011_Rodriguez_MultiresonantFLLsforGridSynchronizationofPowerConvertersUnderDistortedGridConditions} & \\
		observer gains &  $\lvec = \sqrt{2}\cyvec$ \\
		sFLL & $\Gamma = 46$, $\varepsilon = 0.1$ (avoidance of division by zero added) \\
		\midrule
		\textbf{parallelized ANFs}~\cite{2007_Mojiri_Time-DomainSignalAnalysisUsingAdaptiveNotchFilter} & \\ filter gains & $\lvec=\cyvec$ \\
		sFLL (without gain normalization) & $\gamma = 0.5$\\
		\midrule
		\multicolumn{2}{l}{\textbf{Scenario (S$_1$) with constant fundamental frequency}} \\
		initial values of observer~\eqref{eq:observer_sogi} & $\xsogi_{0}=\ve{0}_{20}$ \\
		initial values of sFLL~\eqref{eq:sFLL} and mFLL~\eqref{eq:mFLL with AWU and rate limiter} & $\omegah_{0}=2 \pi \cdot \SI{50}{\radian\per\second}$ (and $\omegah(t) = \omega(t)$ for all $t \geq 0$) \\
		\midrule
		\multicolumn{2}{l}{\textbf{Scenario (S$_2$) with time-varying fundamental frequency}} \\
		initial values of observer~\eqref{eq:observer_sogi} & $\xsogi_{0}=\ve{0}_{20}$ \\
		initial values of sFLL~\eqref{eq:sFLL} and mFLL~\eqref{eq:mFLL with AWU and rate limiter}  &  $\omegah_{0}=\SI{200}{\radian\per\second} \neq \omega(0) = 2 \pi \cdot \SI{50}{\radian\per\second} $\\ 
		\bottomrule
	\end{tabular}
	\caption{Implementation and tuning data for simulations and measurements.}
	\label{tab:implementation data}
\end{table}

Three estimation methods are implemented and their estimation performances are compared for the following \emph{two scenarios}:
\begin{itemize}	
	\item[(S$_1$)] Estimation of an input signal $y$ with constant fundamental angular frequency (i.e.~$\omegah(t)=\omega(t)=2\pi \SI{50}{\radian\per\second}$ for all $t \geq 0$) and \emph{ten} harmonics exhibiting amplitude jumps at $t_1 = \SI{0.2}{\second}$, $t_2 = \SI{0.4}{\second}$ and $t_3 = \SI{0.6}{\second}$. The following three estimation methods are implemented and compared:
	\begin{itemize}
		\item parallelized mSOGIs \textit{without} FLL, i.e.~observer~\eqref{eq:observer_sogi} with $\mv{l}$ as in~\eqref{eq:feedback gain vector l of parallelized mSOGIs}; 
		\item parallelized sSOGIs \textit{without} FLL~\cite{2011_Rodriguez_MultiresonantFLLsforGridSynchronizationofPowerConvertersUnderDistortedGridConditions}, i.e.~observer~\eqref{eq:observer_sogi} with $\mv{l} = \sqrt{2}\mv{c}$ (and $g_\nu=0$ for all $\nu \in \mathbb{H}_n$); and
		\item parallelized Adaptive Notch Filters (ANFs) \textit{without} FLL~\cite{2007_Mojiri_Time-DomainSignalAnalysisUsingAdaptiveNotchFilter}, i.e.~observer~\eqref{eq:observer_sogi} with $\mv{l} = \mv{c}$.
	\end{itemize}
	\item[(S$_2$)] Estimation of an input signal $y$ with time-varying fundamental angular frequency $\omega(\cdot)$ (i.e.~$\omegah_0 \neq \omega(0)$) and \emph{ten} harmonics exhibiting frequency jumps at $t_1 = \SI{0.2}{\second}$ and $t_3 = \SI{0.6}{\second}$ and amplitude jumps at $t_2 = \SI{0.4}{\second}$ and $t_3 = \SI{0.6}{\second}$. The following three estimation methods are implemented and compared:
	\begin{itemize}
		\item parallelized mSOGIs \textit{with} mFLL, i.e.~\eqref{eq:observer_sogi} with $\mv{l}$ as in~\eqref{eq:feedback gain vector l of parallelized mSOGIs} and~\eqref{eq:mFLL with AWU and rate limiter}; 
		\item parallelized sSOGIs \textit{with} sFLL~\cite{2011_Rodriguez_MultiresonantFLLsforGridSynchronizationofPowerConvertersUnderDistortedGridConditions}, i.e.~observer~\eqref{eq:observer_sogi} with $\mv{l} = \sqrt{2}\mv{c}$ ($g_\nu=0$ for all $\nu \in \mathbb{H}_n$) and~\eqref{eq:sFLL}; and
		\item parallelized ANFs \textit{with} sFLL~\cite{2007_Mojiri_Time-DomainSignalAnalysisUsingAdaptiveNotchFilter}, i.e.~observer~\eqref{eq:observer_sogi} with $\mv{l} = \mv{c}$ and sFLL~\eqref{eq:fll_adaption_rule} (without gain normalization\footnote{In \cite{2007_Mojiri_Time-DomainSignalAnalysisUsingAdaptiveNotchFilter}, the gain is '$\gamma$' and set to $80$; the maximal amplitude is $a_1 = 1$. Because the Grid Emulator cannot produce such low voltages and the FLL is driven without a GN, the FLL gain for the ANFs is optimized here with respect to the used amplitudes.}).
	\end{itemize}
\end{itemize}
For both scenarios, the considered input signals $y$ have a significant harmonic distortion. The parameters of the individual harmonics (amplitudes $a_\nu$ and frequencies $f$)  of the signal $y$ for Scenario (S$_1$) and for Scenario (S$_2$) are collected in Tab.~\ref{tab:parameters of harmonics}. Within the considered time interval $[0, \, \SI{0.8}{\second}]$, three jump-like changes in amplitudes and/or fundamental frequency occur at  $\SI{0.2}{\second}$, $\SI{0.4}{\second}$ and $\SI{0.6}{\second}$. Hence, the input signal and its harmonic content changes abruptly and requires the observers to ``restart'' their estimation process for each step-like change.\\

For a fair comparison, all three estimation methods are tuned in such a way that the best feasible estimation performance was achieved within their respective tuning and capability limits. The measurement results for Scenario (S$_1$) are shown in the Figures~\ref{fig:Measurement results for Scenario S1_yhat_ey_fhat}, \ref{fig:Measurement results for Scenario S1_y_nu_e_nu} and \ref{fig:Measurement results for Scenario S1_y_nu_e_nu_ZOOM}. The results for Scenario (S$_2$) are depicted in Figures~\ref{fig:Measurement results for Scenario S2_yhat_ey_fhat}, \ref{fig:Measurement results for Scenario S2_y_nu_e_nu} and \ref{fig:Measurement results for Scenario S2_y_nu_e_nu_ZOOM}.
These results will be discussed in more detail in the next subsections.

\begin{table}[t!]
  \centering
    \includegraphics{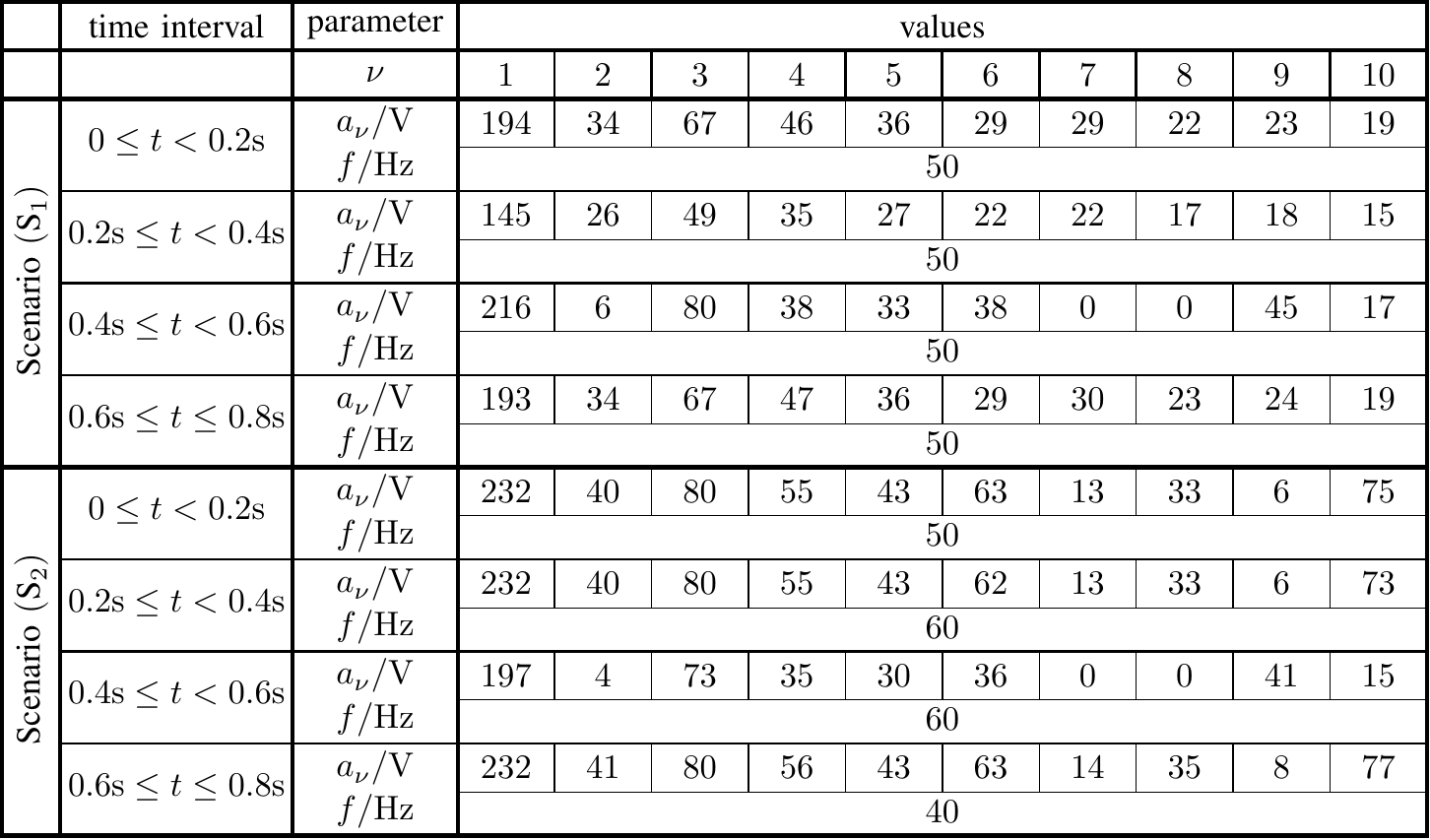}
\caption{Data of considered input signal $y$ for both scenarios (S$_1$) and (S$_2$):~Amplitudes and frequencies ($f = \tfrac{\omega}{2\pi}$) of the ten harmonics.}
\label{tab:parameters of harmonics}
\end{table}

\subsection{Discussion of the measurement results obtained for Scenario (S$_1$)}
For Scenario (S$_1$), the FLL were implemented but the adaption was turned off. Fundamental and estimated angular frequency are identical for this scenario. Therefore, the estimation performances purely according to the respective observer tunings can be compared. The harmonic content of the input signal $y$ undergoes step-like changes at the time instants $\SI{0.2}{\second}$, $\SI{0.4}{\second}$ and $\SI{0.6}{\second}$, respectively (see Tab.~\ref{tab:parameters of harmonics} and Fig.~\ref{fig:Measurement results for Scenario S1_yhat_ey_fhat}). \\

\begin{figure} 
	\centering
	\includegraphics[clip,width=\textwidth]{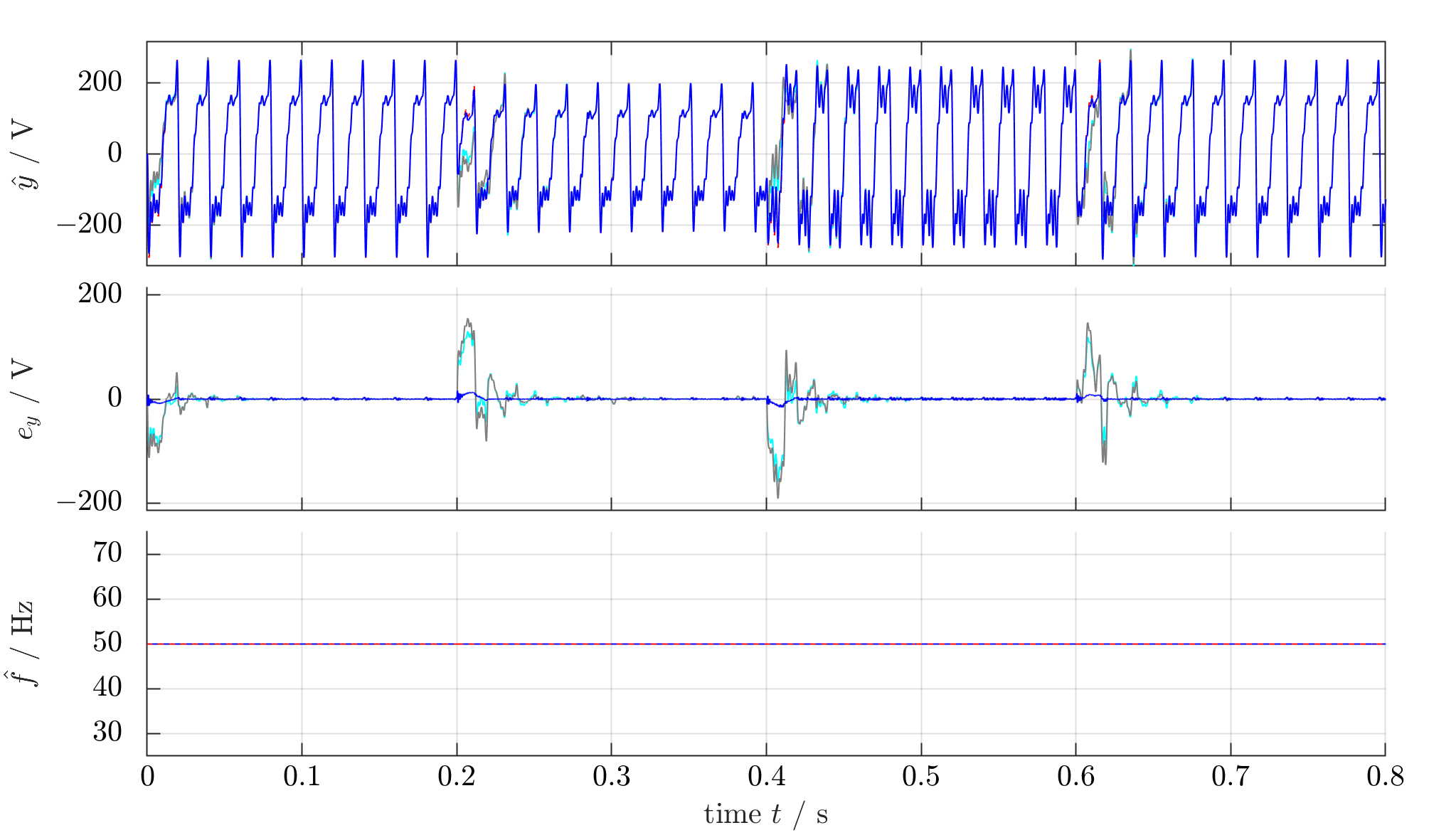}
	\caption{\textbf{Measurement results for Scenario (S$_1$)}:~Comparison of the estimation performances of parallelized mSOGIs ({\protect\blueline}), sSOGIs ({\protect\cyanline}) and ANFs ({\protect\grayline}) \emph{without} FLL. Signals shown from top to bottom are:~Input signal $y$ ({\protect\reddashedline}) \& its estimate $\widehat{y}$; estimation error $e_y = y - \widehat{y}$;  frequency $f=\tfrac{\omega}{2\pi}$ ({\protect\reddashedline}) \& its estimate $\widehat{f}= \tfrac{\widehat{\omega}}{2\pi}$.}
	\label{fig:Measurement results for Scenario S1_yhat_ey_fhat}
\end{figure}

Three measurement plots are presented in Figures~\ref{fig:Measurement results for Scenario S1_yhat_ey_fhat}, \ref{fig:Measurement results for Scenario S1_y_nu_e_nu} and \ref{fig:Measurement results for Scenario S1_y_nu_e_nu_ZOOM}. The overall estimation performances of the parallelized mSOGIs (\blueline), sSOGIs (\cyanline) and ANFs (\grayline) are depicted in Fig.~\ref{fig:Measurement results for Scenario S1_yhat_ey_fhat}: The first, second and third subplots show input signal $y$ (\reddashedline) \& its estimates $\widehat{y}$, the estimation errors $e_y = y - \widehat{y}$ and fundamental frequency $f$ (\reddashedline) \& its estimate $\widehat{f} = \tfrac{\widehat{\omega}}{2 \pi}$, respectively. All three observers are capable of estimating the input signal $y$. All estimation errors $e_y \to 0$ tend to zero after a certain time. The parallelized mSOGIs (\blueline) clearly outperform the other two estimation methods in estimation accuracy and estimation speed for all three step-like changes of the input signal $y$ at $\SI{0.2}{\second}$, $\SI{0.4}{\second}$ and $\SI{0.6}{\second}$. Estimation is completed in less than $\SI{20}{\milli\second}$. This is possibly due to the newly introduced gains $g_\nu$ for all $\nu \in \mathbb{H}_n$ which give the necessary degrees of freedom in observer design (recall discussion in Sect.~\ref{sec:Estimation performance of sSOGI and mSOGI}).\\

\begin{figure}
	\centering
	\begin{tabular}{cc}
		\includegraphics[clip,width=0.5\textwidth]{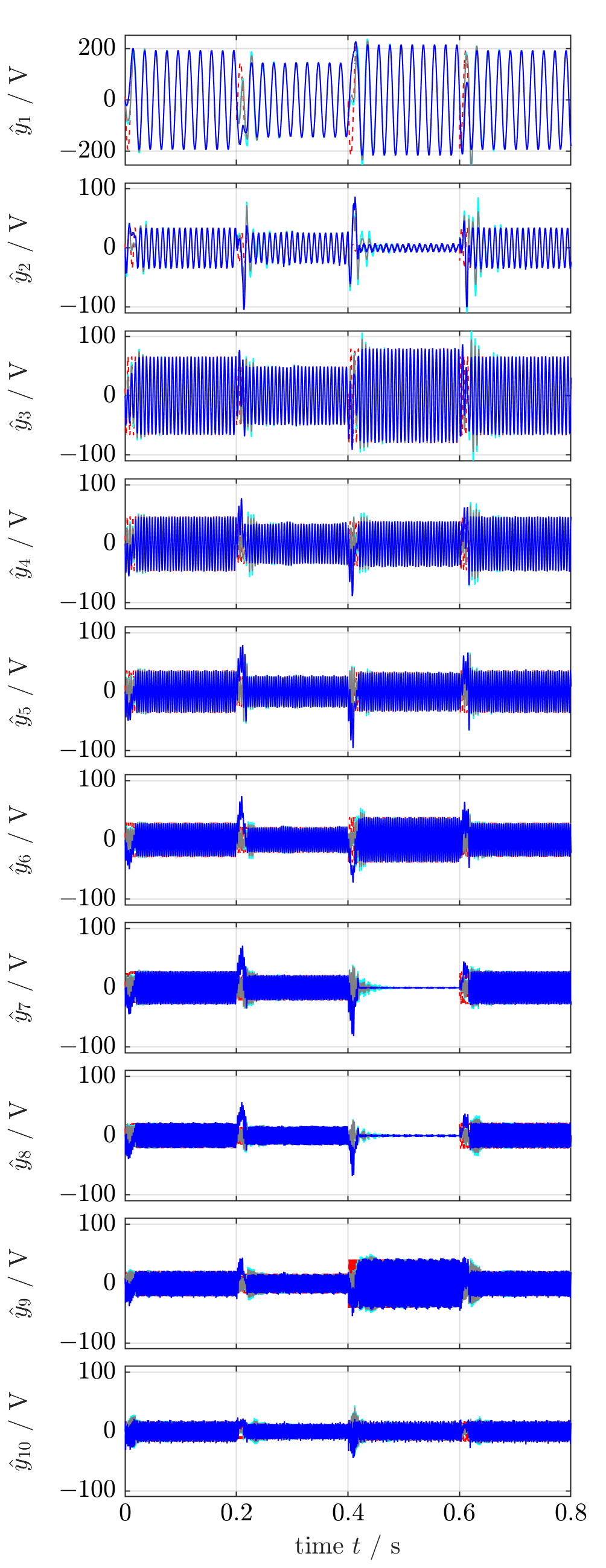}
		& 
		\includegraphics[clip,width=0.5\textwidth]{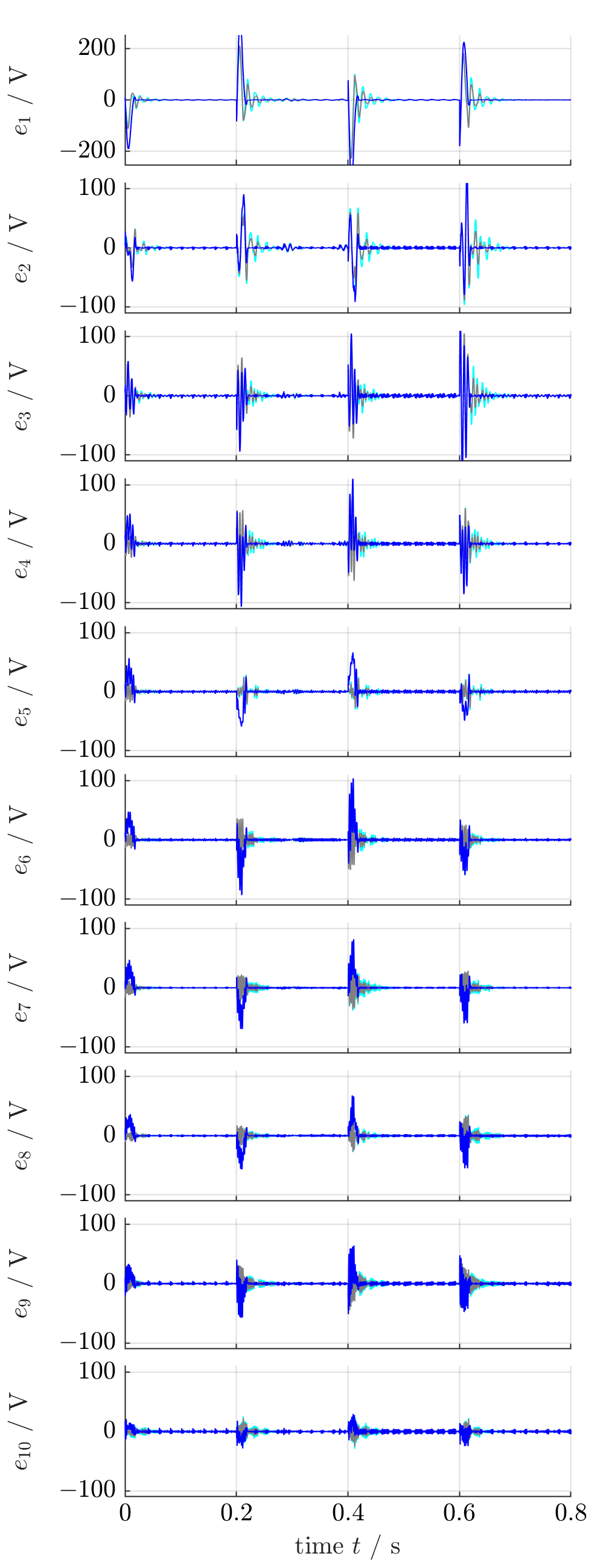}
	\end{tabular}
	\caption{\textbf{Measurement results for Scenario (S$_1$)}:~Comparison of the estimation performances of parallelized mSOGIs ({\protect\blueline}), sSOGIs ({\protect\cyanline}) and ANFs ({\protect\grayline}) \emph{without} FLL. Signals shown from top to bottom are:~Harmonic signals $y_1$ to $y_{10}$ ({\protect\reddashedline}) \& their estimates $\widehat{y}_1$ to $\widehat{y}_{10}$ (left) and harmonic estimation errors $e_1 = y_1 - \widehat{y}_1$ to $e_{10} = y_{10} - \widehat{y}_{10}$ (right).}
	\label{fig:Measurement results for Scenario S1_y_nu_e_nu}
\end{figure}

In Figures~\ref{fig:Measurement results for Scenario S1_y_nu_e_nu} and~\ref{fig:Measurement results for Scenario S1_y_nu_e_nu_ZOOM}, the estimation performances for the ten individual harmonics are illustrated for the complete time interval $[0, \, \SI{0.8}{\second}]$ of Scenario (S$_1$) (see Fig.~\ref{fig:Measurement results for Scenario S1_y_nu_e_nu}) and for the shorter interval $[\SI{0.6}{\second}, \, \SI{0.8}{\second}]$ (see Zoom in Fig.~\ref{fig:Measurement results for Scenario S1_y_nu_e_nu_ZOOM}), respectively. In both figures, on the left hand side, the harmonics $y_1$ to $y_{10}$ (\reddashedline) and theirs estimates $\widehat{y}_1$ to $\widehat{y}_{10}$ are shown; whereas on the right hand side, the estimation errors $e_1 := y_1 -\widehat{y}_1$ to  $e_{10} := y_{10} -\widehat{y}_{10}$ are depicted. Again, all three estimation methods are capable of tracking the respective harmonic components after a certain time:~Amplitudes and phases are estimated correctly with asymptotically vanishing estimation errors. But also for the individual harmonic estimation, the parallelized mSOGIs (\blueline) achieve a much faster estimation performance than the parallelized sSOGIs (\cyanline) and the ANFs (\grayline). In particular for the lower harmonics (such as $\nu_1 = 1$, $\nu_2 = 2$, $\nu_3 = 3$ and $\nu_4 = 4$), the estimation is three to four times faster than that of the other two methods. 
\begin{figure}
	\centering
	\begin{tabular}{cc}
		\includegraphics[clip,width=0.5\textwidth]{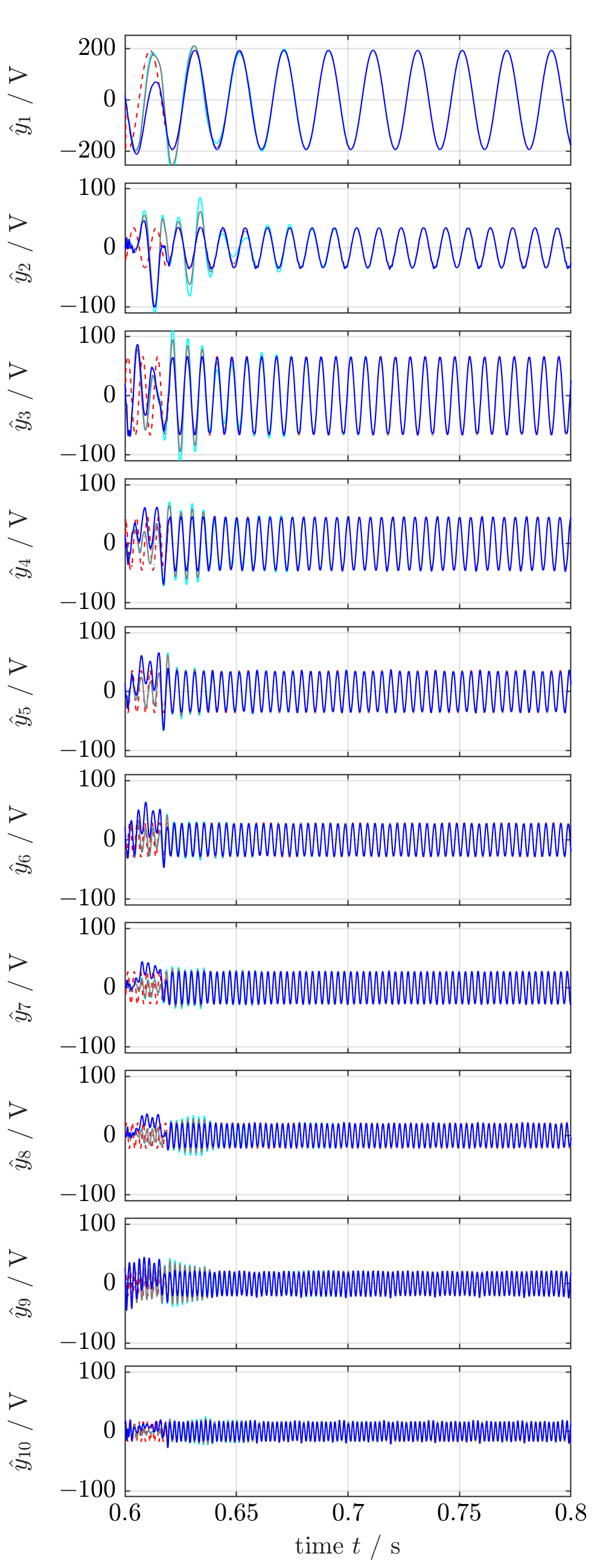}
		& 
		\includegraphics[clip,width=0.5\textwidth]{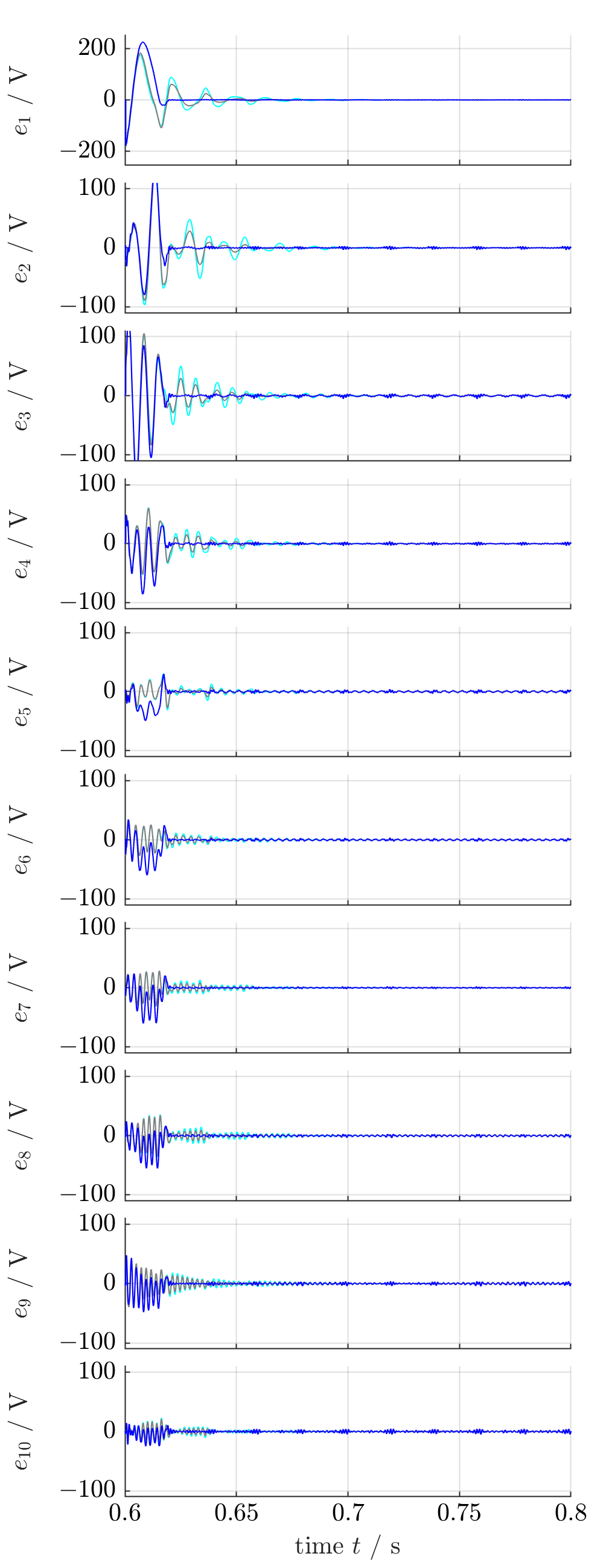}
	\end{tabular}
	\caption{\textbf{Measurement results for Scenario (S$_1$) -- Zoom of time interval $[0.6\si{\second}, 0.8\si{\second}]$}:~Comparison of the estimation performances of parallelized mSOGIs ({\protect\blueline}), sSOGIs ({\protect\cyanline}) and ANFs ({\protect\grayline}) \emph{without} FLL. Signals shown from top to bottom are:~Harmonic signals $y_1$ to $y_{10}$ ({\protect\reddashedline}) \& their estimates $\widehat{y}_1$ to $\widehat{y}_{10}$ (left) and harmonic estimation errors $e_1 = y_1 - \widehat{y}_1$ to $e_{10} = y_{10} - \widehat{y}_{10}$ (right).}
	\label{fig:Measurement results for Scenario S1_y_nu_e_nu_ZOOM}
\end{figure}

\subsection{Discussion of the measurement results obtained for Scenario (S$_2$)}
Scenario (S$_2$) is more challenging. Now, amplitudes and frequency of the input signal $y$ are time-varying. At time instants $\SI{0.2}{\second}$ and $\SI{0.6}{\second}$, the fundamental frequency jumps from $\SI{50}{\hertz}$ to $\SI{60}{\hertz}$ and from $\SI{60}{\hertz}$ to $\SI{40}{\hertz}$, respectively; whereas amplitudes and phases of the harmonic components change abruptly at $\SI{0.4}{\second}$ and $\SI{0.6}{\second}$, respectively (see Tab.~\ref{tab:parameters of harmonics} and Fig.~\ref{fig:Measurement results for Scenario S1_yhat_ey_fhat}). The measurement results for Scenario (S$_2$) are plotted in Figures~\ref{fig:Measurement results for Scenario S1_yhat_ey_fhat}, \ref{fig:Measurement results for Scenario S1_y_nu_e_nu} and \ref{fig:Measurement results for Scenario S1_y_nu_e_nu_ZOOM}. These figures show the identical quantities as those shown for Scenario (S$_1$). \\

\begin{figure} 
	\centering
	\includegraphics[clip,width=\textwidth]{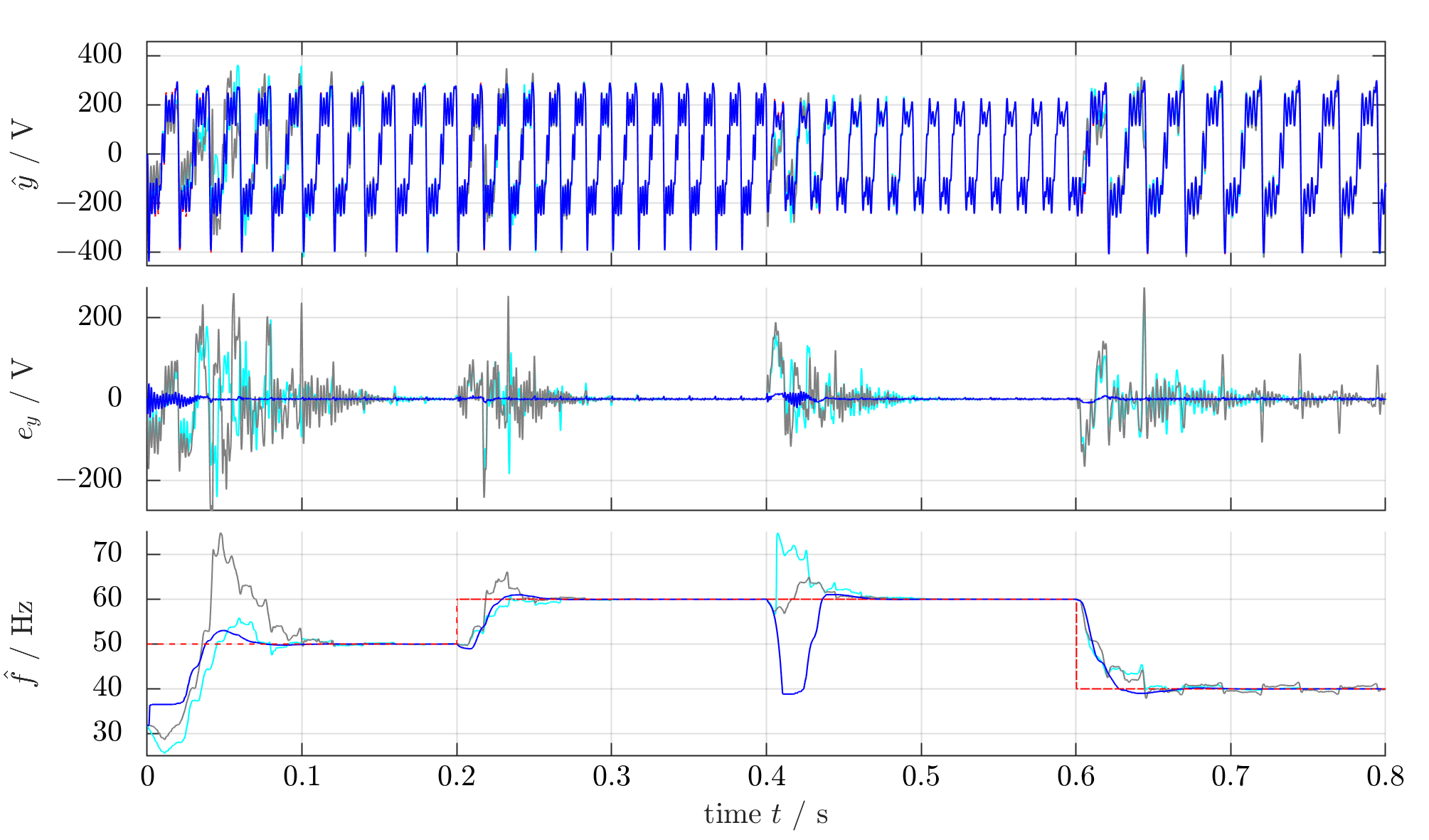}
	\caption{\textbf{Measurement results for Scenario (S$_2$)}:~\textbf{Measurement results for Scenario (S$_1$)}:~Comparison of the estimation performances of parallelized mSOGIs ({\protect\blueline}), sSOGIs ({\protect\cyanline}) and ANFs ({\protect\grayline}) \emph{with} FLL. Signals shown from top to bottom are:~Input signal $y$ ({\protect\reddashedline}) \& its estimate $\widehat{y}$; estimation error $e_y = y - \widehat{y}$;  frequency $f=\tfrac{\omega}{2\pi}$ ({\protect\reddashedline}) \& its estimate $\widehat{f}= \tfrac{\widehat{\omega}}{2\pi}$.}
	\label{fig:Measurement results for Scenario S2_yhat_ey_fhat}
\end{figure}

In Fig.~\ref{fig:Measurement results for Scenario S1_yhat_ey_fhat}, estimated signals $\widehat{y}$, estimation errors $e_y := y -\widehat{y}$ and estimated frequencies $\widehat{f}$ are shown for the parallelized mSOGIs (\blueline),  sSOGIs (\cyanline) and ANFs (\grayline), respectively. All three estimation methods are able to correctly estimate the fundamental frequency asymptotically. But, for the parallelized mSOGIs, estimation accuracy and estimation speed of the proposed mFLL are better and the estimation process is much smoother and exhibits less oscillations. Please note that the dip in the frequency estimation of the mFLL after $\SI{0.4}{\second}$ does \emph{not} endanger stability of the parallelized mSOGIs (which is due to anti-windup and rate limitation). The overall estimation accuracy of the mSOGIs is very convincing as can be seen in $e_y$. The estimation error tends to zero within $20-\SI{40}{\milli\second}$ after all three input changes at $\SI{0.2}{\second}$, $\SI{0.4}{\second}$ and $\SI{0.6}{\second}$ and remains close to zero afterwards. In contrast to that,  the overall estimation accuracy and estimation speed of sSOGIs (\cyanline) and ANFs (\grayline) are rather bad and slow:~Rapid changes in $e_y$ occur for more than $\SI{100}{\milli\second}$ after each change. Within the last interval $[\SI{0.6}{\second}, \, \SI{0.8}{\second}]$, the estimation error of both methods does not even tend to zero within $\SI{200}{\milli\second}$.\\

Similar observations can be made by comparing the individual harmonic estimation performances of the three estimation methods as shown in Fig.~\ref{fig:Measurement results for Scenario S1_y_nu_e_nu} for the overall time interval $[\SI{0}{\second}, \, \SI{0.8}{\second}]$ of Scenario (S$_2$) and in Fig.~\ref{fig:Measurement results for Scenario S1_y_nu_e_nu_ZOOM} for the zoomed time interval  $[\SI{0.6}{\second}, \, \SI{0.8}{\second}]$. Despite the rather bad input estimation performance of parallelized sSOGIs (\cyanline) and ANFs (\grayline), their harmonics estimation accuracy is acceptable:~All harmonic amplitudes and angles are estimated correctly after some time. However, also here the estimation speed of the parallelized mSOGIs (\blueline) is faster for all harmonic components (see $e_1$ to $e_{10}$ in Fig.~\ref{fig:Measurement results for Scenario S1_y_nu_e_nu}). However, the difference in estimation speed is not as significant as it was for Scenario (S$_1$), which shows that the FLL remains the weakest component of the grid estimation process and has to be improved further (future work). \\

The last measurement plots depicted in Fig.~\ref{fig:Measurement results for Scenario S2_y_nu_e_nu_ZOOM} show the zoomed version of the harmonics estimation during the shorter time interval $[\SI{0.6}{\second}, \, \SI{0.8}{\second}]$. Solely, the estimation performance of the parallelized mSOGIs (\blueline) is still acceptable. The estimation accuracies of parallelized sSOGIs (\cyanline) and ANFs (\grayline) exhibit significant oscillations and do not tend to zero (in particular for higher harmonics). Their estimation performances are clearly not acceptable anymore. In conclusion, the measurement results obtained for both scenarios have verified the improved performance of the proposed parallelized mSOGIs (\blueline) with mFLL compared to the slower and less accurate estimation performances of parallelized sSOGI (\cyanline) and ANFs (\grayline), respectively.

\begin{figure}[t!]
  \centering
\begin{tabular}{cc}
\includegraphics[clip,width=0.5\textwidth]{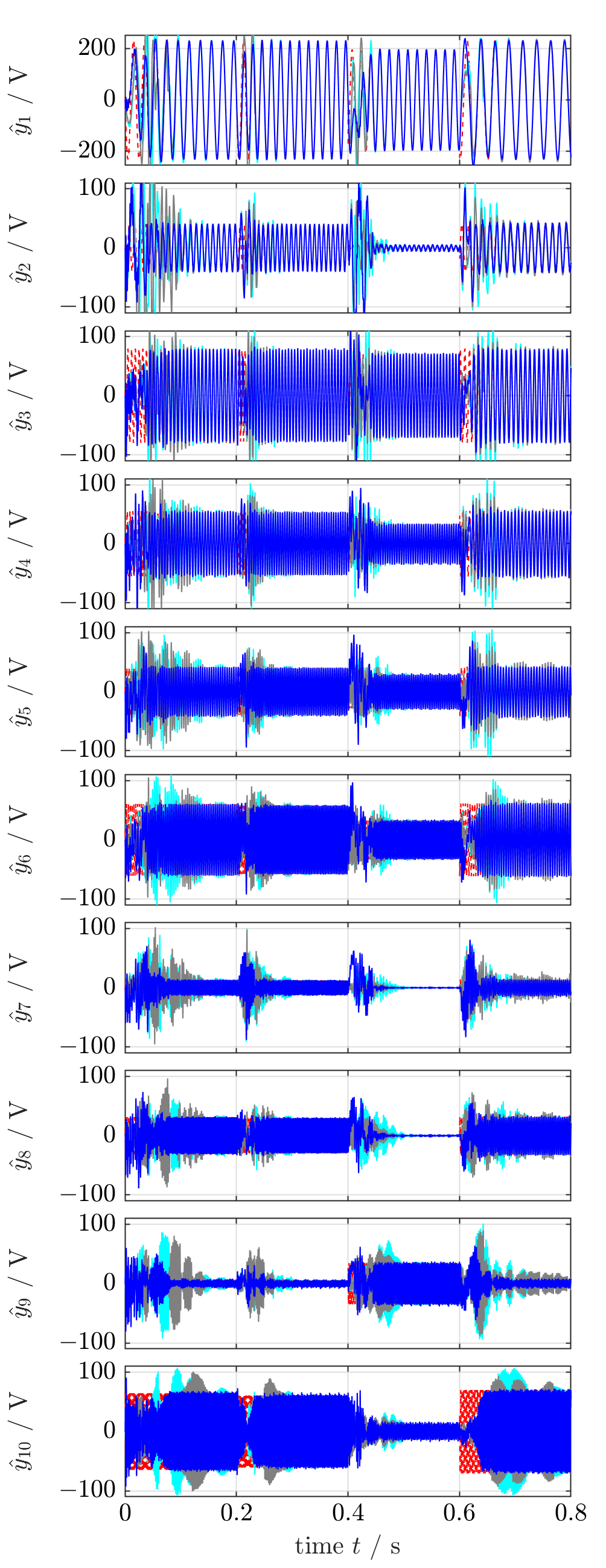}
& 
\includegraphics[clip,width=0.5\textwidth]{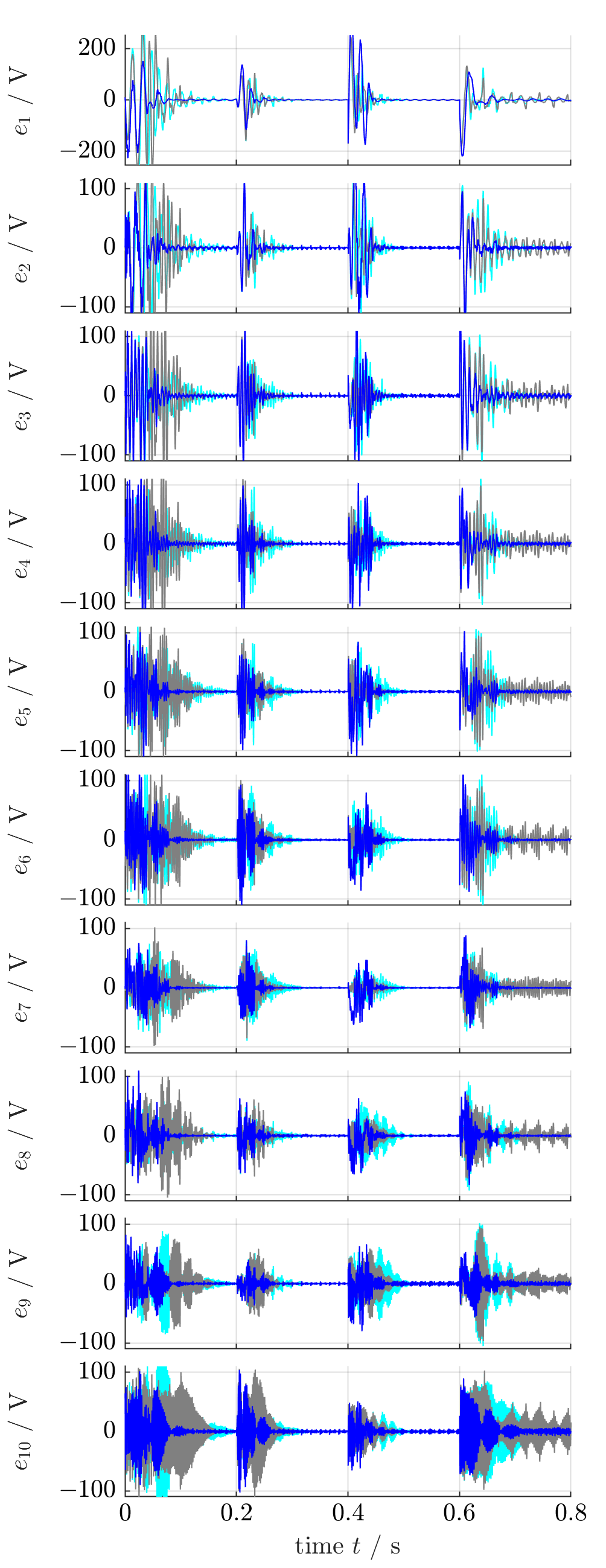}
\end{tabular}
\caption{\textbf{Measurement results for Scenario (S$_2$)}:~Comparison of the estimation performances of parallelized mSOGIs ({\protect\blueline}), sSOGIs ({\protect\cyanline}) and ANFs ({\protect\grayline}) \emph{with} FLL. Signals shown from top to bottom are:~Harmonic signals $y_1$ to $y_{10}$ ({\protect\reddashedline}) \& their estimates $\widehat{y}_1$ to $\widehat{y}_{10}$ (left) and harmonic estimation errors $e_1 = y_1 - \widehat{y}_1$ to $e_{10} = y_{10} - \widehat{y}_{10}$ (right).}
\label{fig:Measurement results for Scenario S2_y_nu_e_nu}
\end{figure}
\begin{figure}[t!]
	\centering
	\begin{tabular}{cc}
		\includegraphics[clip,width=0.5\textwidth]{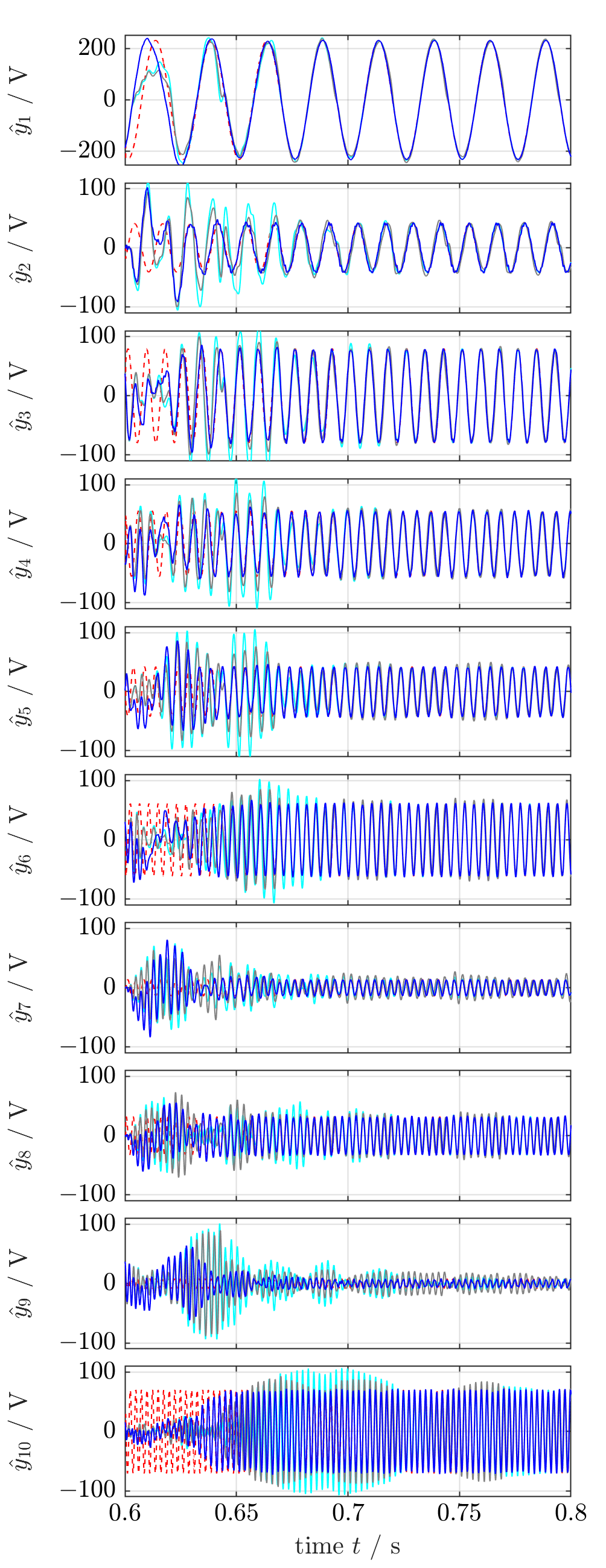}
		& 
		\includegraphics[clip,width=0.5\textwidth]{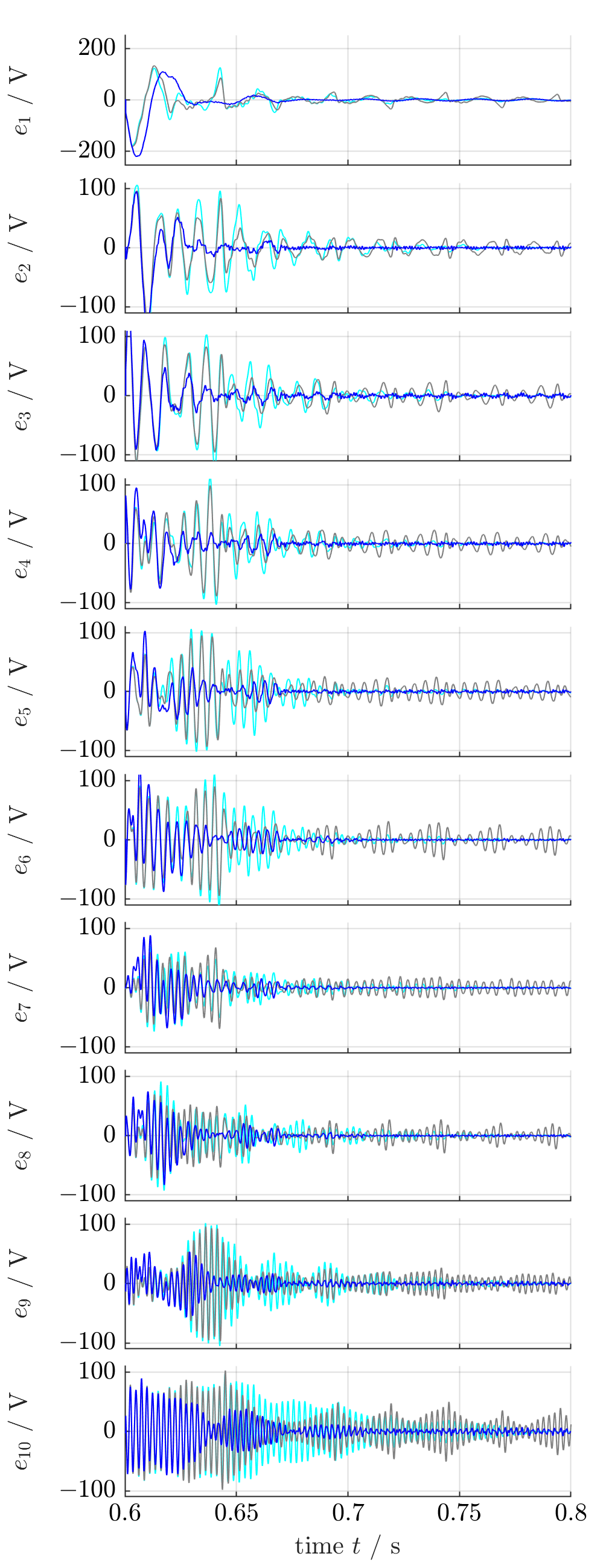}
	\end{tabular}
	\caption{\textbf{Measurement results for Scenario (S$_2$) -- Zoom of time interval $[0.6\si{\second}, 0.8\si{\second}]$}:~Comparison of the estimation performances of parallelized mSOGIs ({\protect\blueline}), sSOGIs ({\protect\cyanline}) and ANFs ({\protect\grayline}) \emph{with} FLL. Signals shown from top to bottom are:~Harmonic signals $y_1$ to $y_{10}$ ({\protect\reddashedline}) \& their estimates $\widehat{y}_1$ to $\widehat{y}_{10}$ (left) and harmonic estimation errors $e_1 = y_1 - \widehat{y}_1$ to $e_{10} = y_{10} - \widehat{y}_{10}$ (right).}
	\label{fig:Measurement results for Scenario S2_y_nu_e_nu_ZOOM}
\end{figure}
%


\section{Conclusion}
\label{sec_Conclusion_and_Outlook}
A modified Second-Order Generalized Integrator (mSOGI) for the $\nu$-th harmonic component and a modified Frequency Locked Loop (mFLL) for the parallelized mSOGIs have been proposed. The number $\nu$ can represent any positive not necessarily natural harmonic of an arbitrarily deteriorated input signal for which fundamental and higher harmonic components shall be estimated in real time.  In contrast to the $\nu$-th standard SOGI (sSOGI) in literature, the $\nu$-th mSOGI allows (theoretically) for an arbitrarily fast estimation of the in-phase and quadrature signal of any specified harmonic component with prescribed settling time.  This is possible due to an additionally introduced feedback gain in the mSOGI design. The proposed mFLL is equipped with sign-correct anti-windup strategy and rate limitation. These modifications enhance the frequency estimation in such a way that the frequency estimate remains positive and bounded and does not change too quickly (independently of mFLL tuning or operating point). Moreover, both enhancements overcome the stability problem of the estimator when a standard FLL (sFLL) is used. Measurement results illustrate and verify the improved estimation performance of the parallelized mSOGIs with and without mFLL in comparison to parallelized sSOGIs and ANFs with and without sFLL.

Future work will focus on (i) further improvements of the mFLL (acceleration of frequency estimation and global stability analysis), (ii) the extension of the presented results to three-phase signals (including DC offsets) and (iii) the real-time estimation of positive, negative and zero sequences of each harmonic component.


\clearpage
\appendix
\section{Appendix}

In the following appendices, observability and stability of the parallelized SOGIs, the pole placement algorithm and the generalization of the adaption law of the modified FLL for the parallelized mSOGIs are discussed in more detail.


\subsection{Recapitulation}

Recall that, any exogenous sinusoidal signal of the form $$y(t) := \sum\limits_{\nu \in \mathbb{H}_n} \underbrace{a_{\nu}(t)\cos\big(\phi_{\nu}(t)\big)}_{=: y_\nu(t)}  \quad \text{ where } \quad \mathbb{H}_n:= \{1,\nu_2, \dots,\nu_n\}, \quad \nu_i \neq \nu_j \quad \text{ for all } \quad i \neq j \in \bc{1, \ldots, n}$$ can be reduplicated by the parallelization of $n$ sinusoidal internal models~\cite[Chapter~20]{2017_Hackl_Non-identifierbasedadaptivecontrolinmechatronics:TheoryandApplication}.  The overall internal model dynamics are given by
\begin{equation}
\left.
\begin{array}{rcl}
\ddtsmall \mv{x}(t)  &=&  \omega(t)\JMAT\ve{x}(t), \qquad \qquad \qquad \qquad  \ve{x}\br{0} = \ve{x}_{0} \neq \mv{0}_{2n} \in \R^{2n} \vspace*{1ex}\\
y(t)  & = & \underbrace{(1,\, 0,\, 1,\, 0,\, \cdots,\, 1,\, 0)}_{=: \,\cyvec^{\top} \in \R^{2n}}\ve{x}(t)
\end{array} \qquad 
\right\}
\label{eq:state space dynamics of overall IM - appendix}
\end{equation}
where $\omega(\cdot) \in \Czero^{\mathrm{pw}} \cap \Linf(\Rzp;[\epsilon_{\omega},\infty))$ (with $\epsilon_{\omega} > 0$) and 
\begin{equation}
\mv{x}:=\big((\underbrace{x_1^{\alpha},\, x_1^{\beta})}_{=: \ve{x}_1^{\top}},\, \ldots,\, \ve{x}_n^{\top}\big)^{\top}, \; \JMAT := \blockdiag\br{\Jbar, \nu_2\Jbar, \cdots, \nu_n\Jbar} \in \R^{2n \times 2n} \; \text{ and } \; \Jbar = \begin{bmatrix} 0 & -1 \\ 1 & 0 \end{bmatrix} =  -\Jbar^\top = - \Jbar^{-1}.
\label{eq:definitions of x, J and Jbar - appendix}
\end{equation}
The initial values of the internal model in~\eqref{eq:state space dynamics of overall IM - appendix} allow to determine amplitude $a_\nu$ and angle $\phi_{\nu}$ of the $\nu$-th harmonic component. \\

The overall observer (estimator) consists of the parallelized mSOGIs, i.e.~the parallelized internal model~\eqref{eq:state space dynamics of overall IM - appendix} with feedback of the input signal $y(\cdot)$ and using the estimated angular frequency $\widehat{\omega}(\cdot)$ instead of $\omega(\cdot)$. Its dynamics are given by
\begin{equation}
\left.\begin{array}{rcl}
\ddtsmall\xsogi(t) &=& \widehat{\omega}(t)\big[\overbrace{\JMAT - \mv{l}\mv{c}^\top}^{=:\mm{A}}\big]\xsogi(t) + \widehat{\omega}(t) \,\mv{l}\, y(t) , \qquad \qquad \widehat{\mv{x}}(t) = \widehat{\mv{x}}_0 \in \R^{2n} \vspace*{2ex}\\
\yhout(t) &=& \cyvec^{\top}\xsogi(t),
\end{array}\right\}
\label{eq:observer_sogi - appendix}
\end{equation}
with observer state vector $\widehat{\mv{x}} := \big(\widehat{\mv{x}}_1^\top, \widehat{\mv{x}}_{\nu_2}^\top, \cdots,  \widehat{\mv{x}}_{\nu_n}^{\top} \big)^\top \in \R^{2n}$ and observer gain vector 
$
\mv{l} := \big(\mv{l}_1^\top, \mv{l}_{\nu_2}^\top, \cdots,  \mv{l}_{\nu_n}^{\top} \big)^\top \in \R^{2n}.
$

\subsection{Observability of the parallelized internal models~\eqref{eq:state space dynamics of overall IM - appendix} (for constant $\widehat{\omega}$)}
The grid state estimation algorithm is based on the idea of observability. The input signal can be considered to be generated by a parallelization of internal models which are individually capable of reduplicating a sinusoidal signal each (representing one harmonic component each). Hence, if this system is observable, an observer can be designed for grid state estimation.

\begin{proposition}[Observability of the linear generating system]
\label{prop:observability of parallelized IMs}
For constant $\omega>0$ and differing harmonics, i.e.~$\nu_i \neq \nu_j$ for all $i\neq j \in \{1,\dots n\}$, generating system~\eqref{eq:state space dynamics of overall IM - appendix} is completely state observable.
\end{proposition}
\begin{proof}
Note that the following two identities hold
$$
\forall k \in \N \colon \qquad \JMAT^{k} = \blockdiag\br{\Jbar^{k}, \nu_2^{k}\Jbar^{k}, \cdots, \nu_n^{k}\Jbar^{k}} \in \R^{2n \times 2n} \quad \text{and} \quad \Jbar^{k} = - \Jbar^{k - 2} \in \R^{2\times2},
$$
which imply
\begin{IEEEeqnarray}{rCl}
\rank{\scriptsize\begin{bmatrix} 
	\cyvec^{\top} \\ 
	\cyvec^{\top}\JMAT \\ 
	\cyvec^{\top}\JMAT^2 \\
	\cyvec^{\top}\JMAT^3 \\
	\vdots \\ 
		\cyvec^{\top}\JMAT^{2n-2}\\ 
	\cyvec^{\top}\JMAT^{2n-1} 
\end{bmatrix}} & = &
 \rank {\scriptsize \begin{bmatrix} 
1 & 0  & 1 & 0 & \cdots & 1 & 0 \\ 
0 & -1 & 0 & -\nu_2 & \cdots & 0 & -\nu_n \\
-1 & 0 & -\nu_2^2 & 0& \cdots & -\nu_n^2 & 0\\
0 & 1 & 0 & \nu_2^{3} & \cdots & 0 & \nu_n^{3} \\
\vdots & & & \vdots & & & \vdots \\ 
(-1)^{n-1} & 0 & (-1)^{n-1}\nu_2^{2n-2} & 0& \cdots & (-1)^{n-1}\nu_n^{2n-2} & 0\\
0 & (-1)^n & 0 & (-1)^n\nu_2^{2n-1} & \cdots & 0 & (-1)^n\nu_n^{2n-1} \\
\end{bmatrix} } = 2n.
\end{IEEEeqnarray}
Hence, the pair $(\mm{J},\mv{c}^\top)$ is observable~\cite[Corollary~12.3.19]{2009_Bernstein_MatrixMathematics---TheoryFactsandFormulaswithApplicationtoLinearSystemTheory}.
\end{proof}

\subsection{Bounded-input bounded-state/bounded-output stability of the nonlinear observer}

As first step, it is shown that for any essentially bounded input signal $y(\cdot)$ and any essentially bounded and strictly positive estimated angular frequency $\widehat{\omega}(\cdot)$ the parallelized mSOGIs are bounded-input bounded-state/bounded-output state stable. In other words, the estimated states $\widehat{\mv{x}}(\cdot)$ and the estimated output $\widehat{y}(\cdot)$ will \emph{not} diverge.

\begin{theorem}[Bounded-input bounded-state/bounded-output stability of the dynamics of the parallelized SOGIs]
 \mbox{}\label{thm:BIBO/S stability}
Consider an essentially bounded input signal $y(\cdot) \in \Linf(\Rzpos;\R)$ and assume that (i) the estimated time-varying fundamental angular frequency is continuous, bounded and uniformly bounded away from zero, i.e.~$\widehat{\omega}(\cdot) \in \Czero(\Rzpos;\Rpos) \cap \Linf(\Rzpos;\Rpos)$ with $\widehat{\omega}(t) \geq \epsilon_{\omega}  > 0$ for all $t\geq 0$, and (ii) the matrix $\mm{A}=\J - \mv{l}\mv{c}^\top$ in~\eqref{eq:observer_sogi - appendix} is a Hurwitz matrix. Then, the time-varying system~\eqref{eq:observer_sogi - appendix} is bounded-input bounded-state/bounded-output stable, i.e.~
$$
\forall t \geq 0\; \exists\, c_x, \, c_y > 0 \colon \qquad \norm{\widehat{\mv{x}}(t)} \leq c_x \quad \text{ and } \quad |\widehat{y}(t)| \leq c_y. 
$$
\end{theorem}

\begin{proof}
First note that, since $\mm{A}$ is Hurwitz, there exists $\mm{P} = \mm{P}^\top >0$ such that, for any given $\mm{Q} = \mm{Q}^\top>0$, the following identity holds~\cite[Corollary~3.3.47]{2005_Hinrichsen_MathematicalSystemsTheoryI---ModellingStateSpaceAnalysisStabilityandRobustness}

\begin{equation}
 \mm{A}^\top \mm{P}  + \mm{P} \mm{A} = - \mm{Q}.
 \label{eq:Lyapunov identity}
\end{equation}

Moreover, note that 

\begin{equation}
 \forall \, a,b \in \R\; \forall\,m >0\colon \qquad 2\, a\, b = \tfrac{a^2}{m} + m b^2 - \big( \tfrac{a}{\sqrt{m}} - \sqrt{m}\,b \big)^2 \leq \tfrac{a^2}{m} + m b^2.
 \label{eq:2ab<=...}
\end{equation}

Next, introduce the non-negative Lyapunov-like function
$$
  V\colon \R^{2n} \to \Rzpos, \qquad \widehat{\mv{x}} \mapsto V(\widehat{\mv{x}}) := \widehat{\mv{x}}^\top \mm{P} \widehat{\mv{x}}
$$
and denote minimal and maximal eigenvalue of $\mm{P}$ by $\lambda_{\min}(\mm{P})$ and $\lambda_{\max}(\mm{P})$, respectively. Then, clearly, the following holds 

\begin{equation}
 \forall \, \widehat{\mv{x}} \in \R^{2n}\colon \qquad \lambda_{\min}(\mm{P})\norm{\widehat{\mv{x}}}^2 \leq V(\widehat{\mv{x}}) \leq \lambda_{\max}(\mm{P})\norm{\widehat{\mv{x}}}^2 \quad \Longrightarrow \quad - \norm{\widehat{\mv{x}}}^2 \leq -\tfrac{1}{\lambda_{\max}(\mm{P})}V(\widehat{\mv{x}}). 
 \label{eq:Inequalities for Lyapunov function}
\end{equation}

The right-hand side of~\eqref{eq:observer_sogi - appendix} is locally Lipschitz continuous with bounded Lipschitz constant and bounded exogenous perturbation. Hence, the solution of~\eqref{eq:observer_sogi - appendix} exists globally on $\Rzpos$~\cite[Theorem~2.2.14 \& Proposition~2.2.19]{2005_Hinrichsen_MathematicalSystemsTheoryI---ModellingStateSpaceAnalysisStabilityandRobustness} (but still might diverge as $t \to \infty$).
The time derivative of $V(\cdot)$ along the solution of~\eqref{eq:observer_sogi - appendix} is, for all $t \geq 0$, given and upper bounded by 

\begin{eqnarray}
\ddtsmall V\big(\widehat{\mv{x}}(t)\big) & = & \ddtsmall \widehat{\mv{x}}(t)^\top\mm{P} \widehat{\mv{x}}(t) +  \widehat{\mv{x}}(t)^\top\mm{P} \ddtsmall \widehat{\mv{x}}(t) \notag \\
	  & \stackrel{\eqref{eq:observer_sogi - appendix}}{=} & \widehat{\omega}(t) \Big[ \widehat{\mv{x}}(t)^\top\big(\mm{A}^\top \mm{P}  + \mm{P} \mm{A} \big)\widehat{\mv{x}}(t) + y(t) \mv{l}^\top \mm{P} \widehat{\mv{x}}(t)  + \widehat{\mv{x}}(t)^\top \mm{P}\mv{l} y(t)  \Big] \notag \\
	  &  = & \widehat{\omega}(t) \Big[ \widehat{\mv{x}}(t)^\top\big(\mm{A}^\top \mm{P}  + \mm{P} \mm{A} \big)\widehat{\mv{x}}(t) + 2 \widehat{\mv{x}}(t)^\top \mm{P}\mv{l} y(t)  \Big] \notag \\
	  & \stackrel{\eqref{eq:Lyapunov identity}}{=} &  \widehat{\omega}(t) \Big[- \widehat{\mv{x}}(t)^\top\mm{Q} \widehat{\mv{x}}(t) + 2 \widehat{\mv{x}}(t)^\top \mm{P}\mv{l} y(t)  \Big] \notag \\
	  & \stackrel{\eqref{eq:Inequalities for Lyapunov function}}{\leq} &   \widehat{\omega}(t) \Big[- \lambda_{\min}(\mm{Q}) \norm{\widehat{\mv{x}}(t)}^2 + 2 \underbrace{\norm{\widehat{\mv{x}}(t)}}_{=:a} \underbrace{\norm{\mm{P}}\norm{\mv{l}} \esnorm{y}}_{=:b}  \Big] \notag \\
	  & \stackrel{\eqref{eq:2ab<=...}}{\leq} &   \widehat{\omega}(t) \Big[- \big(\underbrace{\lambda_{\min}(\mm{Q}) - \tfrac{1}{m}}_{\exists m \geq 1 \text{ s.t. } (\cdot) \geq \epsilon_m > 0}\big) \norm{\widehat{\mv{x}}(t)}^2 + \underbrace{m\norm{\mm{P}}^2\norm{\mv{l}}^2 \esnorm{y}^2}_{=:c_m < \infty}  \Big] \notag \\
	  & \stackrel{\eqref{eq:Inequalities for Lyapunov function}}{\leq} &  \Big[- \tfrac{\epsilon_m \epsilon_{\omega}}{\lambda_{\max}(\mm{P})} V\big(\widehat{\mv{x}}(t)\big) + c_m \esnorm{\widehat{\omega}}  \Big] \notag \\
	  \Longrightarrow V\big(\widehat{\mv{x}}(t)\big) & \leq & V\big(\widehat{\mv{x}}(0)\big) + c_m \esnorm{\widehat{\omega}}\tfrac{\lambda_{\max}(\mm{P})}{\epsilon_m \epsilon_{\omega}},
\label{eq:Lyapunov analysis}
\end{eqnarray}

where,  in the last step, the Bellman-Gronwall Lemma~\cite[p.~102f.]{2002_Khalil_NonlinearSystems} was used in its differential form (see Lemma~5.50 and Example~5.51 in~\cite{2017_Hackl_Non-identifierbasedadaptivecontrolinmechatronics:TheoryandApplication}). Hence, in view of~\eqref{eq:Inequalities for Lyapunov function} and~\eqref{eq:Lyapunov analysis}, and with $\mv{c}$ as in~\eqref{eq:state space dynamics of overall IM - appendix}, one can conclude that 

\begin{multline*}
 \forall\,t\geq0 \colon \norm{\widehat{\mv{x}}(t)} \stackrel{\eqref{eq:Inequalities for Lyapunov function},\eqref{eq:Lyapunov analysis}}{\leq} \sqrt{\tfrac{1}{\lambda_{\min}(\mm{P})}\Big( V\big(\widehat{\mv{x}}(0)\big) + c_m \esnorm{\widehat{\omega}}\tfrac{\lambda_{\max}(\mm{P})}{\epsilon_m \epsilon_{\omega}}\Big)} =: c_x < \infty \\
 \text{ and } \qquad |\widehat{y}(t)| \stackrel{\eqref{eq:observer_sogi - appendix}}{\leq} \norm{\mv{c}}\norm{\widehat{\mv{x}}(t)} \leq \norm{\mv{c}}c_x =: c_y < \infty,
\end{multline*}

which completes the proof. 
\end{proof}

\subsection{Boundedness and exponential decay of the signal estimation error}
It is shown that, for piecewise continuous (sinusoidal) and bounded input signals $y(\cdot) \in \Czero^{\textrm{pw}}(\Rzpos;\R) \cap \Linf(\Rzpos;\R)$, the estimation error of the parallelized mSOGIs (or sSOGIs) is bounded. Additionally, if the piecewise constant fundamental angular frequency $\omega(\cdot)$ is correctly estimated, the estimation error decays exponentially.

To present the result, an important observation must be introduced. Note that, for $\omega(\cdot) \in \Czero^{\textrm{pw}} \cap \Linf(\Rzp;[\epsilon_{\omega},\infty))$, any piecewise continuous (sinusoidal-like) signal of the form $y(\cdot) = \sum_{\nu \in \mathbb{H}_n} a_{\nu} \cos\big( \nu \int_0^{\cdot}\omega(\tau) \dx{\tau} + \phi_{0,\nu}\big)$ on any bounded interval $\mathbb{I}_i := [t_i,  t_{i+1})$, $i \in \N_0$ (such that $\Rzpos = \mathbb{I}_0 \cup \mathbb{I}_1 \cup \mathbb{I}_2 \cup \dots $) can be generated by (the output of) a properly initialized internal model~\cite{1976_Francis_TheInternalModelPrincipleofControlTheory} of the following form 

\begin{equation}
\left.
\begin{array}{rcl}
\forall \, t \in \mathbb{I}_i\colon \quad \ddtsmall \mv{x}(t)  &  = &  \omega(t)\, \J \mv{x}(t), \qquad \mv{x}(t_i)=\mv{x}_{i,0} \in \R^{2n} \vspace*{1ex}\\
  y(t)  & = & \mv{c}^\top \mv{x}(t)
\end{array} \qquad 
\right\}
\label{eq:signal generation by IM}
\end{equation}
with $\J$ as in~\eqref{eq:definitions of x, J and Jbar - appendix}.  $\omega(\cdot)$ can be considered as an external input to the internal model. Clearly, for any real (finite) initial value $\mv{x}_{i,0} \in \R^{2n}$ for the $i$-th time interval $\mathbb{I}_i$, all states of the internal model~\eqref{eq:signal generation by IM} are essentially bounded, i.e.~$\mv{x}(\cdot)\in \Linf(\Rzpos;\R^{2n})$. Note that $\omega(\cdot)$, $\phi_{0,\nu}$ and $a_\nu$ for all $\nu \in \mathbb{H}_n$ might change for each interval $\mathbb{I}_i$. Now, the result can be stated.

\begin{theorem}[Boundedness and exponential decay of the signal estimation error]
\label{thm:asymptotic tracking}

Let $\epsilon_{\omega}>0$, $\mathbb{H}_n = \{\nu_1,\nu_2, \dots, \nu_n\}$, $a_{\nu} \geq 0$, $\phi_{0,\nu}$ for all $\nu \in \mathbb{H}_n$ and $\widehat{\omega}(\cdot), \, \omega(\cdot) \in  \Czero^{\textrm{pw}} \cap \Linf(\Rzpos;[\epsilon_{\omega},\infty))$. Consider any piecewise continuous and bounded input signals, i.e.~$y(\cdot) = \sum_{\nu \in \mathbb{H}_n} a_{\nu} \cos\big( \nu \int_0^{\cdot}\omega(\tau)\dx{\tau} + \phi_{0,\nu}\big) \in \Czero^{\textrm{pw}}  \cap \Linf(\Rzpos;\R)$ on any bounded interval $\mathbb{I}_i := [t_i,  t_{i+1})$, $i \in \N_0$ (such that $\Rzpos = \mathbb{I}_0 \cup \mathbb{I}_1 \cup \mathbb{I}_2 \cup \dots $), generated by the internal model~\eqref{eq:signal generation by IM} and assume that $y(\cdot)$ is fed to the parallelized SOGI system~\eqref{eq:observer_sogi - appendix} with $\mm{A}:= \J - \mv{l}\mv{c}^\top$ being a Hurwitz matrix. Then, the estimation error, defined by 
\begin{equation}
 \mv{e}_x(t) := \mv{x}(t) - \widehat{\mv{x}}(t) \in \R^{2n}
 \label{eq:estimation error vector}
\end{equation}
with $\mv{x}(t)$ as in~\eqref{eq:signal generation by IM} and $\widehat{\mv{x}}(t)$ as in~\eqref{eq:observer_sogi - appendix}, is bounded, i.e.~there exists $c_e>0$ such that $\norm{\mv{e}_x(t)}\leq c_e$ for all $t \geq0$. Moreover, if, for some $i\in\N_0$, $\omega(t) = \widehat{\omega}(t)$ for all $t \in \Iss \subseteq \mathbb{I}_i$, then the norm of the estimation error is exponentially decaying, i.e.~there exist constants $c_V, \,\mu_V >0$ such that $\norm{\mv{e}_x(t)} \leq c_V\,\norm{\mv{e}_x(t_i)} \e^{-\mu_V (t -t_i)}$ for all $t \in \Iss$.
\end{theorem}

\begin{proof}
Note that, for any interval $\mathbb{I}_i$, combining~\eqref{eq:observer_sogi - appendix} and~\eqref{eq:signal generation by IM} yields
\begin{equation}
\forall \, t \in \mathbb{I}_i\colon \quad \ddtsmall \underbrace{\begin{pmatrix} \widehat{\mv{x}}(t) \\ \mv{x}(t) \end{pmatrix}}_{\in \R^{4n}}
   =  \begin{bmatrix}
	\widehat{\omega}(t)\mm{A} & \widehat{\omega}(t)\mv{l}\mv{c}^\top \\
	\mm{O}_{2n \times 2n} & \omega(t) \J
    \end{bmatrix}\begin{pmatrix} \widehat{\mv{x}}(t) \\ \mv{x}(t) \end{pmatrix}, \qquad \begin{pmatrix} \widehat{\mv{x}}(t_i) \\ \mv{x}(t_i) \end{pmatrix} =\begin{pmatrix} \widehat{\mv{x}}_{i,0} \\ \mv{x}_{i,0} \end{pmatrix} \in \R^{4n}.
\label{eq:parallelized SOGIs and IM}
\end{equation}
Next, introduce the angular frequency estimation error 
\begin{equation}
 \forall\, t \in \mathbb{I}_i\colon \quad e_{\omega}(t) := \omega(t) - \widehat{\omega}(t) \quad \Longleftrightarrow \quad \omega(t) = e_{\omega}(t) + \widehat{\omega}(t),
\label{eq:angular frequency error}
\end{equation}
and evaluate the time derivative of the estimation error vector as follows
\begin{eqnarray}
\forall \, t \in \mathbb{I}_i\colon \quad \ddtsmall \underbrace{\big(\mv{x}(t) - \widehat{\mv{x}}(t)\big)}_{\stackrel{\eqref{eq:estimation error vector}}{=}\mv{e}_x(t)}
  & = &  \begin{bmatrix}
		-\mm{I}_{2n} & \mm{I}_{2n} \\
   		\end{bmatrix} \ddtsmall\begin{pmatrix} \widehat{\mv{x}}(t) \\ \mv{x}(t) \end{pmatrix} \notag \\
  & \stackrel{\eqref{eq:parallelized SOGIs and IM}}{=} & -\widehat{\omega}(t)\mm{A} \widehat{\mv{x}}(t) - \Big(\widehat{\omega}(t)\mv{l}\mv{c}^\top - \omega(t) \J \Big) \mv{x}(t)  \notag \\
  & \stackrel{\eqref{eq:angular frequency error}}{=} & -\widehat{\omega}(t)\mm{A} \widehat{\mv{x}}(t) - \widehat{\omega}(t)\underbrace{\big(\mv{l}\mv{c}^\top - \J \big)}_{\stackrel{\eqref{eq:observer_sogi - appendix}}{=}-\mm{A}} \mv{x}(t) +   e_{\omega}(t) \J \mv{x}(t) \notag \\
  & = & \widehat{\omega}(t)\mm{A} \mv{e}_x(t) + e_{\omega}(t) \J \mv{x}(t)  \label{eq:dynamics of estimation error}.
\label{eq:dynamics of estimation error with omega in front - alternative} 
\end{eqnarray}
Now, the time derivative of the Lyapunov-like function $V(\mv{e}_x(\cdot))=\mv{e}_x(\cdot)^\top \mm{P}\mv{e}_x(\cdot)$ (with $\mm{P}$ as introduced in \eqref{eq:Lyapunov identity}) is given for all $t \in \mathbb{I}_i = [t_i,\, t_{i+1})$, along the solution of~\eqref{eq:dynamics of estimation error}, as follows
\begin{eqnarray}
\ddtsmall V\big(\mv{e}_x(t)\big) & = & \ddtsmall \mv{e}_x(t)^\top\mm{P} \mv{e}_x(t) +  \mv{e}_x(t)^\top\mm{P} \ddtsmall \mv{e}_x(t) \notag \\
	  & \stackrel{\eqref{eq:dynamics of estimation error}}{=} & \widehat{\omega}(t)  \mv{e}_x(t)^\top\big(\mm{A}^\top \mm{P}  + \mm{P} \mm{A} \big)\mv{e}_x(t) + e_{\omega}(t) \big(  \mv{e}_x(t)^\top  \mm{P} \J \mv{x}(t) + \mv{x}(t)^\top \J^\top \mm{P} \mv{e}_x(t)  \big)   \notag \\
	  & = & \widehat{\omega}(t)  \mv{e}_x(t)^\top\big(\mm{A}^\top \mm{P}  + \mm{P} \mm{A} \big)\mv{e}_x(t) + 2e_{\omega}(t)  \mv{e}_x(t)^\top \mm{P}  \J \mv{x}(t)   \notag \\
	  & \stackrel{\eqref{eq:Lyapunov identity}}{=} &   - \widehat{\omega}(t)\mv{e}_x(t)^\top\mm{Q} \mv{e}_x(t) + 2e_{\omega}(t) \mv{e}_x(t)^\top \mm{P}  \J \mv{x}(t)  \notag \\
	  & \stackrel{\eqref{eq:Inequalities for Lyapunov function}}{\leq} &   - \widehat{\omega}(t) \lambda_{\min}(\mm{Q}) \norm{\mv{e}_x(t)}^2 + 2  \underbrace{\sqrt{\widehat{\omega}(t)}\norm{\mv{e}_x(t)}}_{=:a} \underbrace{|e_{\omega}(t)|\norm{\mm{P}} \norm{\J} \esnorm{\mv{x}}\tfrac{1}{\sqrt{\widehat{\omega}(t)}}}_{=:b}   \notag \\
	  & \stackrel{\eqref{eq:2ab<=...}}{\leq} &    - \widehat{\omega}(t) \big(\underbrace{\lambda_{\min}(\mm{Q}) - \tfrac{1}{m}}_{\exists m \geq 1 \text{ s.t. } (\cdot) \geq \epsilon^\prime_m > 0}\big) \norm{\mv{e}_x(t)}^2 + e_{\omega}(t)^2\underbrace{\tfrac{m}{\epsilon_{\omega}}\norm{\mm{P}}^2\norm{\J}^2 \esnorm{\mv{x}}^2}_{=:c_m^\prime < \infty}  \notag \\
	  & \stackrel{\eqref{eq:Inequalities for Lyapunov function}}{\leq} &  - \underbrace{\tfrac{\epsilon_m^\prime \epsilon_{\omega}}{\lambda_{\max}(\mm{P})}}_{=:\mu_V>0} V\big(\mv{e}_x(t)\big) + e_{\omega}(t)^2 c_m^\prime    \notag \\
	  \Longrightarrow V\big(\mv{e}_x(t)\big) & \leq & V\big(\mv{e}_x(t_i)\big)\e^{-\mu_V(t-t_i)} + c_m^\prime   \int_{t_i}^t e_{\omega}(\tau)^2 \e^{-\mu_V(t-\tau)} \dx{\tau}, 
\label{eq:Lyapunov analysis for estimation error}
\end{eqnarray}
where,  in the last step, the Bellman-Gronwall Lemma in its differential form (see Lemma~5.50 and Example~5.51 in~\cite{2017_Hackl_Non-identifierbasedadaptivecontrolinmechatronics:TheoryandApplication}) was used again. Note that $e_{\omega}(\cdot) \in \Linf(\Rzp;\R)$, since $\widehat{\omega}(\cdot), \, \omega(\cdot) \in \Linf(\Rzp;[\epsilon_{\omega},\infty))$ on each interval $\mathbb{I}_i$. Hence,
\begin{eqnarray}
 \forall\, t \in \mathbb{I}_i\colon \quad \norm{\mv{e}_x(t)}^2 
  & \stackrel{\eqref{eq:Inequalities for Lyapunov function},\eqref{eq:Lyapunov analysis for estimation error}}{\leq} & \tfrac{1}{\lambda_{\min}(\mm{P})} \Big[ V\big(\mv{e}_x(t_i)\big)\e^{-\mu_V(t-t_i)} + c_m^\prime   \int_{t_i}^t e_{\omega}(\tau)^2 \e^{-\mu_V(t-\tau)} \dx{\tau}\Big], \notag \\
 & \stackrel{\eqref{eq:Inequalities for Lyapunov function}}{\leq} & \tfrac{\lambda_{\max}(\mm{P})}{\lambda_{\min}(\mm{P})}\norm{\mv{e}_x(t_i)}\e^{-\mu_V(t-t_i)} + \tfrac{c_m^\prime}{\lambda_{\min}(\mm{P})}   \int_{t_i}^t e_{\omega}(\tau)^2 \e^{-\mu_V(t-\tau)} \dx{\tau},
  \label{eq:error norm exponentially decaying}
\end{eqnarray}
and, clearly, for all $t \in \Iss \subset \mathbb{I}_i$ where $e_{\omega}(t)=0$, the estimation error is exponentially decaying. This completes the proof.
\end{proof}
\begin{remark}[Exponential stability and input-to-state stability]
\label{rem:Exponential stability and input-to-state stability}
Note that, if $e_{\omega}(t)=0$ for all $t \geq t_i$ for some $t_i\geq 0$ and $i\in \N_0$, \eqref{eq:error norm exponentially decaying} gives exponential stability  and, hence, asymptotic estmation (tracking), i.e.~$\lim_{t \to \infty}\norm{\mv{e}_x(t)} = \mv{0}_{2n}$ which implies $\lim_{t \to \infty}|y(t)-\widehat{y}(t)| = 0$. Moreover, note that~\eqref{eq:error norm exponentially decaying} directly implies input-to-state stability~(see e.g.~Part ``Input to State Satbility:~Basic Concepts and Results'' by E.D.~Sontag in~\cite{2008_Agrachev_NonlinearandOptimalControlTheory}). 
\end{remark}

\subsection{Pole placement algorithm for the parallelized mSOGIs}
\label{sec_app_pole_placement}

Before the main results of this section can be presented, a preliminary observation has to be made. Consider the matrix
\begin{equation}
 \forall\, n \in \mathbb{N},\, \forall\, z_1, \ldots, z_n \in \mathbb{C} \colon \SMAT^{-1} := \begin{bmatrix} \Rsub{1}^{-1} & \Rsub{2}^{-1} & \ldots & \Rsub{n}^{-1} \\ \sum\limits_{\substack{i = 1 \\ i \neq 1}}^{n}z_i^2\Rsub{1}^{-1} & \sum\limits_{\substack{i = 1 \\ i \neq 2}}^{n}z_i^2\Rsub{2}^{-1} & \cdots & \sum\limits_{\substack{i = 1 \\ i \neq n}}^{n}z_i^2\Rsub{n}^{-1} \\ \vdots & \vdots & \ddots & \vdots \\ \prod\limits_{\substack{i = 1 \\ i \neq 1}}^{n}z_i^2\Rsub{1}^{-1} & \prod\limits_{\substack{i = 1 \\ i \neq 2}}^{n}z_i^2\Rsub{2}^{-1} & \cdots & \prod\limits_{\substack{i = 1 \\ i \neq n}}^{n}z_i^2\Rsub{n}^{-1} \end{bmatrix}, \quad \Rsub{i}^{-1} = \begin{bmatrix} 1 & 0 \\ 0 & - z_i \end{bmatrix}.
 \label{eq:inverse matrix S^-1}
\end{equation}
Its inverse is given by
\begin{equation}
\SMAT = \begin{bmatrix} \tfrac{z_1^{2\br{n - 1}}}{\prod\limits_{\substack{i = 1 \\ i \neq 1}}^{n}\br{z_1^2 - z_i^2}}\Rsub{1} & - \tfrac{z_1^{2\br{n - 2}}}{\prod\limits_{\substack{i = 1 \\ i \neq 1}}^{n}\br{z_1^2 - z_i^2}}\Rsub{1} & \cdots & \tfrac{\br{- 1}^{n + 1}}{\prod\limits_{\substack{i = 1 \\ i \neq 1}}^{n}\br{z_1^2 - z_i^2}}\Rsub{1} \\ \tfrac{z_2^{2\br{n - 1}}}{\prod\limits_{\substack{i = 1 \\ i \neq 2}}^{n}\br{z_2^2 - z_i^2}}\Rsub{2} & - \tfrac{z_2^{2\br{n - 2}}}{\prod\limits_{\substack{i = 1 \\ i \neq 2}}^{n}\br{z_2^2 - z_i^2}}\Rsub{2} & \cdots & \tfrac{\br{- 1}^{n + 1}}{\prod\limits_{\substack{i = 1 \\ i \neq 2}}^{n}\br{z_2^2 - z_i^2}}\Rsub{2} \\ \vdots & \vdots & \ddots & \vdots \\ \tfrac{z_n^{2\br{n - 1}}}{\prod\limits_{\substack{i = 1 \\ i \neq n}}^{n}\br{z_n^2 - z_i^2}}\Rsub{n} & - \tfrac{z_n^{2\br{n - 2}}}{\prod\limits_{\substack{i = 1 \\ i \neq n}}^{n}\br{z_n^2 - z_i^2}}\Rsub{n} & \cdots & \tfrac{\br{- 1}^{n + 1}}{\prod\limits_{\substack{i = 1 \\ i \neq n}}^{n}\br{z_n^2 - z_i^2}}\Rsub{n} \end{bmatrix},
\label{eq:matrix S fully given}
\end{equation}
since the product of the $c$-th column of $\SMAT^{-1}$ and the $r$-th row of $\SMAT$ yields
\begin{eqnarray*}
&& \tfrac{z_r^{2\br{n - 1}}}{\prod\limits_{\tiny{\substack{j = 1 \\ j \neq r}}}^{n}\br{z_r^2 - z_j^2}}\Rsub{c}\Rsub{r}^{-1} - \sum\limits_{\substack{i = 1 \\ i \neq c}}^{n}z_i^2\tfrac{z_r^{2\br{n - 2}}}{\prod\limits_{\tiny{\substack{j = 1 \\ j \neq r}}}^{n}\br{z_r^2 - z_j^2}}\Rsub{c}\Rsub{r}^{-1} + \ldots + \prod\limits_{\substack{i = 1 \\ i \neq c}}^{n}z_i^2\tfrac{\br{- 1}^{n + 1}}{\prod\limits_{\tiny{\substack{j = 1 \\ j \neq r}}}^{n}\br{z_r^2 - z_j^2}}\Rsub{c}\Rsub{r}^{-1} \\
&=& \br{z_r^{2\br{n - 1}} - z_r^{2\br{n - 2}}\sum\limits_{\substack{i = 1 \\ i \neq c}}^{n}z_i^2 + \ldots + \br{- 1}^{n + 1}\prod\limits_{\substack{i = 1 \\ i \neq c}}^{n}z_i^2}\tfrac{1}{\prod\limits_{\tiny{\substack{j = 1 \\ j \neq r}}}^{n}\br{z_r^2 - z_j^2}}\Rsub{c}\Rsub{r}^{-1} = \begin{cases} \ve{0}_{2\times2}, & c \neq r \\ \ve{I}_{2}, & c = r. \end{cases}
\end{eqnarray*}
Now, the main result can be stated.

\begin{proposition}[Pole placement]
Consider the matrix $\AMAT := \JMAT - \mv{l}\cyvec^{\top}$ with $\JMAT$ as in~\eqref{eq:definitions of x, J and Jbar - appendix} and $\cyvec$ as in~\eqref{eq:state space dynamics of overall IM - appendix}. If and only if the feedback vector $\mv{l}$ is chosen as 
\begin{equation}
	\mv{l} = \SMAT \,\widetilde{\mv{p}}_{\mm{A}}^*, 
\label{eq:feedback gain vector l of parallelized mSOGIs - appendix}
\end{equation}
then the desired characteristic polynomial
\begin{equation}
\chi_{\mm{A}}^*\br{s} := \prod\limits_{i = 1}^{2n}\br{s - p_i^*}
\label{eq:des_char_poly - appendix}
\end{equation}
and characteristic polynomial
\begin{equation}
\chi_{\AMAT}\br{s} = \prod\limits_{i = 1}^{n}\br{s^2 + \nu_i^2} - \sum\limits_{i = 1}^{n}\ggainp{i}\nu_i^2\prod\limits_{\substack{k = 1 \\ k \neq i}}^{n}\br{s^2 + \nu_{k}^2} + s\sum\limits_{i = 1}^{n}\kgainp{i}\nu_i\prod\limits_{\substack{k = 1 \\ k \neq i}}^{n}\br{s^2 + \nu_k^2}
\label{eq:char_poly - appendix}
\end{equation}
have identical coefficients and, hence, $\AMAT = \JMAT - \mv{l}\cyvec^{\top}$ is a Hurwitz matrix with eigenvalues $p_i^* \in \Cneg$, $i \in \bc{1,\ldots,2n}$, as specified in~\eqref{eq:des_char_poly - appendix}. 
\end{proposition}
\begin{proof}
For arbitrary $k_i$ and $g_i$ in the feedback gain vector $\mv{l}$ as in~\eqref{eq:feedback gain vector l of parallelized mSOGIs - appendix}, recall the characteristic polynomial $\chi_{\mm{A}}$ given in \eqref{eq:char_poly - appendix} and collect its coefficients in the following coefficient vector 
\begin{eqnarray}
\cpol \!&=&\! \br{\!\sum\limits_{i = 1}^{n}\kgainp{i}\nu_i, \sum\limits_{i = 1}^{n}\nu_i^2 \!\!-\! \ggainp{i}\nu_i^2, \sum\limits_{i = 1}^{n}\kgainp{i}\nu_i\sum\limits_{\substack{j = 1 \\ j \neq i}}^{n}\nu_j^2, \sum\limits_{i = 1}^{n}\nu_i^2\sum\limits_{j = i + 1}^{n}\nu_j^2 \!\!-\! \ggainp{i}\nu_i^2\sum\limits_{\substack{j = 1 \\ j \neq i}}^{n}\nu_j^2, \ldots, \sum\limits_{i = 1}^{n}\kgainp{i}\nu_i\prod\limits_{\substack{j = 1 \\ j \neq i}}^{n}\nu_j^2, \prod\limits_{i = 1}^{n}\nu_i^2 \!\!-\!\! \sum\limits_{i = 1}^{n}\ggainp{i}\prod\limits_{j = 1}^{n}\nu_j^2\!}^{\!\!\!\top}\!\!\! \notag \\
&=& \br{0, \sum\limits_{i = 1}^{n}\nu_i^2, 0, \sum\limits_{i = 1}^{n}\nu_i^2\sum\limits_{j = i + 1}^{n}\nu_j^2, \ldots, 0, \prod\limits_{i = 1}^{n}\nu_i^2}^{\top} \notag \\
\label{eq:coeff_vector}
&& + \br{\sum\limits_{i = 1}^{n}\kgainp{i}\nu_i, - \sum\limits_{i = 1}^{n}\ggainp{i}\nu_i^2, \sum\limits_{i = 1}^{n}\kgainp{i}\nu_i\sum\limits_{\substack{j = 1 \\ j \neq i}}^{n}\nu_j^2, - \sum\limits_{i = 1}^{n}\ggainp{i}\nu_i^2\sum\limits_{\substack{j = 1 \\ j \neq i}}^{n}\nu_j^2, \ldots, \sum\limits_{i = 1}^{n}\kgainp{i}\nu_i\prod\limits_{\substack{j = 1 \\ j \neq i}}^{n}\nu_j^2, - \sum\limits_{i = 1}^{n}\ggainp{i}\prod\limits_{j = 1}^{n}\nu_j^2}^{\top}.
\end{eqnarray}
A comparison with the desired polynomial in \eqref{eq:des_char_poly - appendix}, having the coefficient vector
\begin{equation*}
\mv{p}_{\mm{A}}^* := \br{- \sum\limits_{i = 1}^{2n}p_i^*,\;\; \sum\limits_{i = 1}^{2n}p_i^*\sum\limits_{j = i + 1}^{2n}p_j^*,\;\; - \sum\limits_{i = 1}^{2n}p_i^*\sum\limits_{j = i + 1}^{2n}p_j^*\sum\limits_{k = j + 1}^{2n}p_k^*,\;\; \ldots,\;\; \prod\limits_{i = 1}^{2n}p_i^*}^{\top},
\end{equation*}
leads to the linear system of equations
\begin{eqnarray}
\overset{\eqref{eq:coeff_vector}}{\Longrightarrow} \mv{p}_{\mm{A}}^* - \begin{pmatrix} 0 \\ \sum\limits_{i = 1}^{n}\nu_i^2 \\ \vdots \\ 0 \\ \prod\limits_{i = 1}^{n}\nu_i^2 \end{pmatrix} 
&=& \underbrace{\begin{bmatrix} \Rsub{1}^{-1} & \Rsub{2}^{-1} & \ldots & \Rsub{n}^{-1} \\ \sum\limits_{\substack{i = 1 \\ i \neq 1}}^{n}\nu_i^2\Rsub{1}^{-1} & \sum\limits_{\substack{i = 1 \\ i \neq 2}}^{n}\nu_i^2\Rsub{2}^{-1} & \cdots & \sum\limits_{\substack{i = 1 \\ i \neq n}}^{n}\nu_i^2\Rsub{n}^{-1} \\ \vdots & \vdots & \ddots & \vdots \\ \prod\limits_{\substack{i = 1 \\ i \neq 1}}^{n}\nu_i^2\Rsub{1}^{-1} & \prod\limits_{\substack{i = 1 \\ i \neq 2}}^{n}\nu_i^2\Rsub{2}^{-1} & \cdots & \prod\limits_{\substack{i = 1 \\ i \neq n}}^{n}\nu_i^2\Rsub{n}^{-1} \end{bmatrix}}_{\stackrel{\eqref{eq:inverse matrix S^-1}}{=}\mm{S}^{-1}}\mv{l}, \quad \text{with } \Rsub{i}^{-1} = \begin{bmatrix} 1 & 0 \\ 0 & - \nu_i \end{bmatrix} .
\end{eqnarray}
Inserting $\mv{l}$ as in~\eqref{eq:feedback gain vector l of parallelized mSOGIs - appendix} and invoking the preliminary result in~\eqref{eq:matrix S fully given}, one indeed obtains $\mv{p}_{\mm{A}}^* = \cpol$. Or in other words, the feedback gain vector to achieve pole placement is given by
\begin{eqnarray}
\overset{\eqref{eq:matrix S fully given}}{\Longrightarrow} \mv{l} &=& \underbrace{\begin{bmatrix} \tfrac{\nu_1^{2\br{n - 1}}}{\prod\limits_{\substack{i = 1 \\ i \neq 1}}^{n}\br{\nu_1^2 - \nu_i^2}}\Rsub{1} & - \tfrac{\nu_1^{2\br{n - 2}}}{\prod\limits_{\substack{i = 1 \\ i \neq 1}}^{n}\br{\nu_1^2 - \nu_i^2}}\Rsub{1} & \cdots & \tfrac{\br{- 1}^{n + 1}}{\prod\limits_{\substack{i = 1 \\ i \neq 1}}^{n}\br{\nu_1^2 - \nu_i^2}}\Rsub{1} \\ \tfrac{\nu_2^{2\br{n - 1}}}{\prod\limits_{\substack{i = 1 \\ i \neq 2}}^{n}\br{\nu_2^2 - \nu_i^2}}\Rsub{2} & - \tfrac{\nu_2^{2\br{n - 2}}}{\prod\limits_{\substack{i = 1 \\ i \neq 2}}^{n}\br{\nu_2^2 - \nu_i^2}}\Rsub{2} & \cdots & \tfrac{\br{- 1}^{n + 1}}{\prod\limits_{\substack{i = 1 \\ i \neq 2}}^{n}\br{\nu_2^2 - \nu_i^2}}\Rsub{2} \\ \vdots & \vdots & \ddots & \vdots \\ \tfrac{\nu_n^{2\br{n - 1}}}{\prod\limits_{\substack{i = 1 \\ i \neq n}}^{n}\br{\nu_n^2 - \nu_i^2}}\Rsub{n} & - \tfrac{\nu_n^{2\br{n - 2}}}{\prod\limits_{\substack{i = 1 \\ i \neq n}}^{n}\br{\nu_n^2 - \nu_i^2}}\Rsub{n} & \cdots & \tfrac{\br{- 1}^{n + 1}}{\prod\limits_{\substack{i = 1 \\ i \neq n}}^{n}\br{\nu_n^2 - \nu_i^2}}\Rsub{n} \end{bmatrix}}_{=:\, \SMAT}\underbrace{\br{\mv{p}_{\mm{A}}^* - \begin{pmatrix} 0 \\ \sum\limits_{i = 1}^{n}\nu_i^2 \\ \vdots \\ 0 \\ \prod\limits_{i = 1}^{n}\nu_i^2 \end{pmatrix}}}_{=:\, \widetilde{\mv{p}}_{\mm{A}}^*}.
\end{eqnarray}
Clearly, if and only if $\mv{p}_{\mm{A}}^* = \cpol$ holds, the eigenvalues of $\AMAT=\JMAT - \mv{l}\cyvec^{\top}$ are given by $p_i^* \in \Cneg$, $i \in \bc{1,\ldots,2n}$, as specified in~\eqref{eq:des_char_poly - appendix} and $\AMAT$ is a Hurwitz matrix. This completes the proof.
\end{proof}

\subsection{Generalization of the adaption law of the mFLL for the parallelized mSOGIs}\label{sec:generalization-of-the-adaption-law-of-the-mfll-for-the-parallelized-msogis}

\subsubsection{Preliminaries}

To ease the understanding of the following derivations, preliminary calculations are introduced. 

First, consider a stable transfer function in the frequency domain, given by
$$
\mathcal{G}\br{\jmath\omega} := \tfrac{n\br{\omega}}{d\br{\omega}} = \tfrac{\Re\br{n\br{\omega}} + \jmath\Im\br{n\br{\omega}}}{\Re\br{d\br{\omega}} + \jmath\Im\br{d\br{\omega}}} = \tfrac{\Re\br{n\br{\omega}}\Re\br{d\br{\omega}} + \Im\br{n\br{\omega}}\Im\br{d\br{\omega}}}{\Re^2\br{d\br{\omega}} + \Im^2\br{d\br{\omega}}} + \jmath\tfrac{\Re\br{d\br{\omega}}\Im\br{n\br{\omega}} - \Re\br{n\br{\omega}}\Im\br{d\br{\omega}}}{\Re^2\br{d\br{\omega}} + \Im^2\br{d\br{\omega}}}.
$$
Its amplitude and phase responses are given by
\begin{equation}
\label{eq:amp_pha_resp}
A_{\mathcal{G}}\br{\omega} = \sqrt{\tfrac{\Re^2\br{n\br{\omega}} + \Im^2\br{n\br{\omega}}}{\Re^2\br{d\br{\omega}} + \Im^2\br{d\br{\omega}}}}
\quad \text{and} \quad 
\Phi_{\mathcal{G}}\br{\omega} = \artan{\tfrac{\Re\br{d\br{\omega}}\Im\br{n\br{\omega}} - \Re\br{n\br{\omega}}\Im\br{d\br{\omega}}}{\Re\br{n\br{\omega}}\Re\br{d\br{\omega}} + \Im\br{n\br{\omega}}\Im\br{d\br{\omega}}}}
\quad \text{, respectively}.
\end{equation}
Moreover, by invoking the following trigonometric identities~\cite[Sect.~4.3--4.4]{1964_Abramowitz_HandbookofMathematicalFunctionsWithFormulasGraphsandMathematicalTables}
\begin{align}
\label{eq:trig_iden}
\left.\begin{array}{rclrcl}
\sine{\alpha \pm \beta} &=& \sine{\alpha}\cosine{\beta} \pm \cosine{\alpha}\sine{\beta}, & \cosine{\alpha \pm \beta} &=& \cosine{\alpha}\cosine{\beta} \mp \sine{\alpha}\sine{\beta}, \\
\sine{\artan{\frac{y}{x}}} &=& \frac{y}{\sqrt{x^2 + y^2}} \qqquad \text{and} & \cosine{\artan{\frac{y}{x}}} &=& \frac{x}{\sqrt{x^2 + y^2}},
\end{array}\right\}
\end{align}
the following expressions are obtained 
\begin{equation}
\label{eq:a_cos_sin_phi}
\left.\begin{array}{rcl}
A_{\mathcal{G}}\br{\omega}\cosine{\Phi_{\mathcal{G}}\br{\omega}} \!\!\!\!\!\!&\overset{\eqref{eq:amp_pha_resp}, \eqref{eq:trig_iden}}{=}&\!\!\!\!\!\! \vspace{0.1cm}\sqrt{\tfrac{\Re^2\br{n\br{\omega}} + \Im^2\br{n\br{\omega}}}{\Re^2\br{d\br{\omega}} + \Im^2\br{d\br{\omega}}}}\tfrac{\Re\br{n\br{\omega}}\Re\br{d\br{\omega}} + \Im\br{n\br{\omega}}\Im\br{d\br{\omega}}}{\sqrt{\br{\Re\br{n\br{\omega}}\Re\br{d\br{\omega}} + \Im\br{n\br{\omega}}\Im\br{d\br{\omega}}}^2 + \br{\Re\br{d\br{\omega}}\Im\br{n\br{\omega}} - \Re\br{n\br{\omega}}\Im\br{d\br{\omega}}}^2}} \\
&=&\!\!\!\!\!\! \vspace{0.1cm}\tfrac{\Re\br{n\br{\omega}}\Re\br{d\br{\omega}} + \Im\br{n\br{\omega}}\Im\br{d\br{\omega}}}{\Re^2\br{d\br{\omega}} + \Im^2\br{d\br{\omega}}}
\quad \text{and} \\
A_{\mathcal{G}}\br{\omega}\sine{\Phi_{\mathcal{G}}\br{\omega}} \!\!\!\!\!\!&\overset{\eqref{eq:amp_pha_resp}, \eqref{eq:trig_iden}}{=}&\!\!\!\!\!\! \vspace{0.1cm}\sqrt{\tfrac{\Re^2\br{n\br{\omega}} + \Im^2\br{n\br{\omega}}}{\Re^2\br{d\br{\omega}} + \Im^2\br{d\br{\omega}}}}\tfrac{\Re\br{d\br{\omega}}\Im\br{n\br{\omega}} - \Re\br{n\br{\omega}}\Im\br{d\br{\omega}}}{\sqrt{\br{\Re\br{n\br{\omega}}\Re\br{d\br{\omega}} + \Im\br{n\br{\omega}}\Im\br{d\br{\omega}}}^2 + \br{\Re\br{d\br{\omega}}\Im\br{n\br{\omega}} - \Re\br{n\br{\omega}}\Im\br{d\br{\omega}}}^2}} \\
&=&\!\!\!\!\!\! \tfrac{\Re\br{d\br{\omega}}\Im\br{n\br{\omega}} - \Re\br{n\br{\omega}}\Im\br{d\br{\omega}}}{\Re^2\br{d\br{\omega}} + \Im^2\br{d\br{\omega}}}.
\end{array}\right\}
\end{equation}

Second, let $t \in \R$ and $\nu, \omega > 0$ and consider the integral given by
\begin{eqnarray*}
	\int\limits_{t}^{t + \tfrac{2\pi}{\nu\omega}}\cosine{\nu\omega\tau + \varphi_1}\cosine{\nu\omega\tau + \varphi_2}\dx\tau &\overset{\eqref{eq:trig_iden}}{=}& \int\limits_{t}^{t + \tfrac{2\pi}{\nu\omega}}\br{\cos^2\br{\nu\omega\tau}\cosine{\varphi_1}\cosine{\varphi_2} + \sin^2\br{\nu\omega\tau}\sine{\varphi_1}\sine{\varphi_2}}\dx\tau \\
	&&- \int\limits_{t}^{t + \tfrac{2\pi}{\nu\omega}}\sine{\nu\omega\tau}\cosine{\nu\omega\tau}\sine{\varphi_1 + \varphi_2}\dx\tau.
\end{eqnarray*}
According to~\cite[p.~163f]{2000_Rade_SpringersMathematischeFormeln}, it follows that
\begin{equation}
\label{eq:integral_cosine_squared}
\int\limits_{t}^{t + \tfrac{2\pi}{\nu\omega}}\cosine{\nu\omega\tau + \varphi_1}\cosine{\nu\omega\tau + \varphi_2}\dx\tau = \tfrac{\pi}{\nu\omega}\cosine{\varphi_1}\cosine{\varphi_2} + \tfrac{\pi}{\nu\omega}\sine{\varphi_1}\sine{\varphi_2} \overset{\eqref{eq:trig_iden}}{=} \tfrac{\pi}{\nu\omega}\cosine{\varphi_1 - \varphi_2}.
\end{equation}

\subsubsection{Steady-state analysis (amplitude and phase responses) of parallelized mSOGIs}
\label{sec_app_amp_pha_resp}
For constant $\widehat{\omega}>0$ and some $i \in \{1,\dots, n\}$, from Figures \ref{fig:whole_model} and \ref{fig:esogi_and_poles}\,(a) the transfer functions for the $i$-th in-phase signal ($\Ytra{i}\br{s}$), the $i$-th quadrature signal ($\Qtra{i}\br{s}$) and the overall estimation error ($\Eytra\br{s}$) of the closed-loop observer system~\eqref{eq:observer_sogi} are obtained as follows 
\begin{eqnarray*}
\Ytra{i}\br{s} &:=& \tfrac{\yhp{i}\br{s}}{y\br{s}} = \tfrac{\br{\kgainp{i}\nu_i\omegah s - \ggainp{i}\nu_i^2\omegah^2}\prod\limits_{\tiny\substack{k = 1 \\ k \neq i}}^{n}\br{s^2 + \nu_k^2\omegah^2}}{\prod\limits_{k = 1}^{n}\br{s^2 + \nu_k^2\omegah^2} - \sum\limits_{k = 1}^{n}\ggainp{k}\nu_k^2\omegah^2\prod\limits_{\tiny\substack{l = 1 \\ l \neq k}}^{n}\br{s^2 + \nu_l^2\omegah^2} + \sum\limits_{k = 1}^{n}\kgainp{k}\nu_k\omegah s\prod\limits_{\tiny\substack{l = 1 \\ l \neq k}}^{n}\br{s^2 + \nu_l^2\omegah^2}} \\
\Qtra{i}\br{s} &:=& \tfrac{\qhp{i}\br{s}}{y\br{s}} = \tfrac{\br{\ggainp{i}\nu_i\omegah s + \kgainp{i}\nu_i^2\omegah^2}\prod\limits_{\tiny\substack{k = 1 \\ k \neq i}}^{n}\br{s^2 + \nu_k^2\omegah^2}}{\prod\limits_{k = 1}^{n}\br{s^2 + \nu_k^2\omegah^2} - \sum\limits_{k = 1}^{n}\ggainp{k}\nu_k^2\omegah^2\prod\limits_{\tiny\substack{l = 1 \\ l \neq k}}^{n}\br{s^2 + \nu_l^2\omegah^2} + \sum\limits_{k = 1}^{n}\kgainp{k}\nu_k\omegah s\prod\limits_{\tiny\substack{l = 1 \\ l \neq k}}^{n}\br{s^2 + \nu_l^2\omegah^2}} \\
\Eytra\br{s} &:=& \tfrac{\ey\br{s}}{y\br{s}} = \tfrac{\prod\limits_{k = 1}^{n}\br{s^2 + \nu_k^2\omegah^2}}{\prod\limits_{k = 1}^{n}\br{s^2 + \nu_k^2\omegah^2} - \sum\limits_{k = 1}^{n}\ggainp{k}\nu_k^2\omegah^2\prod\limits_{\tiny\substack{l = 1 \\ l \neq k}}^{n}\br{s^2 + \nu_l^2\omegah^2} + \sum\limits_{k = 1}^{n}\kgainp{k}\nu_k\omegah s\prod\limits_{\tiny\substack{l = 1 \\ l \neq k}}^{n}\br{s^2 + \nu_l^2\omegah^2}}.
\end{eqnarray*}
By invoking~\eqref{eq:amp_pha_resp} from the preliminaries above and defining
\begin{equation}
\label{eq:abbreviations}
\xi\br{\omega} := \sum\limits_{k = 1}^{n}\kgainp{k}\nu_k\omegah\omega\prod\limits_{\tiny\substack{l = 1 \\ l \neq k}}^{n}\br{\nu_l^2\omegah^2 - \omega^2} \quad \text{and} \quad \zeta\br{\omega} := \prod\limits_{k = 1}^{n}\br{\nu_k^2\omegah^2 - \omega^2} - \sum\limits_{k = 1}^{n}\ggainp{k}\nu_k^2\omegah^2\prod\limits_{\tiny\substack{l = 1 \\ l \neq k}}^{n}\br{\nu_l^2\omegah^2 - \omega^2},
\end{equation}

their respective amplitude and phase responses can be computed as follows
\begin{eqnarray}
\AYtra{i}\br{\omega} &=& \nu_i\omegah\prod\limits_{\tiny\substack{k = 1 \\ k \neq i}}^{n}\br{\nu_k^2\omegah^2 - \omega^2}\sqrt{\tfrac{\kgainp{i}^2\omega^2 + \ggainp{i}^2\nu_i^2\omegah^2}{\br{\prod\limits_{k = 1}^{n}\br{\nu_k^2\omegah^2 - \omega^2} - \sum\limits_{k = 1}^{n}\ggainp{k}\nu_k^2\omegah^2\prod\limits_{\tiny\substack{l = 1 \\ l \neq k}}^{n}\br{\nu_l^2\omegah^2 - \omega^2}}^2 + \br{\sum\limits_{k = 1}^{n}\kgainp{k}\nu_k\omegah\omega\prod\limits_{\tiny\substack{l = 1 \\ l \neq k}}^{n}\br{\nu_l^2\omegah^2 - \omega^2}}^2}} \notag \\
\label{eq:amp_resp_yi}
&\overset{\eqref{eq:abbreviations}}{=}& \nu_i\omegah\prod\limits_{\tiny\substack{k = 1 \\ k \neq i}}^{n}\br{\nu_k^2\omegah^2 - \omega^2}\sqrt{\tfrac{\kgainp{i}^2\omega^2 + \ggainp{i}^2\nu_i^2\omegah^2}{\zeta^2\br{\omega} + \xi^2\br{\omega}}}; \\
\PYtra{i}\br{\omega} &=& \artan{\tfrac{\kgainp{i}\omega\br{\prod\limits_{k = 1}^{n}\br{\nu_k^2\omegah^2 - \omega^2} - \sum\limits_{k = 1}^{n}\ggainp{k}\nu_k^2\omegah^2\prod\limits_{\tiny\substack{l = 1 \\ l \neq k}}^{n}\br{\nu_l^2\omegah^2 - \omega^2}} + \ggainp{i}\nu_i\omegah\sum\limits_{k = 1}^{n}\kgainp{k}\nu_k\omegah\omega\prod\limits_{\tiny\substack{l = 1 \\ l \neq k}}^{n}\br{\nu_l^2\omegah^2 - \omega^2}}{\kgainp{i}\omega\sum\limits_{k = 1}^{n}\kgainp{k}\nu_k\omegah\omega\prod\limits_{\tiny\substack{l = 1 \\ l \neq k}}^{n}\br{\nu_l^2\omegah^2 - \omega^2} - \ggainp{i}\nu_i\omegah\br{\prod\limits_{k = 1}^{n}\br{\nu_k^2\omegah^2 - \omega^2} - \sum\limits_{k = 1}^{n}\ggainp{k}\nu_k^2\omegah^2\prod\limits_{\tiny\substack{l = 1 \\ l \neq k}}^{n}\br{\nu_l^2\omegah^2 - \omega^2}}}} \notag \\
\label{eq:pha_resp_yi}
&\overset{\eqref{eq:abbreviations}}{=}& \artan{\tfrac{\kgainp{i}\omega\zeta\br{\omega} + \ggainp{i}\nu_i\omegah\xi\br{\omega}}{\kgainp{i}\omega\xi\br{\omega} - \ggainp{i}\nu_i\omegah\zeta\br{\omega}}}; \\
\AQtra{i}\br{\omega} &=& \nu_i\omegah\prod\limits_{\tiny\substack{k = 1 \\ k \neq i}}^{n}\br{\nu_k^2\omegah^2 - \omega^2}\sqrt{\tfrac{\ggainp{i}^2\omega^2 + \kgainp{i}^2\nu_i^2\omegah^2}{\br{\prod\limits_{k = 1}^{n}\br{\nu_k^2\omegah^2 - \omega^2} - \sum\limits_{k = 1}^{n}\ggainp{k}\nu_k^2\omegah^2\prod\limits_{\tiny\substack{l = 1 \\ l \neq k}}^{n}\br{\nu_l^2\omegah^2 - \omega^2}}^2 + \br{\sum\limits_{k = 1}^{n}\kgainp{k}\nu_k\omegah\omega\prod\limits_{\tiny\substack{l = 1 \\ l \neq k}}^{n}\br{\nu_l^2\omegah^2 - \omega^2}}^2}} \notag \\
\label{eq:amp_resp_qi}
&\overset{\eqref{eq:abbreviations}}{=}& \nu_i\omegah\prod\limits_{\tiny\substack{k = 1 \\ k \neq i}}^{n}\br{\nu_k^2\omegah^2 - \omega^2}\sqrt{\tfrac{\ggainp{i}^2\omega^2 + \kgainp{i}^2\nu_i^2\omegah^2}{\zeta^2\br{\omega} + \xi^2\br{\omega}}}; \\
\PQtra{i}\br{\omega} &=& \artan{\tfrac{\ggainp{i}\omega\br{\prod\limits_{k = 1}^{n}\br{\nu_k^2\omegah^2 - \omega^2} - \sum\limits_{k = 1}^{n}\ggainp{k}\nu_k^2\omegah^2\prod\limits_{\tiny\substack{l = 1 \\ l \neq k}}^{n}\br{\nu_l^2\omegah^2 - \omega^2}} - \kgainp{i}\nu_i\omegah\sum\limits_{k = 1}^{n}\kgainp{k}\nu_k\omegah\omega\prod\limits_{\tiny\substack{l = 1 \\ l \neq k}}^{n}\br{\nu_l^2\omegah^2 - \omega^2}}{\ggainp{i}\omega\sum\limits_{k = 1}^{n}\kgainp{k}\nu_k\omegah\omega\prod\limits_{\tiny\substack{l = 1 \\ l \neq k}}^{n}\br{\nu_l^2\omegah^2 - \omega^2} + \kgainp{i}\nu_i\omegah\br{\prod\limits_{k = 1}^{n}\br{\nu_k^2\omegah^2 - \omega^2} - \sum\limits_{k = 1}^{n}\ggainp{k}\nu_k^2\omegah^2\prod\limits_{\tiny\substack{l = 1 \\ l \neq k}}^{n}\br{\nu_l^2\omegah^2 - \omega^2}}}} \notag \\
\label{eq:pha_resp_qi}
&\overset{\eqref{eq:abbreviations}}{=}& \artan{\tfrac{\ggainp{i}\omega\zeta\br{\omega} - \kgainp{i}\nu_i\omegah\xi\br{\omega}}{\ggainp{i}\omega\xi\br{\omega} + \kgainp{i}\nu_i\omegah\zeta\br{\omega}}}; \\
\label{eq:amp_resp_ey}
\AEytra\br{\omega} &=& \tfrac{\prod\limits_{k = 1}^{n}\br{\nu_k^2\omegah^2 - \omega^2}}{\sqrt{\br{\prod\limits_{k = 1}^{n}\br{\nu_k^2\omegah^2 - \omega^2} - \sum\limits_{k = 1}^{n}\ggainp{k}\nu_k^2\omegah^2\prod\limits_{\tiny\substack{l = 1 \\ l \neq k}}^{n}\br{\nu_l^2\omegah^2 - \omega^2}}^2 + \br{\sum\limits_{k = 1}^{n}\kgainp{k}\nu_k\omegah\omega\prod\limits_{\tiny\substack{l = 1 \\ l \neq k}}^{n}\br{\nu_l^2\omegah^2 - \omega^2}}^2}} \overset{\eqref{eq:abbreviations}}{=} \tfrac{\prod\limits_{k = 1}^{n}\br{\nu_k^2\omegah^2 - \omega^2}}{\sqrt{\zeta^2\br{\omega} + \xi^2\br{\omega}}} \\
\label{eq:pha_resp_ey}
\PEytra\br{\omega} &=& \artan{\tfrac{- \sum\limits_{k = 1}^{n}\kgainp{k}\nu_k\omegah\omega\prod\limits_{\tiny\substack{l = 1 \\ l \neq k}}^{n}\br{\nu_l^2\omegah^2 - \omega^2}}{\prod\limits_{k = 1}^{n}\br{\nu_k^2\omegah^2 - \omega^2} - \sum\limits_{k = 1}^{n}\ggainp{k}\nu_k^2\omegah^2\prod\limits_{\tiny\substack{l = 1 \\ l \neq k}}^{n}\br{\nu_l^2\omegah^2 - \omega^2}}} \overset{\eqref{eq:abbreviations}}{=} \artan{\tfrac{- \xi\br{\omega}}{\zeta\br{\omega}}}.
\end{eqnarray}
Hence, for an input signal of the form 
$$
y(t) = \sum\limits_{\nu \in \mathbb{H}_n}a_{\nu}\cosine{\nu\omega t + \phi_{\nu,0}},
$$
the estimated in-phase and quadrature signal and the overall error, in \emph{quasi-steady state}, for all $i \in \{1,\dots,n\}$, are given by
\begin{equation}
\left.
\begin{array}{rcl}
\yhp{i}(t) & = & \sum\limits_{\nu \in \mathbb{H}_n}\AYtra{i}\br{\nu\omega}a_{\nu}\cosine{\nu\omega t + \phi_{\nu,0} + \PYtra{i}\br{\nu\omega}}, \\
\qhp{i}(t) & = & \sum\limits_{\nu \in \mathbb{H}_n}\AQtra{i}\br{\nu\omega}a_{\nu}\cosine{\nu\omega t + \phi_{\nu,0} + \PQtra{i}\br{\nu\omega}}, \quad \text{ and } \\
\ey(t) & = &\sum\limits_{\nu \in \mathbb{H}_n}\AEytra\br{\nu\omega}a_{\nu}\cosine{\nu\omega t + \phi_{\nu,0} + \PEytra\br{\nu\omega}},
\end{array}
\right\}
\label{eq:yhat_i(t), qhat_i(t) and e_y(t) in quasi-steady state}
\end{equation}
respectively.

\subsubsection{Sign-correct adaption law}
\label{sec_fll_principle}
Now, the main result of this appendix can be presented. 
\begin{proposition}[Sign-correct adaption over one period]
Let $\omega>0$ and $T_i :=\tfrac{2\pi}{\nu_i\omega}$ for $\nu_i \in \mathbb{H}_n$ and $i \in \{1,\dots,n\}$. Consider system~\eqref{eq:observer_sogi - appendix} with $\omegah > 0$ and introduce the integral 
\begin{equation}
 \int\limits_{t}^{t + T_i}\eyssp{i}\br{\tau}\lamvec^{\top}\xsogissp{i,i}\br{\tau}\dx\tau
 \label{eq:integral}
\end{equation}
where $\eyssp{i}$ is the $i$-th error component and $\xsogissp{i,i}$ is the $i$-th component of the $i$-th state vector in quasi-steady state (indicated by the subscript "$\infty$"). Then, the following holds
\begin{equation}
 \forall\, \lamvec \in \R^{2n} \in \bc{\Bigl.\mv{\alpha}\,\Bigr|\, \lvec^{\top}\underbrace{\blockdiag\br{\mm{O}_{2\times2},\, \ldots,\, \Jbar,\, \ldots,\, \mm{O}_{2\times2}}}_{=: \Jsub{i}}\mv{\alpha} < 0} \colon \int\limits_{t}^{t + T_i}\eyssp{i}\br{\tau}\lamvec^{\top}\xsogissp{i,i}\br{\tau}\dx\tau \begin{cases} > 0, &\omegah < \omega \\ = 0, &\omegah = \omega \\ < 0, &\omegah > \omega. \end{cases}
 \label{eq:sign-correct adaption condition}
\end{equation}
Moreover, if $\lamvec = \Jsub{i}^{-1}\lvec$, then the integral~\eqref{eq:integral} over one period $T_i$ attains its maximal (or minimal, resp.) value and the phases of $\eyssp{i}(t)$ and $\xsogissp{i,i}(t)$ are identical.
\end{proposition}
\begin{proof}
Define $\lamvec := (0,\, 0,\, \ldots,\, \underbrace{\lamyp{i},\, \lamqp{i}}_{=: \lamvecp{i}^{\top} \in \R^{2}},\, \ldots,\, 0,\, 0)^{\top} \in \R^{2n}$ and observe that
\begin{eqnarray}
\label{eq:cond_lamvec}
\lamvecp{i}^{\top}\xsogissp{i,i}(t) & \stackrel{\eqref{eq:yhat_i(t), qhat_i(t) and e_y(t) in quasi-steady state}}{=}& \lamqp{i}\AQtra{i}\br{\nu_i\omega}a_{\nu_i}\cosine{\nu_i\omega t + \PQtra{i}\br{\nu_i\omega}} + \lamyp{i}\AYtra{i}\br{\nu_i\omega}a_{\nu_i}\cosine{\nu_i\omega t + \PYtra{i}\br{\nu_i\omega}}.
\end{eqnarray}
By invoking the trigonometric identities~\cite[p.~125]{2000_Rade_SpringersMathematischeFormeln}
\begin{equation}
\label{eq:sum_sines}
\left.\begin{array}{l}
\sum\limits_{i = 1}^{n}a_i\cosine{\alpha_i} = \sqrt{\sum\limits_{i = 1}^{n}\sum\limits_{i = 1}^{n}a_ia_j \cdot \cosine{\alpha_i - \alpha_j}}\cosine{\alpha_1 + \artan{\tfrac{\sum\limits_{i = 1}^{n}a_i\sine{\alpha_i - \alpha_1}}{\sum\limits_{i = 1}^{n}a_i\cosine{\alpha_i - \alpha_1}}}} \qquad \text{and} \\
\artan{\tfrac{y_1}{x_1}} + \artan{\tfrac{y_2}{x_2}} = \artan{\tfrac{\tfrac{y_1}{x_1} + \tfrac{y_2}{x_2}}{1 - \tfrac{y_1y_2}{x_1x_2}}},
\end{array}\right\}
\end{equation}
it follows
\begin{eqnarray}
\lamvecp{i}^{\top}\xsogissp{i,i}(t) &\overset{\eqref{eq:cond_lamvec},\eqref{eq:sum_sines}}{=}& \sqrt{\lamqp{i}^2\AQtra{i}^2\br{\nu_i\omega}a_{\nu_i}^2 + \lamyp{i}^2\AYtra{i}^2\br{\nu_i\omega}a_{\nu_i}^2 + 2\lamqp{i}\lamyp{i}\AQtra{i}\br{\nu_i\omega}\AYtra{i}\br{\nu_i\omega}a_{\nu_i}^2\cosine{\PQtra{i}\br{\nu_i\omega} - \PYtra{i}\br{\nu_i\omega}}} \notag \\
&&\cdot \cosine{\nu_i\omega t + \PQtra{i}\br{\nu_i\omega} + \artan{\tfrac{\lamyp{i}\AYtra{i}\br{\nu_i\omega}a_{\nu_i}\sine{\PYtra{i}\br{\nu_i\omega} - \PQtra{i}\br{\nu_i\omega}}}{\lamqp{i}\AQtra{i}\br{\nu_i\omega}a_{\nu_i} + \lamyp{i}\AYtra{i}\br{\nu_i\omega}a_{\nu_i}\cosine{\PYtra{i}\br{\nu_i\omega} - \PQtra{i}\br{\nu_i\omega}}}}} \notag \\
&\overset{\substack{\eqref{eq:trig_iden}, \eqref{eq:a_cos_sin_phi} \\ \eqref{eq:amp_resp_yi} - \eqref{eq:pha_resp_qi}}}{=}& 
a_{\nu_i}\tfrac{\nu_i^2\omegah\prod\limits_{\tiny\substack{k = 1 \\ k \neq i}}^{n}\br{\nu_k^2\omegah^2 - \nu_i^2\omega^2}}{\sqrt{\zeta^2\br{\nu_i\omega} + \xi^2\br{\nu_i\omega}}}\sqrt{\lamqp{i}^2\br{\ggainp{i}^2\omega^2 + \kgainp{i}^2\omegah^2} + \lamyp{i}^2\br{\kgainp{i}^2\omega^2 + \ggainp{i}^2\omegah^2} + 2\lamyp{i}\lamqp{i}\kgainp{i}\ggainp{i}\br{\omega^2 - \omegah^2}} \notag \\
&&\cdot\cosine{\nu_i\omega t + \artan{\tfrac{\ggainp{i}\omega\zeta\br{\nu_i\omega} - \kgainp{i}\omegah\xi\br{\nu_i\omega}}{\ggainp{i}\omega\xi\br{\nu_i\omega} + \kgainp{i}\omegah\zeta\br{\nu_i\omega}}} + \artan{\tfrac{\lamyp{i}\tfrac{\br{\kgainp{i}^2 + \ggainp{i}^2}\omegah\omega}{\sqrt{\br{\ggainp{i}^2\omega^2 + \kgainp{i}^2\omegah^2}}}}{\lamqp{i}\sqrt{\ggainp{i}^2\omega^2 + \kgainp{i}^2\omegah^2} + \lamyp{i}\tfrac{\kgainp{i}\ggainp{i}\br{\omega^2 - \omegah^2}}{\sqrt{\br{\ggainp{i}^2\omega^2 + \kgainp{i}^2\omegah^2}}}}}} \notag \\
&\overset{\eqref{eq:sum_sines}}{=}& a_{\nu_i}\tfrac{\nu_i^2\omegah\prod\limits_{\tiny\substack{k = 1 \\ k \neq i}}^{n}\br{\nu_k^2\omegah^2 - \nu_i^2\omega^2}}{\sqrt{\zeta^2\br{\nu_i\omega} + \xi^2\br{\nu_i\omega}}}\sqrt{\br{\lamyp{i}\kgainp{i} + \lamqp{i}\ggainp{i}}^2\omega^2 + \br{\lamyp{i}\ggainp{i} - \lamqp{i}\kgainp{i}}^2\omegah^2} \notag \\
&&\cdot \cosine{\nu_i\omega t + \artan{\tfrac{\tfrac{\ggainp{i}\omega\zeta\br{\nu_i\omega} - \kgainp{i}\omegah\xi\br{\nu_i\omega}}{\ggainp{i}\omega\xi\br{\nu_i\omega} + \kgainp{i}\omegah\zeta\br{\nu_i\omega}} + \tfrac{\lamyp{i}\omegah\omega\br{\kgainp{i}^2 + \ggainp{i}^2}}{\lamqp{i}\br{\ggainp{i}^2\omega^2 + \kgainp{i}^2\omegah^2} + \lamyp{i}\br{\kgainp{i}\ggainp{i}\br{\omega^2 - \omegah^2}}}}{1 - \tfrac{\ggainp{i}\omega\zeta\br{\nu_i\omega} - \kgainp{i}\omegah\xi\br{\nu_i\omega}}{\ggainp{i}\omega\xi\br{\nu_i\omega} + \kgainp{i}\omegah\zeta\br{\nu_i\omega}}\tfrac{\lamyp{i}\omegah\omega\br{\kgainp{i}^2 + \ggainp{i}^2}}{\lamqp{i}\br{\ggainp{i}^2\omega^2 + \kgainp{i}^2\omegah^2} + \lamyp{i}\br{\kgainp{i}\ggainp{i}\br{\omega^2 - \omegah^2}}}}}} \notag \\
&=& a_{\nu_i}\tfrac{\nu_i^2\omegah\prod\limits_{\tiny\substack{k = 1 \\ k \neq i}}^{n}\br{\nu_k^2\omegah^2 - \nu_i^2\omega^2}}{\sqrt{\zeta^2\br{\nu_i\omega} + \xi^2\br{\nu_i\omega}}}\sqrt{\br{\lamyp{i}\kgainp{i} + \lamqp{i}\ggainp{i}}^2\omega^2 + \br{\lamyp{i}\ggainp{i} - \lamqp{i}\kgainp{i}}^2\omegah^2} \notag \\
\label{eq:cond_lamvec_ss_trig_iden}
&&\cdot \cosine{\nu_i\omega t + \artan{\tfrac{\lamqp{i}\br{\ggainp{i}\omega\zeta\br{\nu_i\omega} - \kgainp{i}\omegah\xi\br{\nu_i\omega}} + \lamyp{i}\br{\ggainp{i}\omegah\xi\br{\nu_i\omega} + \kgainp{i}\omega\zeta\br{\nu_i\omega}}}{\lamqp{i}\br{\ggainp{i}\omega\xi\br{\nu_i\omega} + \kgainp{i}\omegah\zeta\br{\nu_i\omega}} + \lamyp{i}\br{\kgainp{i}\omega\xi\br{\nu_i\omega} - \ggainp{i}\omegah\zeta\br{\nu_i\omega}}}}}.
\end{eqnarray}
Now, multiplying $\lamvecp{i}^{\top}\xsogissp{i,i}(t)$ with $\eyssp{i}(t)$ yields
\begin{align}
&\eyssp{i}(t)\lamvecp{i}^{\top}\xsogissp{i,i}(t) \overset{\eqref{eq:amp_resp_ey}, \eqref{eq:pha_resp_ey}, \eqref{eq:cond_lamvec_ss_trig_iden}}{=} a_{\nu_i}^2\tfrac{\nu_i^4\omegah\prod\limits_{\tiny\substack{k = 1 \\ k \neq i}}^{n}\br{\nu_k^2\omegah^2 - \nu_i^2\omega^2}^2\br{\omegah^2 - \omega^2}}{\zeta^2\br{\nu_i\omega} + \xi^2\br{\nu_i\omega}}\sqrt{\br{\lamyp{i}\kgainp{i} + \lamqp{i}\ggainp{i}}^2\omega^2 + \br{\lamyp{i}\ggainp{i} - \lamqp{i}\kgainp{i}}^2\omegah^2} \notag \\
\label{eq:just_before_integration}
&\cosine{\nu_i\omega\tau + \artan{\tfrac{- \xi\br{\nu_i\omega}}{\zeta\br{\nu_i\omega}}}}\cosine{\nu_i\omega\tau + \artan{\tfrac{\lamqp{i}\br{\ggainp{i}\omega\zeta\br{\nu_i\omega} - \kgainp{i}\omegah\xi\br{\nu_i\omega}} + \lamyp{i}\br{\ggainp{i}\omegah\xi\br{\nu_i\omega} + \kgainp{i}\omega\zeta\br{\nu_i\omega}}}{\lamqp{i}\br{\ggainp{i}\omega\xi\br{\nu_i\omega} + \kgainp{i}\omegah\zeta\br{\nu_i\omega}} + \lamyp{i}\br{\kgainp{i}\omega\xi\br{\nu_i\omega} - \ggainp{i}\omegah\zeta\br{\nu_i\omega}}}}}.
\end{align}
Solving the integral~\eqref{eq:integral} over one period $T_i = \tfrac{2\pi}{\nu_i\omega}$ gives
\begin{align}
&\int\limits_{t}^{t + \tfrac{2\pi}{\nu_i\omega}}\eyssp{i}\br{\tau}\lamvecp{i}^{\top}\xsogissp{i,i}\br{\tau}\dx\tau \overset{\eqref{eq:just_before_integration}}{=} a_{\nu_i}^2\tfrac{\nu_i^4\omegah\prod\limits_{\tiny\substack{k = 1 \\ k \neq i}}^{n}\br{\nu_k^2\omegah^2 - \nu_i^2\omega^2}^2\br{\omegah^2 - \omega^2}}{\zeta^2\br{\nu_i\omega} + \xi^2\br{\nu_i\omega}}\sqrt{\br{\lamyp{i}\kgainp{i} + \lamqp{i}\ggainp{i}}^2\omega^2 + \br{\lamyp{i}\ggainp{i} - \lamqp{i}\kgainp{i}}^2\omegah^2} \notag \\
&\int\limits_{t}^{t + \tfrac{2\pi}{\nu_i\omega}}\cosine{\nu_i\omega\tau + \artan{\tfrac{- \xi\br{\nu_i\omega}}{\zeta\br{\nu_i\omega}}}}\cosine{\nu_i\omega\tau + \artan{\tfrac{\lamqp{i}\br{\ggainp{i}\omega\zeta\br{\nu_i\omega} - \kgainp{i}\omegah\xi\br{\nu_i\omega}} + \lamyp{i}\br{\ggainp{i}\omegah\xi\br{\nu_i\omega} + \kgainp{i}\omega\zeta\br{\nu_i\omega}}}{\lamqp{i}\br{\ggainp{i}\omega\xi\br{\nu_i\omega} + \kgainp{i}\omegah\zeta\br{\nu_i\omega}} + \lamyp{i}\br{\kgainp{i}\omega\xi\br{\nu_i\omega} - \ggainp{i}\omegah\zeta\br{\nu_i\omega}}}}}\dx\tau \notag \\
\overset{\eqref{eq:integral_cosine_squared}}{=}& a_{\nu_i}^2\tfrac{\nu_i^4\omegah\prod\limits_{\tiny\substack{k = 1 \\ k \neq i}}^{n}\br{\nu_k^2\omegah^2 - \nu_i^2\omega^2}^2\br{\omegah^2 - \omega^2}}{\zeta^2\br{\nu_i\omega} + \xi^2\br{\nu_i\omega}}\sqrt{\br{\lamyp{i}\ggainp{i} - \lamqp{i}\kgainp{i}}^2\omegah^2 + \br{\lamyp{i}\kgainp{i} + \lamqp{i}\ggainp{i}}^2\omega^2} \notag \\
&\tfrac{\pi}{\nu_i\omega}\cosine{\artan{\tfrac{- \xi\br{\nu_i\omega}}{\zeta\br{\nu_i\omega}}} - \artan{\tfrac{\lamqp{i}\br{\ggainp{i}\omega\zeta\br{\nu_i\omega} - \kgainp{i}\omegah\xi\br{\nu_i\omega}} + \lamyp{i}\br{\ggainp{i}\omegah\xi\br{\nu_i\omega} + \kgainp{i}\omega\zeta\br{\nu_i\omega}}}{\lamqp{i}\br{\ggainp{i}\omega\xi\br{\nu_i\omega} + \kgainp{i}\omegah\zeta\br{\nu_i\omega}} + \lamyp{i}\br{\kgainp{i}\omega\xi\br{\nu_i\omega} - \ggainp{i}\omegah\zeta\br{\nu_i\omega}}}}} \notag \\
\overset{\eqref{eq:sum_sines}}{=}& a_{\nu_i}^2\tfrac{\pi\nu_i^3\omegah\prod\limits_{\tiny\substack{k = 1 \\ k \neq i}}^{n}\br{\nu_k^2\omegah^2 - \nu_i^2\omega^2}^2\br{\omegah^2 - \omega^2}}{\omega\br{\zeta^2\br{\nu_i\omega} + \xi^2\br{\nu_i\omega}}}\sqrt{\br{\lamyp{i}\ggainp{i} - \lamqp{i}\kgainp{i}}^2\omegah^2 + \br{\lamyp{i}\kgainp{i} + \lamqp{i}\ggainp{i}}^2\omega^2}\cosine{\artan{\tfrac{- \br{\lamqp{i}\ggainp{i} + \lamyp{i}\kgainp{i}}\omega}{\br{\lamqp{i}\kgainp{i} - \lamyp{i}\ggainp{i}}\omegah}}} \notag \\
\label{eq:solved_integral}
\overset{\eqref{eq:trig_iden}}{=}& a_{\nu_i}^2\tfrac{\pi\nu_i^3\omegah^2\prod\limits_{\tiny\substack{k = 1 \\ k \neq i}}^{n}\br{\nu_k^2\omegah^2 - \nu_i^2\omega^2}^2\br{\omegah^2 - \omega^2}}{\omega\br{\zeta^2\br{\nu_i\omega} + \xi^2\br{\nu_i\omega}}}\br{\lamqp{i}\kgainp{i} - \lamyp{i}\ggainp{i}} = a_{\nu_i}^2\tfrac{\pi\nu_i^2\omegah^2\prod\limits_{\tiny\substack{k = 1 \\ k \neq i}}^{n}\br{\nu_k^2\omegah^2 - \nu_i^2\omega^2}^2\br{\omega^2 - \omegah^2}}{\omega\br{\zeta^2\br{\nu_i\omega} + \xi^2\br{\nu_i\omega}}}\lvec^{\top}\Jsub{i}\lamvec.
\end{align}
Since $\omega > 0$, observe that only $\omega^2 - \omegah^2$ can change its sign in~\eqref{eq:solved_integral}; all other terms of the nominator and denominator are positive. Hence, only for $\lvec^{\top}\Jsub{i}\lamvec > 0$, the following condition is satisfied
$$
 \forall \mv{\lambda} \in \left\{\left. \mv{\alpha} \in \R^{2n} \right| \mv{l}^\top\mm{J}_i\mv{\lambda} > 0\right\} \colon \int\limits_{t}^{t + T_i}\eyssp{i}\br{\tau}\lamvec^{\top}\xsogissp{i,i}\br{\tau}\dx\tau \begin{cases} > 0, &\omegah < \omega \\ = 0, &\omegah = \omega \\ < 0, &\omegah > \omega, \end{cases};
$$
which proves assertion~\eqref{eq:sign-correct adaption condition}. 
To optimize~\eqref{eq:solved_integral} (to obtain maximal or minimal value), $\lamvec = \Jsub{i}^{- 1}\lvec$ must hold, since
$$
 \lvec^{\top}\Jsub{i}\lamvec = \lvec^{\top}\lvec = \kgainp{i}^2 + \ggainp{i}^2.
$$
Moreover, for $\lamvec = \Jsub{i}^{- 1}\lvec$, the phase is given by
$$
 \Phi_{\lamvec^{\top}\xsogissp{i,i}(t)} \stackrel{\eqref{eq:cond_lamvec_ss_trig_iden}}{=} \artan{\tfrac{\lamqp{i}\br{\ggainp{i}\omega\zeta\br{\nu_i\omega} - \kgainp{i}\omegah\xi\br{\nu_i\omega}} + \lamyp{i}\br{\ggainp{i}\omegah\xi\br{\nu_i\omega} + \kgainp{i}\omega\zeta\br{\nu_i\omega}}}{\lamqp{i}\br{\ggainp{i}\omega\xi\br{\nu_i\omega} + \kgainp{i}\omegah\zeta\br{\nu_i\omega}} + \lamyp{i}\br{\kgainp{i}\omega\xi\br{\nu_i\omega} - \ggainp{i}\omegah\zeta\br{\nu_i\omega}}}} \overset{\lamvec = \Jsub{i}^{- 1}\lvec}{=} \artan{\tfrac{- \xi\br{\nu_i\omega}}{\zeta\br{\nu_i\omega}}},
$$
which is identical to $\PEytra\br{\nu_i\omega}$ in~\eqref{eq:pha_resp_ey}. This completes the proof.
\end{proof}
Concluding, the choice $\lamvec = \Jsub{i}^{- 1}\lvec$ gives the optimal choice for the adaption law~\eqref{eq:mFLL with AWU and rate limiter} of the modified Frequency Locked Loop:~It guarantees sign-correct (phase-correct) and optimal adaption, in the sense that the maximal or minimal value of $\eyssp{i}\br{t}\lamvec^{\top}\xsogissp{i,i}\br{t}$ is fed to the adaption law.

\clearpage
\bibliographystyle{ieeetr}
\bibliography{./MyNewBIB_allinone}
\end{document}